\documentclass[11pt,a4paper]{article}
\usepackage{hyperref}
\usepackage{cite}

\usepackage{amsmath}
\usepackage{amsfonts,amsthm,amssymb}
\numberwithin{equation}{section}

\newcommand{\beqa}{\begin{eqnarray}}
\newcommand{\eeqa}{\end{eqnarray}}
\newcommand{\rf}[1]{(\ref{#1})}

\newtheorem{theorem}{Theorem}[section]
\newtheorem{proposition}{Proposition}[section]
\newtheorem{lemma}{Lemma}[section]

\newtheorem{corollary}{Corollary}[section]
\newtheorem{definition}{Definition}[section]

\newtheorem{property}{Property}[section]
{\theoremstyle{remark}
}
\newtheorem{identity}{Identity}
\newcommand{\End}{\operatorname{End}}

\newcommand{\bra}[1]{\langle\,#1\,|}
\newcommand{\ket}[1]{|\,#1\,\rangle}

\newcommand{\la}{\lambda}

\numberwithin{equation}{section}

\begin{document}

\begin{flushright}
LPENSL-TH-08-18
\end{flushright}

\bigskip \bigskip

\begin{center}
\textbf{\LARGE On quantum separation of variables}

\vspace{50pt}
\end{center}

\begin{center}
{\large \textbf{J.~M.~Maillet}\footnote[1]{{Univ Lyon, Ens de Lyon,
Univ Claude Bernard Lyon 1, CNRS, Laboratoire de Physique, UMR 5672, F-69342
Lyon, France; maillet@ens-lyon.fr}} {\large \textbf{\, and}} ~~ \textbf{G. Niccoli}\footnote[2]{%
{Univ Lyon, Ens de Lyon, Univ Claude Bernard Lyon 1, CNRS,
Laboratoire de Physique, UMR 5672, F-69342 Lyon, France;
giuliano.niccoli@ens-lyon.fr}}}
\end{center}

%\begin{center}
%\vspace{50pt} \today \vspace{50pt}
%\end{center}

\vspace{40pt}

\begin{itemize}
\item[ ] \textbf{Abstract.}\thinspace \thinspace We present a new approach to construct the separate variables basis leading to the full characterization of the transfer matrix spectrum of quantum integrable lattice models. The basis is generated by the repeated action of the transfer matrix itself on a generically chosen state of the Hilbert space. The fusion relations for the transfer matrix, stemming from the Yang-Baxter algebra properties, provide the necessary closure relations to define the action of the transfer matrix on such a basis in terms of elementary local shifts, leading to a separate transfer matrix spectral problem.  Hence our scheme extends to the quantum case a key feature of Liouville-Arnold classical integrability framework where the complete set of conserved charges defines both the level manifold and the flows on it leading to the construction of action-angle variables. We work in the framework of the quantum inverse scattering method. As a first example of our approach, we give the construction of such a basis for models associated to  $Y(gl_n)$ and argue how it extends to their trigonometric and elliptic versions.  Then we show how our general scheme applies concretely to fundamental models associated to the $Y(gl_2)$ and $Y(gl_3)$ $R$-matrices leading to the full characterization of their spectrum. For $Y(gl_2)$ and its trigonometric deformation a particular case of our method reproduces Sklyanin's construction of separate variables. For $Y(gl_3)$ it gives new results, in particular through the proper identification of the shifts acting on the separate basis.  We stress that our method also leads to the full characterization of the spectrum of other known quantum integrable lattice models, including  in particular trigonometric and elliptic spin chains, open chains with general integrable boundaries, and further higher rank cases that we will describe in forthcoming publications. 
\end{itemize}

\begin{center}
{\em Dedicated to the memory of L. D. Faddeev}
\end{center}

\newpage
%\tableofcontents
%\newpage

\section{Introduction}
Despite a huge literature on quantum integrable systems, the very notion of quantum integrability did not reach yet a status that can be put on the same footing as its classical counterpart. In particular there is no quantum analogue of the Liouville-Arnold theorem \cite{Arn76L}. Due to the pioneering works of Sklyanin \cite{Skl85,Skl90,Skl92,Skl92a,Skl95,Skl96}, the situation is much more advanced for the notion of separation of variables (SoV). In the classical case, Hamilton-Jacobi theory and separation of variables, simplify drastically the quadratures one has to solve while following the Liouville-Arnold scheme leading to the construction of action-angle variables. In fact, in most cases, the full resolution of a system can be explicitly achieved only when effective separate variables are known.  It should be noted however that to construct such separate variables the sole knowledge of the complete set of independent commuting (under Poisson brackets) conserved quantities at the basis of the Liouville theorem is not enough in practice; as a matter of fact, to make their construction explicit, one has to rely in general on some additional (algebraic) structures like the Lax matrix, its associated $r$-matrix and the Yang-Baxter algebra they obey, see e.g. \cite{Skl92a,Skl95}. The real breakthrough achieved by Sklyanin was to realize that in paradigmatic examples, and in particular for the integrable models associated to $gl_2$ algebra, the classical and quantum inverse scattering methods \cite{FadS78,FadST79,FadT79,Skl79,Skl79a,FadT81,Skl82,Fad82,Fad96} generically provide the necessary ingredients to construct the separate variables. Moreover this scheme  leads rather straightforwardly to the spectral curve equation and hence, in the quantum case, to the complete spectrum characterization of the transfer matrix and of the associated Hamiltonian. It should be stressed at that point that such a method provides not only the eigenvalues of the transfer matrix and of the Hamiltonian but also the construction of the complete set of associated eigenstates. This is to be compared to other methods also using the quantum inverse scattering framework and the related Yang-Baxter algebra that in general are not easily shown to reach such a level of spectrum completeness. 

Having now the quantum case in mind, the key idea of the Sklyanin approach is to identify the separate variables, say $Y_n$, as the operator zeros of some diagonalizable commuting family of operators having simple common spectrum. Within this approach it is usually given by a distinguished operator in the Yang-Baxter algebra depending on the continuous spectral parameter $\lambda$, let us call it  $B(\lambda)$, commuting for different values of $\lambda$, and such that $B(Y_n) = 0$. Then the use of the Yang-Baxter algebra permits to construct the local (conjugated) shifts acting on the coordinates $Y_n$ as the properly defined evaluation of another distinguished operator of the Yang-Baxter algebra, say $A(\lambda)$, at $\lambda = Y_n$.  It is then possible to  prove, using the Yang-Baxter algebra, that the transfer matrix acts by simple shifts on the $B$-spectrum and to determine at the same time the quantum spectral separate equations (quantum spectral curve) determining the full set of eigenvalues of the transfer matrix and of the associated Hamiltonian. In particular for quantum lattice integrable models the separate basis is identified with the eigenstates basis of $B(\la)$. This method works extremely well in numerous examples mainly associated to $gl_2$, see e.g., \cite{Skl85,Skl90,Skl92,Skl92a,Skl95,Skl96,BabBS96,Smi98a,DerKM01,DerKM03,DerKM03b,BytT06,vonGIPS06,FraSW08,AmiFOW10,NicT10,Nic10a,Nic11,FraGSW11,GroMN12,GroN12,Nic12,Nic13,Nic13a,Nic13b,GroMN14,FalN14,KitMN14,NicT15,NicT16,KitMNT17,MaiNP17,MaiNP18}, in particular in cases where the Algebraic Bethe ansatz fails. It appears however that for higher rank cases some difficulties could arise, see e.g., \cite{Skl96,MarS16}  in particular due to the fact that the identification of the needed operators $B(\la)$ and $A(\la)$ becomes more involved, making the construction of the quantum spectral curve rather non-trivial (see comments for the $gl_3$ case studied in \cite{Skl96} in the appendix A). These issues reveal that the problem of the identification of a pair of such $A(\la)$ and $B(\la)$ operators having all required properties is a cornerstone of this approach, questioning the effective applicability of the method for an arbitrary given integrable system. 

This situation motivated us to look for the construction of separate basis for generic quantum integrable lattice models that would not rely on the determination of such $B(\la)$ and $A(\la)$ operators. Moreover, our wish was to construct  a basis having the built in property that the action of the transfer matrix on it should  be given by simple local shifts, making its wave function and spectral problem separated {\it per se}. Although this idea could look a priori too ambitious, it appears that such a construction is in fact possible for all cases we have been exploring so far, and in particular for models out of reach of the standards SoV or Algebraic Bethe ansatz methods. Moreover, it turns out that it takes a rather simple and universal form as it involves the sole knowledge of the transfer matrix itself and of its fusion properties stemming from the underlying Yang-Baxter algebra. \\

The aim of the present article is to explain this construction and to show how it works concretely in some paradigmatic examples. \\

The main idea is that a separate basis can be obtained by the multiple action of the transfer matrix $T(\lambda)$  itself, evaluated in distinguished points $\xi_n$, on a generically chosen co-vector of the Hilbert space. In most quantum integrable lattice models it may be given by the following set of co-vectors:
\begin{equation}
\bra{h_1,\dots,h_n} = \bra{L} \prod_{i=1}^N T(\xi_i)^{h_i}
\label{sb1}
\end{equation}
where $i = 1, \dots, N$ and $h_i \in \{0, 1, \dots, d_i -1 \}$, the dimension of the Hilbert space being $d = \prod_{i=1}^N d_i$, the $\xi_n$ are distinguished values characterizing the representation of the quantum lattice model (in most cases they will be related to the so-called inhomogeneity parameters) such that all of them are different pairwise (i.e., $\xi_i \neq \xi_j \pm n\eta$ if $i \neq j$ with $\eta$ a characteristic constant and $n$ an integer spanning some model dependent range of relative integers), and $\bra{L}$ is a generically chosen co-vector in $\cal{H}^*$ that obviously should not be an eigenstate of the transfer matrix. In standard Heisenberg spin chains, all $d_i$ are equal to some value $n$ and $d = n^N$, $N$ being the number of lattice sites. Eventually, a slightly more general definition could be necessary:
\begin{equation}
\bra{h_1,\dots,h_n} = \bra{L} \prod_{i=1}^N  \prod_{k_i = 1}^{h_i} T(\xi_i^{(k_i)}) \ ,
\label{sb2}
\end{equation}
where $i = 1, \dots, N$ and $h_i \in \{0, 1, \dots, d_i -1 \}$, the different points $\xi_i^{(k_i)}$, $k_i = 1, \dots,d_i-1$ can be seen as shifted from the first value $\xi_i^{(1)}$ and it is understood that the corresponding factor $T(\xi_i^{(k_i)})$ is absent whenever the corresponding $h_i = 0$. Moreover $\xi_i^{(k_i)} \neq \xi_j^{(k_j)}$ for any choices of $k_i$ and $k_j$ as soon as $i \neq j$.  Obviously \eqref{sb2} reduces to \eqref{sb1} if all points $\xi_j^{(k_j)}$ for any given $j$ are equal and identified to $\xi_j$. One could even imagine some more general formula, the key idea being that the basis is generated by the repeated action of a set of conserved charges of the model at hand on a generically chosen co-vector $\bra{L}$ that should not be, as it is obvious, an eigenstate of the conserved charges $T(\xi_i)$. Let us remark that these formulae are reminiscent of the Frobenius method for generating invariant factors of a matrix, see e.g. \cite{DumFL04}.

The main consequence of the existence of such a basis with (discrete) coordinates $h_{1},...,h_{N}$ is that the wave function $\Psi _{t}(h_{1},...,h_{N})$ in these coordinates of any common eigenvector $\ket{t}$ of the set of conserved charges $T(\xi_i)$ factorizes as a product of $N$ wave functions of one variable, namely:
\begin{equation}
\Psi _{t}(h_{1},...,h_{N})\equiv \langle h_{1},...,h_{N%
}|t\rangle = \langle L  |t\rangle\prod_{i=1}^{N}t(\xi_{i})^{h_{i}}\ ,
\label{precursor-SoV-form}
\end{equation}%
where $t(\xi_i)$ is the eigenvalue of the operator $T(\xi_i)$ associated to the eigenvector $\ket{t}$. In fact this remark is at the origin of the idea of considering \eqref{sb1} as a possible basis of the Hilbert space together with the fact (see bellow) that the transfer matrix evaluated in the points $\xi_i$ seems, as needed, to act naturally by local shifts on it.  At this point, let us give some general comments about such a simple expression for the separate basis:
\begin{itemize}
\item
Except if the dimension of the Hilbert space is one, it is obvious that the set \eqref{sb1} cannot be a basis if the chosen co-vector $\bra{L}$ is  an eigenstate of the transfer matrix. Hence $\bra{L}$ should be a generic state whose orbit under the action of the conserved charges of the system span indeed a basis of $\cal{H}^*$. In particular it should be that $\langle L  |t\rangle \neq 0$ for any non zero transfer matrix eigenstate $|t\rangle$. 
\item
We will show in the next section that if \eqref{sb1} defines a basis it implies that the common spectrum of the set of conserved charges $T(\xi_i)$ is $w$-simple,  i.e. there is only one common eigenvector $\ket{t}$ of the set $T(\xi_i)$ corresponding to a set of given eigenvalues $t(\xi_i)$. It does not mean that $T(\lambda)$ is necessarily diagonalizable as there could be non trivial Jordan blocks. But these different Jordan blocks are all associated to different eigenvalues.
\item 
Although it could be rather astonishing at first sight that the set given \eqref{sb1} defines a basis of the space of states, in most known cases of quantum lattice integrable models it can be proven rather easily that it is indeed the case for generically chosen left state $\bra{L}$. Exceptions concern in fact some peculiar situations where the transfer matrix does not span a complete set of conserved charges, like in the periodic $XXZ$ Heisenberg model for which the third component of the spin $S_z$ is a conserved quantity that is not generated by the transfer matrix; other exceptions are some special choices of the co-vector $\bra{L}$, like an eigenstate of the transfer matrix or a specially constructed co-vector such that the orbit generated by the transfer matrix action stays in a subspace of $\cal{H}^*$ of strictly positive co-dimension. This property  relies in fact on rather mild properties of the representation of the quantum space of states carried out by the quantum Lax operator. In section 2 we will give an elementary proof of this fact in the example of fundamental $gl_n$ based models. 
\item 
The main advantage of a basis such as  \eqref{sb1} is that the action of the transfer matrix on it is obviously given by elementary shifts  as soon as the transfer matrix $T(\lambda)$  can be reconstructed by means of some interpolation formula in terms of its value in the points $\xi_i$ (or $\xi_j^{(k_j)}$) eventually supplemented, by the knowledge of some central element  describing the asymptotic behavior of the transfer matrix. This is in particular the case as soon as the transfer matrix is a polynomial (or a trigonometric or elliptic polynomial) in $\lambda$. In the most simple rational case this is realized if $T(\lambda)$ is a polynomial of degree $N$ in $\lambda$, its asymptotic behavior being given by some central element $T_{\infty}$ such that
\begin{equation}
T(\lambda) = T_{\infty} \prod_{i = 1}^N (\lambda - \xi_i) + \sum_{i = 1}^N T(\xi_i) \prod_{j \neq i, j=1}^N \frac{\lambda - \xi_j}{\xi_i - \xi_j} \ ,
\label{tml1}
\end{equation}
where by hypothesis in \eqref{sb1} all $\xi_j$ are different pairwise. Then for most elements of the set \eqref{sb1} the action of $T(\lambda)$ is given by elementary shifts within the same set. However, it appears immediately that one has to take care of what happens at the boundaries of the set \eqref{sb1}, namely typically when acting with $T(\xi_j)$ on a co-vector of the set \eqref{sb1} for which $h_j = d_j - 1$ already. This is precisely the place where the information coming from the Yang-Baxter algebra about the transfer matrix fusion properties enters by providing the necessary closure relations enabling us to compute this action back in terms of the co-vectors of the set \eqref{sb1}. In other words, we need to know about the structure constants of the associative and commutative algebra of conserved charges. In sections 3, 4 and 5 we will show explicitly how this works for $gl_2$ and $gl_3$ based models. In particular for $gl_2$ case we will show that a particular choice of the co-vector $\bra{L}$ just reproduces Sklyanin's separate basis. For $gl_3$ it leads to a new separate basis and to the full characterization of the spectrum (eigenvalues and eigenvectors). In particular we will show how these closure relations combined with \eqref{sb1} lead to the determination of the quantum spectral curve for these cases. 
\item
Separate basis construction  \eqref{sb1} is very reminiscent of a key property of the classical Liouville-Arnold theorem for classical integrable systems. Indeed in the classical case the complete set of conserved charges defines a level manifold, and also the tangent vectors associated with any point on it. These tangent vectors might be used to define flows going from a given point to another point on this level manifold.. In the definition \eqref{sb1} the separate basis is indeed generated by the (here discrete) flows of the conserved quantities $T(\xi_i)$. The construction \eqref{sb1} shed some new light on the classical case itself that will be considered in a separate article. 
\item
The separate basis for the transfer matrix spectral problem is generated by the transfer matrix itself, i.e., from the sole knowledge of a complete set of conserved quantities, with the additional necessary input of the closure relations stemming from the Yang-Baxter algebra and $R$-matrix representations. Those additional information determine in fact the transfer matrix spectrum. The construction \eqref{sb1} opens the way for a new definition of quantum integrability and of completeness of a given set of conserved charges: it corresponds to cases where the set \eqref{sb1} forms a basis of the Hilbert space.
\end{itemize}

In this article our aim is to give the general principles of our method and to show how it works concretely for some simple interesting models such as the quasi-periodic $XXX$ and $XXZ$ spin-1/2 chains associated to the 6-vertex $R$-matrix and then the quasi-periodic model associated to the fundamental representation of the $Y(gl_3)$ $R$-matrix. 

We would like to stress that we have already developed the same SoV program, going from the construction of the SoV basis up to the characterization of the transfer matrix spectrum as solutions to quantum spectral curves (functional equations of difference type), for some other important classes of integrable quantum models. These are the models associated to fundamental representations of the Yang-Baxter and reflection algebra for the Yangian $Y(gl_{n})$, the quantum group $U_{q}(gl_{n})$ and the t-J model. In order to show how our new SoV method works for non-fundamental models, we have applied it also to the models associated to cyclic and higher spin representations. All these new results will be soon presented in forthcoming articles. 

We are also confident that our approach can be applicable for larger classes of integrable quantum models in the framework of the quantum inverse scattering framework. This is certainly the case for models like the Izergin-Korepin model  and the Hubbard model, for which we have already implemented our basis construction. Models associated to other representations of the Yang-Baxter algebra, like non-compact or infinite dimensional ones, can be also considered using the concepts and ideas developed in the present article. These more general situations are currently under study.

This article is organized as follows. In section 2 we give the general properties of basis of the Hilbert space given by sets \eqref{sb1} or \eqref{sb2}. We also investigate on general ground the implications of \eqref{sb1} being an Hilbert space basis for the properties of the transfer matrix spectrum. Then considering the example of  fundamental model associated to an $Y(gl_n)$ rational $R$-matrix we show that the set \eqref{sb1} indeed determines a basis of $\cal{H}^*$. We also give simple arguments showing that our proof can be extended straightforwardly to the trigonometric and elliptic cases. In section 3 we consider in detail the $Y(gl_2)$ based models and make contact with the Sklyanin's construction of the separate basis in this case. In section 4 we show that these features extend to the trigonometric case, so providing new SoV complete characterization of the transfer matrix spectrum. Then in section 5 we apply our method to the  $Y(gl_3)$ case. There we give the proper identification of the shifts acting on the separate basis and show how to determine the quantum spectral curve. The full characterization of the spectrum is also given. In the last section we give some conclusions and perspectives. In the appendix, we discuss similarities and differences with respect to Sklyanin approach for the  $Y(gl_3)$ case.
%%%%%%%%%%%%%%%%%%%%%%%%%%%%%%%%%%%
%%%%%%%%%%%%%%%%%%%%%%%%%%%%%%%%%%%%%

\section{The quantum SoV basis from a complete set of commuting charges}
Before going to the proof that the proposal \eqref{sb1} provides an SoV basis for models of interest, we need first to show that it indeed defines a basis of the space of states for such models. We will first prove this for fundamental models associated to the $Y(gl_n)$ rational $R$-matrix. Then we will give arguments showing that \eqref{sb1} defines also a basis of the Hilbert space for fundamental  trigonometric models. The second purpose of this section is  to describe the main consequences of  \eqref{sb1} being a basis of the space of states with regards to the properties of the transfer matrix spectrum. In particular we will show that as soon as \eqref{sb1} is a basis of the space of states, the common spectrum of the charges $T(\xi_i)$ is $w$-simple. Here by w-simplicity we mean that for a given eigenvalue there exists only one eigenvector (up to trivial scalar multiplication). %It doesn't mean in general that the corresponding transfer matrix $T(\lambda)$ is diagonalizable as there could in principle exist Jordan blocks.
However for $Y(gl_n)$ based fundamental models with quasi-periodic boundary conditions described by some matrix $K$ we will also show that the corresponding transfer matrix is diagonalizable with simple spectrum as soon as $K$ is diagonalizable with simple spectrum. Let us start by recalling the basic definitions and properties of a separate basis for quantum integrable models that we will use in this article and by reviewing the key features of the Sklyanin approach of this problem in the framework of the quantum inverse scattering method.

\subsection{Quantum separation of variables}
Here we introduce a definition of quantum separation of variables, directly
in the framework of integrable quantum models on a finite dimensional
quantum space ${\mathcal{H}}$. Let us consider a quantum system with Hamiltonian $H \in \End ({\cal{H}})$
exhibiting a one parameter family of commuting conserved charge operators  $%
T(\lambda ) \in \End ({\cal{H}})$, hence having the two properties:
\begin{align}
i)&\,[T(\lambda ),T(\lambda ^{\prime })]{=0}\text{ \ \ }\forall \lambda
,\lambda ^{\prime }\in \mathbb{C},\text{ \ \ \ } \notag \\
ii)&\text{ \ }[T(\lambda ),H]{%
=0}\text{ \ }\forall \lambda \in \mathbb{C},  
\label{Q-integrable-1}
\end{align}%
where in fact one asks that the Hamiltonian $H$ can be generated by $T(\lambda )$. Here for simplicity we consider conserved charge operators parametrized by one spectral parameter $\lambda$. However, the framework we develop in this article can be easily extended to more general cases where the spectral parameter is not just a complex number. For quantum integrable lattice models, $T(\lambda )$ will be given in genral by the transfer matrix. Moreover let us consider the case of a finite dimensional Hilbert space ${\mathcal{H}}$ realized as tensor product of $N$ local Hilbert spaces ${\mathcal{H}}_n$ associated to each lattice site $n$, namely, ${\mathcal{H}} = \otimes_{n = 1}^N {\mathcal{H}}_n$, $N$ being the number of lattice sites. We denote by $dim ({\mathcal{H}}_n) = d_n$ and $dim{\mathcal{H}} = d$ the finite dimensions of these Hilbert spaces with $d = \prod_{n = 1}^N  d_n$. Let us introduce a covector basis of $\cal{H}^*$ of the
form:%
\begin{equation}
\mathcal{S}_{L}\equiv \{\langle y_{1}^{(h_{1})},...,y{}_{N}^{(h_{%
N})}|\text{ }\forall h_{i}\in \{1,...,d_{i}\},\text{ }i\in \{1,...,%
N\}\text{ with }\prod_{n=1}^{N}d_{n}=d\},
\end{equation}%
where the $y_{i}^{(h_{i})}\in \Sigma _{i},$ some $d_{i}$ dimensional set of
complex number. In this covector basis we define the following set of 
$N$ commuting operators $Y_{n}\in \End(\mathcal{H})$ by:%
\begin{equation}
\langle y_{1}^{(h_{1})},...,y{}_{N}^{(h_{N})}|Y_{n}\equiv
y_{n}^{(h_{n})}\langle y_{1}^{(h_{1})},...,y{}_{N}^{(h_{N%
})}|
\end{equation}%
and the associated $N$ commuting shift operators $\Delta _{n}\in 
\End (\mathcal{H})$:%
\begin{equation}
\langle y_{1}^{(h_{1})},...,y_{n}^{(h_{n})},...,y{}_{N}^{(h_{%
N})}|\Delta _{n}\equiv \langle
y_{1}^{(h_{1})},...,y_{n}^{(h_{n}+1-d_{n}\delta _{h_{n},d_{n}})},...,y{}_{%
N}^{(h_{N})}|,
\end{equation}%
for cyclic type representations and 2$N$ shift operators $\Delta
_{n}^{(\pm )}\in \End(\mathcal{H})$: 
\begin{eqnarray}
\langle y_{1}^{(h_{1})},...,y_{n}^{(h_{n})},...,y{}_{N}^{(h_{%
N})}|\Delta _{n}^{(+)} &\equiv &\left( 1-\delta
_{h_{n},d_{n}}\right) \langle
y_{1}^{(h_{1})},...,y_{n}^{(h_{n}+1-d_{n}\delta _{h_{n},d_{n}})},...,y{}_{%
N}^{(h_{N})}|, \\
\langle y_{1}^{(h_{1})},...,y_{n}^{(h_{n})},...,y{}_{N}^{(h_{%
N})}|\Delta _{n}^{(-)} &\equiv &\left( 1-\delta _{h_{n},1}\right)
\langle y_{1}^{(h_{1})},...,y_{n}^{(h_{n}-1+d_{n}\delta
_{h_{n},1})},...,y{}_{N}^{(h_{N})}|,
\end{eqnarray}%
for highest weight type representations. Moreover, let us denote by $D_{n}$
the coordinate representation of the cyclic\ shift operator $\Delta _{n}$:%
\begin{equation}
D_{n}g(y_{n}^{(h)})\equiv g(y_{n}^{(h+1-d_{n}\delta _{h,d_{n}})}),
\end{equation}%
and by $D_{n}^{\left( \pm \right) }$ the coordinate representations of the
highest weight shift operators $\Delta _{n}^{\left( \pm \right) }$:%
\begin{eqnarray}
D_{n}^{\left( +\right) }g(y_{n}^{(h)}) &\equiv &\left( 1-\delta
_{h,d_{n}}\right) g(y_{n}^{(h+1-d_{n}\delta _{h,d_{n}})}), \\
D_{n}^{\left( -\right) }g(y_{n}^{(h)}) &\equiv &\left( 1-\delta
_{h,1}\right) g(y_{n}^{(h-1+d_{n}\delta _{h,1})}).
\end{eqnarray}%
Then we can rephrase Sklyanin definition as it follows:

\begin{definition} We say that $\mathcal{S}_{L}$ is a
separate variables basis (or equivalently $Y_{n}$ are a system of quantum
separate variables) for the family of commuting conserved charges $T(\lambda
)$ if and only if for any $T$-eigenvalue $t(\lambda )$ and $T$-eigenstate $%
|t\rangle $ we have:%
\begin{equation}
\Psi _{t}(h_{1},...,h_{N})\equiv \langle y_{1}^{(h_{1})},...,y{}_{%
N}^{(h_{N})}|t\rangle =\prod_{n=1}^{N%
}Q_{t}(y_{n}^{(h_{n})}),
\end{equation}
where, for all $n\in \{1,...,N\}$, $t(\lambda )$ and $Q_{t}(\lambda
)$ are solutions of separate equations in the spectrum of the separate
variables $y_{n}^{(h)}\in \Sigma _{n}$ of the type:%
\begin{equation}
F_{n}(D_{n}, t(y_{n}^{(h)}),y_{n}^{(h)})\, Q_{t}(y_{n}^{(h)})=0,
\label{SoV-general-1}
\end{equation}%
for the cyclic\ type representations or%
\begin{equation}
F_{n}(D_{n}^{\left( +\right) },D_{n}^{\left(
-\right) }, t(y_{n}^{(h)}),y_{n}^{(h)})\, Q_{t}(y_{n}^{(h)})=0,
\label{SoV-general-2}
\end{equation}%
for the highest weight type representations. Note that here the ordering of objects means that the shifts $D_{n}$ or $D_{n}^{\left( +\right) }$ and $D_{n}^{\left(-\right)}$  can act on $Q_{t}(y_{n}^{(h)})$ and also on the $t(y_{n}^{(h)})$. 
\end{definition}

These $N$ quantum separate relations are a natural quantum analogue
of the classical ones in the Hamilton-Jacobi's approach. As already pointed out by Sklyanin, a possible quantum analog definition of 
\textit{degrees of freedom} for a quantum integrable model is just the
number $N$ of quantum separate variables. In the classical
case the separate relations are used to solve the equations of motion, mainly constructing the change  to the action-angles
variables by quadrature in an additive separate form w.r.t. the separate
variables. In the quantum case, instead, these separate relations are
used to solve the spectral problem of the family of commuting conserved
charges $T(\lambda )$ by determining its eigenvalues and eigenfunctions in
the SoV basis. In particular, these separate relations are $N$
systems of discrete difference equations of maximal order $d_{n}$ for all $%
n\in \{1,...,N\}$, on the spectrum of the separate variables to be
solved for $t(\lambda )$ and $Q_{t}(\lambda )$ within a given class of
functions.

In the quantum case as in the classical case one important problem to solve
given an integrable system is to define its separate variables basis. In the
framework of the quantum inverse scattering, Sklyanin has given a procedure
to define the separate variables for integrable quantum models associated to
Yang-Baxter algebra representations for the rank $1$ and $2$ cases, for the higher rank case see e.g. \cite{Smi01}. Let us give a short review of it. The key point of the Sklyanin approach is to
exhibit a couple of commuting operator families, say $B(\lambda )\in \End (%
\mathcal{H})$ and $A(\lambda )\in \End (\mathcal{H})$, written in terms of
the generators of the Yang-Baxter algebra (the monodromy matrix elements),
such that:
\begin{itemize}
\item The commutation relations between $A(\lambda )$ and $B(\lambda )$ imply 
that $A(\lambda )$, in the operator zero $Y_{n}$ of $B(\lambda )$, is
proportional to the shift operator, $D_{n}$ or $D_{n}^{\left( +\right) }$
according to the type of representation. 
\item $A(\lambda )$, $B(\lambda )$ and
the (higher) transfer matrices of the integrable models $T_{i}(\lambda )$ (the
quantum spectral invariants) satisfy a closed difference equation of the generic form:%
\begin{equation}
\sum_{j=0}^{r+1}\prod_{a=0}^{j-1}A(\lambda -a\eta )T_{r+1-j}(\lambda
)=B(\lambda )\Xi (\lambda ),
\label{pre-spectral-curve-1}
\end{equation}%
where $r$ is the rank of the difference equation, $T_{0}(\lambda )=1$, and $\Xi (\lambda )$ is some operator constructed from the associated Yang-Baxter algebra generators, suggesting that a quantum
analog of the spectral curve equations is satisfied in the operator zero $%
Y_{n}$ of $B(\lambda )$.
\end{itemize}

Then the operator zeros $Y_{n}$ generate the quantum separate variables. Note that if one can prove that for a given integrable quantum model the number of separate variables $N$ is constant, this integer define a
natural quantum analog of the number of degree of freedom, as pointed out by
Sklyanin. Then, the separate relations take a form \eqref{SoV-general-1} or \eqref{SoV-general-2} in the spectrum of the
separate variables  under the following conditions:\\

A) $B(\lambda )$ is diagonalizable with simple spectrum. \\

B) $A(\lambda )$ effectively acts as a shift operator on the full $B(\lambda )$-spectrum, namely, it
generates the complete SoV basis starting from one of its vectors:%
\begin{equation}
\langle y_{1}^{(h_{1})},...,y{}_{N}^{(h_{N})}|=\langle
y_{1}^{(1)},...,y{}_{N}^{(1)}|\prod_{a=1}^{N%
}\prod_{k_{a}=1}^{h_{a}-1}A(y_{a}^{(k_{a})}).
\end{equation}%\\

C) $\Xi (\lambda )$ is finite in the spectrum of the separate variables $y_{n}^{(h)}\in \Sigma _{n}$ for all $n\in \{1,...,N\}$.\\

In the rank $1$ case (essentially associated to $gl_2$ and its quantum deformations), for some Yang-Baxter algebra representations the operator families $A(\lambda )$ and $B(\lambda )$ just coincide with two elements of
the monodromy matrix and the quantum SoV basis can be defined as soon as the property
A) is satisfied, i.e. one can prove that B) and C) are satisfied in the
given representation. Some modifications or generalizations of Sklyanin's definitions of the
operator families $A(\lambda )$ and $B(\lambda )$ have eventually to be introduced to
describe more general Yang-Baxter algebra representations, as for the $8$-vertex or for the reflection algebra cases. In all these generalizations however
the existence of the SoV basis is reduced to the property A) for the new $%
B(\lambda )$. Nevertheless, one has to remark that for the rank $1$ case there are
still some integrable quantum models for which the SoV approach does not apply
even by modification of the Sklyanin's definitions. Simple examples are the cases of XXZ spin chains associated to a diagonal twist or to some very special integrable boundary conditions.

In the higher rank case
the situation is even more involved, the expressions given by Sklyanin for
the family $A(\lambda )$ and $B(\lambda )$ being eventually not polynomials in the
monodromy matrix elements. Hence, while the properties A) can be satisfied the properties  B) and C) look non-trivial and need to be proven in the given representations. Until now it is not known to us  if there exists a representation of the Yang-Baxter algebra such
that the properties B) and C) are satisfied within the Sklyanin approach. More in detail, while we can found simple fundamental representations satisfying the property A) for the $gl_3$ case, the properties B) and C) are clearly not satisfied for the family $A(\lambda )$ given by Sklyanin in this 
representation (see appendix A). This is essentially due to the fact that the operator $\Xi (\lambda )$, which for higher rank cases is no longer a polynomial of the entries of the monodromy matrix, containing inverses of these entries,  is not finite over the full spectrum of $B(\lambda)$ in the fundamental representations.  The main consequence is that the family $A(\lambda )$ doesn't realize the shift over the full $B(\lambda)$-spectrum and hence is not able to generate the $B(\lambda)$ eigenvector basis. Moreover, in such a situation, the equation \eqref{pre-spectral-curve-1} cannot lead to the quantum spectral curve equation as its right hand side does not vanish on the full $B(\lambda)$-spectrum.  

These observations provided us with a strong motivation to look for  a different and more universal   definition of the separation of variables basis to overcome these problems so far encountered in particular for the higher rank case.

\subsection{Towards an SoV basis from transfer matrices: general properties}
Although our approach can be formulated in more general terms, for simplicity, in the following, we restrict our analysis to integrable quantum models possessing a one parameter family of commuting conserved charges generated by a transfer matrix $T(\lambda )$, $\lambda \in \mathbb{C}$. Moreover, we will assume that this transfer matrix is a polynomial of finite degree in the variable $\lambda$ or a polynomial of some simple function of $\lambda$ (trigonometric or elliptic cases). Our aim in this section is to explore the consequences of \eqref{sb1} or \eqref{sb2} being a basis have on the transfer matrix spectrum. This will indeed provide us with necessary requirements to ask for the transfer matrix in order to be able to define SoV basis of the type \eqref{sb1} or \eqref{sb2}. For this purpose let us first define the notion of "basis generating" or "independence property" for the family of conserved charges as follows.
\begin{definition}
A one parameter family of commuting conserved charges $T(\lambda )$, $\lambda \in \mathbb{C}$,  acting on the Hilbert space $\cal{H}$ of finite dimension $d$ will be said to be "basis generating" or to have the "independence property"  if it satisfies the following condition: \\

$iii)$  There exist an integer decomposition of the dimension $d$ of $\cal H$ as $\prod_{n=1}^{%
N}d_{n}=d$, a covector $\langle L|$ in $\cal{H}^{*}$  and $N$ sets of complex numbers $%
\{y_{n}^{(1)},...,y_{n}^{(d_{n}-1)}\}$, $n=1, \dots, N$,  such that the set of $d$ covectors $\langle h_{1},...,h_{N}|$ defined by%
\begin{equation}
\langle h_{1},...,h_{N}|\equiv \langle L|\prod_{a=1}^{N%
}\prod_{k_{a}=1}^{h_{a}}T(y_{a}^{(k_{a})})\text{ \ for any }\{h_{1},...,h_{%
N}\}\in \otimes _{n=1}^{N}\{0,...,d_{n}-1\}
\label{precursor-SoV}
\end{equation}%
is a covector basis of ${\cal{H}}^{*}$, with the convention that if $h_a = 0$, then the corresponding product $\prod_{k_{a}=1}^{h_{a}}T(y_{a}^{(k_{a})})$ is absent and set equal to one in the above formula. 
\label{basis-generatingT}
\end{definition}
Note, that the properties i) and ii) in equation $(\ref{Q-integrable-1})$
represent a quantum analog of the classical definition of integrals of
motion in involution. We refer to the property presented in this definition as "independence
condition" as it can be seen as a quantum analog of the independence of the
maximal set of integral of motions in involution for a Liouville's
completely integrable classical system. This last property states that for a classical Hamiltonian 
system with $N$ degrees of freedom the $N$ vector fields
associated to these integral of  motion that define a level manifold are independent and that they
define almost everywhere a tangent basis on this level manifold. 
Let us remark that whenever a covector $\langle L|$ exists and satisfies this property for a given transfer matrix, it is in general not unique. Indeed, as soon as  the operator $T(\lambda )$ is invertible for some value of $\lambda$ then, using the commutativity of the family $T(\lambda )$,  the covector $\langle L| T(\lambda )$ also satisfies the same property. As already noted in the introduction, it is also obvious that the co-vector  $\langle L|$ cannot be an eigencovector of the family $T(\lambda )$ as in that case the set given by \eqref{precursor-SoV} reduces to the one-dimensional vector space generated by $\langle L|$. 
For the next considerations about the transfer matrix spectrum, we give now the detailed definition of weak simplicity (or weakly non-degenerate
spectrum).

\begin{definition} 
We say that an operator $X\in \End(\cal{H})$ is $w$-simple or has a weakly non-degenerate spectrum if and only if for any $X$-eigenvalue $k$ there exists one and only one (up to trivial multiplication by a scalar) $X$-eigenstate $|k\rangle $. For $\cal H$ finite dimensional, a matrix representing this operator in a basis of $\cal H$ is called a nonderogatory matrix. A matrix is nonderogatory if and only if its characteristic polynomial is equal to its minimal polynomial. Going to its Jordan form, it means that each eigenvalue is associated to a unique Jordan block. Hence two different Jordan blocks have different eigenvalues \cite{HorJL85}.
\label{w-simplicity}
\end{definition}

Let us remark that an operator which is diagonalizable with simple spectrum is $w$-simple, however, the $w$-simplicity does not imply that the operator is diagonalizable. Indeed, an operator which has non-trivial Jordan blocks can still be $w$-simple if any two different Jordan blocks have different eigenvalues \cite{HorJL85}. Nonderogatory matrices have a very nice property with regards to their characteristic polynomial. Let $P_X (t) = a_0 + t a_1 + t^2 a_2 + \dots + t^{d-1} a_{d-1} + t^d$ be the characteristic polynomial of a nonderogatory matrix $X$, with $d=dim (\cal H)$. Then the matrix $X$ can be transformed by a similarity transformation into the so-called companion matrix $C$ of its characteristic polynomial, namely, there exists an invertible matrix $V_X$ such that:
\begin{equation}
V_X X V_X ^{-1} = C = \left( 
\begin{array}{ccccc}
0 & 1 & 0 &  \cdots & 0 \\ 
0 & 0 & 1 & \ddots & \vdots\\ 
\vdots & \vdots & \ddots &  \ddots & \vdots \\ 
0 & 0 & \cdots & 0 & 1 \\ 
-a_0 & -a_1 &  \cdots & -a_{d-2} & - a_{d-1}
\end{array}%
\right).
\label{companion}
\end{equation}
Let us consider the canonical covector basis in $\cal H^*$ denoted by $\langle e_j|$ with $j = 1, \dots , d$ with $d=dim (\cal H)$. We have that $\langle e_j| C = \langle e_{j+1}|$ for any $j = 1, \dots, d-1$. We denote by $\langle f_j| = \langle e_j| V_X$ the transformed covector basis. Then we have the following property:
\begin{proposition}
\label{K-basis}
Let us consider any $w$-simple operator $X$ acting on $\cal H$ represented by the matrix $X$ in the canonical basis. Then there exists a covector $\langle S|$ such that the set $\langle S| X^{n-1}$, $n=1, \dots , d$ is a covector basis of $\cal H^*$. Moreover let us suppose that the dimension $d$ of $\cal H$ has the following integer decomposition, $d = \prod_{n=1}^{N}d_{n}$ in terms of $N$  integers $d_n$, then there exist $N$ commuting matrices $X_i$, $i=1, \dots ,N$ such that the set
 \begin{equation}
\bra{h_1,\dots,h_n} = \bra{S} \prod_{i=1}^N X_i^{h_i}
\label{sb1-New}
\end{equation}
where $i = 1, \dots, N$ and $h_i \in \{0, 1, \dots, d_i -1 \}$ is a covector basis of $\cal H^*$.
\end{proposition}
\begin{proof}
With the above notations it is enough to set $\langle S| = \langle f_1|$. Then the set $\langle S| X^{n-1}$, $n=1, \dots , d$  coincide with the set $\langle f_n|$ which is a covector basis by construction. To prove that there exists at least one $N$-tuple of commuting matrices satisfying the second part of the Proposition, we can proceed as it follows. 
%is just a matter of relabeling the integers from $1$ to $d$.
Let us first remark that as soon as $d = \prod_{n=1}^{N}d_{n}$ any multiple index $(h_1, \dots , h_N)$ with $h_i \in \{0, 1, \dots, d_i -1 \}$ is uniquely associated to the integer $n_{ h_1, \dots , h_N} \in \{0, \dots , d-1\}$ defined as:
\begin{equation}
n_{ h_1, \dots , h_N}=\sum_{k=1}^{N} h_k  \delta_k \ ,
\end{equation}
where  $\delta_k = \prod_{n=1}^{k-1} d_n $ (with the convention that $\delta_1 = 1$) is an increasing sequence of integers determined from the dimensions $d_n$. The proof of this statement is elementary as it is sufficient to prove that the map is injective, which is rather straightforward. Then, defining the obviously commuting matrices:
\begin{equation}
 X_j = X^{\delta_j} \ ,
 \label{Kj}
 \end{equation}
 we immediately get that
 \begin{equation}
X^{n_{ h_1, \dots , h_N}}= X^{\sum_{k=1}^{N} h_k \prod_{n=1}^{k-1} d_n} =  \prod_{i=1}^N X_i^{h_i} \ ,
\end{equation}
and it just remain to apply this on the chosen covector $\langle S| = \langle f_1|$ to get the result.
\end{proof}
Let us remark at this point that the basis we just constructed enables one to obtain an explicit expression of the eigenvectors of the $w$-simple matrix $X$ from the sole knowledge of its eigenvalues, namely using the $X$-spectrum characterization from its characteristic polynomial. In fact there is the following lemma.
\begin{lemma}
\label{pre-sov-X-spectrum}Let $X$ be a $w$-simple operator in $\mathcal{H}$,
with the above notations and definitions, the set of covectors $\langle
S|X^{n-1}=\langle f_{n}|$ , $n=1,\dots ,d$, $\langle S|=\langle f_{1}|$, is
a covector basis of $\mathcal{H^{\ast }}$. Then, if $\lambda $ is an
eigenvalue of $X$, the vector $|\Lambda \rangle $ characterized by its
components $\langle f_{n}|\Lambda \rangle =\alpha \lambda ^{n-1}$, with $%
\alpha =\langle S|\Lambda \rangle \neq 0$, is the unique nonzero eigenvector
associated to $\lambda $ up to trivial multiplication by a scalar.
\end{lemma}

\begin{proof}
From the $w$-simple character of the operator $X$ it follows by definition
that for any of its eigenvalue $\lambda $ there exists one and only one
eigenvector $|\Lambda \rangle $, then by the definition of the basis it is
immediate to get the following chain of identities:%
\begin{equation}
\langle f_{n}|\Lambda \rangle =\langle S|X^{n-1}|\Lambda \rangle =\lambda
^{n-1}\langle S|\Lambda \rangle 
\end{equation}%
which being $\langle f_{n}|$ a covector basis implies that it must holds:%
\begin{equation}
\langle S|\Lambda \rangle \neq 0.
\end{equation}
\end{proof}

It is also instructive to present a direct proof of the fact that a vector $|\Lambda \rangle $ having components $\langle f_{n}|\Lambda \rangle = \langle S|\Lambda \rangle \lambda ^{n-1}$  on the covector basis is indeed an eigenvector of $X$ as soon as $\lambda$ is an eigenvalue.  In fact, this allows to present an interesting mechanism similar to the one that will
appear in all the following when applying the SoV method. Namely we want to
prove that for any $n=1,\dots ,d$, it holds $\langle f_{n}|X|\Lambda \rangle =\lambda
\langle f_{n}|\Lambda \rangle $. This is trivial for all values of $n$
except for $n=d$. Indeed, if $n\leq d-1$, $\langle f_{n}|X|\Lambda \rangle
=\langle f_{n+1}|\Lambda \rangle =\lambda ^{n}\langle S|\Lambda \rangle
=\lambda \langle f_{n}|\Lambda \rangle $. For $n=d$ however, the covector $\langle f_{d}|X = \langle S|X^d$  is not a member of the basis. Hence this action has to be decomposed back onto the above basis. To achieve this, one has just to use the fact that the characteristic
polynomial of $X$, $P_{X}(t)$, evaluated on the matrix $X$ vanishes,
namely, $P_{X}(X)=0$. Hence, $X^{d}=-a_{0}-a_{1}X\dots -a_{d-1}X^{d-1}$.
Therefore, one can compute $\langle f_{d}|X|\Lambda \rangle =\langle
S|X^{d}|\Lambda \rangle$ in terms of the known values $\langle
S|X^{n}|\Lambda \rangle =\lambda^{n}\langle S|\Lambda \rangle$ for 
$n=0,\dots ,d-1$. Then, $\lambda $ itself being an eigenvalue, it verifies
also $P_{X}(\lambda )=0$. Hence $\lambda ^{d}=-a_{0}-a_{1}\lambda \dots
-a_{d-1}\lambda ^{d-1}$. As a consequence, we have the following chain of
equalities: 
\begin{align*}
\langle f_{d}|X|\Lambda \rangle & =\langle S|X^{d}|\Lambda \rangle =-\langle
S|\sum_{n=0}^{d-1}a_{n}X^{n}|\Lambda \rangle \ , \\
& =-\langle S|\Lambda \rangle \sum_{n=0}^{d-1}a_{n}\lambda ^{n}=\langle S|\Lambda \rangle \lambda
^{d} =\lambda \langle f_{d}|\Lambda \rangle \ ,
\end{align*}%
hence proving that $|\Lambda \rangle $ is an eigenvector of $X$ with
eigenvalue $\lambda $. 

Let us comment that the proof just uses the fact that the secular
equation given by the vanishing of the characteristic polynomial holds true
both for the operator $X$ and for its eigenvalues. In the SoV framework, the
analogous situation follows from the fact that the fusion identities for the
transfer matrices holds true both for the transfer matrix as an operator and
for its eigenvalues. It makes the above property a precursor of the SoV
mechanism that will be used in the next sections. Let us further notice that
if $X$ is $w$-simple then any matrix commuting with $X$ is a polynomial in $X
$ of at most degree $d-1$.

Let us finally mention another interesting proof of the previous proposition which uses the Jordan block canonical form of $X$ and the computation of the determinant of confluent Vandermonde matrices instead of using as above its companion matrix $C$. 
\begin{proposition}
\label{confluent-vandermonde}
Let $X\in \End (\cal{H})$ and let us denote by $X_{J}$ an upper-triangular Jordan
form of $X$ obtained through a change of basis induced by the invertible matrix $W_{X}$:%
\begin{equation}
X=W_{X}X_{J}W_{X}^{-1}
\end{equation}%
with the following block form:%
\begin{equation}
X_{J}=\left( 
\begin{array}{cccc}
X_{J}^{\left( 1\right) } & 0 & \cdots & 0 \\ 
0 & X_{J}^{\left( 2\right) } & \ddots & 0 \\ 
0 & \ddots & \ddots & 0 \\ 
0 & 0 & \cdots & X_{J}^{\left( M\right) }%
\end{array}%
\right)
\end{equation}%
where any $X_{J}^{\left( a\right) }$ is a $n_{a}\times n_{a}$ upper-triangular Jordan block
with eigenvalue $k_{a}$, where $\sum_{a=1}^{M}n_{a}= d$. Then let us 
denote by $\langle S|$ the generic covector in $\cal H^*$ such that:%
\begin{equation}
\langle S|W_{X}=(x_{1}^{\left( 1\right) },...,x_{n_{1}}^{\left(
1\right) },x_{1}^{\left( 2\right) },...,x_{n_{2}}^{\left( 2\right)
},...,x_{1}^{\left( M\right) },...,x_{n_{M}}^{\left( M\right) })\ .
\end{equation}%
Then, let  $|s_{j} \rangle = W_{X}^{-1} |e_{j} \rangle$ be the canonical vector basis in $\cal H$ after the change of basis induced by $W_{X}$, we have: 
\begin{equation}
\det_{d}||\left( \langle S|X^{i-1} |s_{j} \rangle \right)
_{i,j\in \{1,...,d\}}|| = \prod_{a=1}^{M}\left(
x_{1}^{\left( a\right) }\right) ^{n_{a}}\prod_{1\leq a<b\leq
M}(k_{b}-k_{a})^{n_{a}n_{b}} \ .
\label{c-vandermonde}
\end{equation}%
\end{proposition}
\begin{proof}
The proof uses standard techniques of matrix algebra \cite{HorJL85} and we give just a sketch of it. It consists first in computing the action of the powers of the matrix $X$ for each Jordan block. To compute the determinant one first show that by addition and subtractions of lines, each multiplied by adequate coefficients, the result does not depend on the variables $x_{j}^{\left( a\right)}$ for $j \neq 1$ for any $a$. Then one can perform the computation setting $x_{j}^{\left( a\right)} = \delta_{1,j} x_{1}^{\left( a\right)}$. It is then easy to extract the product of the variables $x_{1}^{\left( a\right)}$ from the determinant and to reduce to the case where for any $a$, $x_{j}^{\left( a \right)} = \delta_{1,j}$. In such a case the matrix $\left( \langle S|X^{i-1} |s_{j} \rangle \right)_{i,j\in \{1,...,d\}}$ is a confluent Vandermonde matrix which has determinant given by the products of all the ordered differences of the Jordan blocks eigenvalues raised to the power given by the product of the dimensions of the Jordan blocks, see e.g. \cite{Vav97}.
\end{proof}
Then we have the following obvious corollary:
\begin{corollary}
Let $X\in \End (\cal{H})$ be $w$-simple, hence its Jordan form is such that its eigenvalues $k_j$ are pairwise distinct. With the notations and definitions of Proposition \ref{confluent-vandermonde}, as soon as we take:%
\begin{equation}
\prod_{a=1}^{M}x_{1}^{\left( a\right) }\neq 0 \ ,
\end{equation}%
the determinant \eqref{c-vandermonde} is non zero, meaning that the set  $\langle S|X^{i-1}$ for $i\in \{1,...,d\}$ is a covector basis of $\cal H^*$.
\label{pre-sov}
\end{corollary}
Now we can state the main result concerning the $w$-simplicity of the transfer matrix $T(\lambda)$:
\begin{proposition}
Let us consider an integrable quantum model  defined on an Hilbert space of states $\cal H$ of finite dimension $d$ that admits the integer decomposition $d = \prod_{n=1}^{N}d_{n}$ in terms of $N$  integers $d_n$ with a transfer matrix operator $T(\lambda)$ satisfying the properties i) and
ii) in $(\ref{Q-integrable-1})$. \\
Then if $T(\lambda)$ satisfies the independence property iii) in Definition \ref{precursor-SoV}, the one
parameter family of commuting conserved charges $T(\lambda )$ is
$w$-simple, i.e. for any $T$-eigenvalue $t(\lambda )$ it exists
one and only one $T$-eigenstate $|t\rangle $, moreover this $T$-eigenstate
is characterized (uniquely up to a normalization) by the following separated 
wave-function:%
\begin{equation}
\Psi _{t}(h_{1},...,h_{N})\equiv \langle h_{1},...,h_{N}%
|t\rangle = \langle L|t\rangle\prod_{a=1}^{N%
}\prod_{k_{a}=1}^{h_{a}}t(y_{a}^{(k_{a})}),
\label{precursor-SoV-form2}
\end{equation}%
in the basis \rf{precursor-SoV}.\\
Conversely, let us assume that the commuting family of conserved charges $T(\lambda )$ satisfying  the properties i), ii) is $w$-simple, then there exists a
rearrangement of the conserved charges, that is there exists a family $\hat{T}(\lambda )$ function of $T(\lambda )$ and conversely $T(\lambda )$ can be reconstructed from $\hat{T}(\lambda )$, and $\hat{T}(\lambda )$ is such that $[\hat{T}%
(\lambda ),T(\mu )]=0$ and $[\hat{T}(\lambda ), \hat{T}(\mu )]=0$ for all $\lambda ,\mu \in \mathbb{C}$, satisfying the
properties i), ii) and iii) for some covector $\langle L|$ in $\cal{H}^{*}$.
\end{proposition}

\begin{proof}
The first part of the Proposition is trivial, namely the knowledge of the eigenvalue $t(\lambda)$ determines completely the components of the associated eigenvector $\ket t$ on the covector basis $\langle h_{1},...,h_{N}|$, hence its unicity. The converse uses the fact that we can fix some value, say $\lambda_0$ such that the transfer matrix $T(\lambda_0) = X$ is $w$-simple. In that case any operator commuting with $X$ is a polynomial of $X$ (property of nonderogatory matrices, see \cite{HorJL85}). Hence the complete knowledge of conserved charges is contained in $X$ and its successive powers. Let us define the operator $\hat{T}(\lambda )$ as the polynomial of degree $N-1$ having the value $X_j$ defined as in \eqref{Kj} in given points $\xi_j$, $j=1, \dots , N$. It is then enough to apply Proposition (\ref{K-basis}) to  obtain the result. Moreover the family $T(\lambda )$ being a polynomial in $X = T(\lambda_0)$ can be reconstructed from the family $\hat{T}(\lambda )$.
\end{proof}
Note that the optimal decomposition of $d$ is given by its prime decomposition but here we consider any integer decomposition of $d$. However, for integrable lattice models of interest that we will consider in the following sections, we will use the prime decomposition of $d$ as it leads to quantum spectral curve equation of minimal degree. 
To conclude these general considerations, the results so far obtained show that as soon as the transfer matrix is $w$-simple there exists a precursor of an SoV basis \rf{precursor-SoV}.
%As we know, for integrable lattice models with generic representations in each lattice site $j$ labeled by an inhomogeneity parameter $\xi_j$, the transfer matrix is generically $w$-simple.
It should be stressed however that to get the full SoV scheme we also need to have a way to determine the quantum spectral curve associated to it. This in fact amounts to unravel the nature of the commutative algebra of conserved charges that enable to get a characterization of the transfer matrix spectrum in terms of secular equations of degrees $d_i$, in general much smaller than the dimension $d$ of the Hilbert space. This is somehow analogous to get effectively the canonical rational form \cite{DumFL04} of the transfer matrix written in terms of companion blocks having characteristic polynomials of smaller degree compared to the characteristic polynomial. This is provided as we shall see in concrete examples in the following sections by the fusion relations satisfied by the tower of fused transfer matrices. These fusion relations are in their turn direct consequences of the Yang-Baxter algebra and hence contain the integrability properties of the model at hand. 

Let us now give some elementary example of the above basis construction for the quasi-periodic $Y(gl_n)$ fundamental models.

\subsection{The example of the quasi-periodic $Y(gl_n)$ fundamental model}

Let us consider the Yangian $gl_{n}$ $R$-matrix%
\begin{equation}
R_{a,b}(\lambda _{a}-\lambda _{b})=(\lambda _{a}-\lambda _{b})I_{a,b}+%
\mathbb{P}_{a,b}\in \End(V_{a}\otimes V_{b}),\text{ \ with }V_{a}=\mathbb{C}%
^{n}\text{, }V_{b}=\mathbb{C}^{n}\text{, }n\in \mathbb{N}^{\ast },
\end{equation}%
where $\mathbb{P}_{a,b}$ is the permutation operator on the tensor product $%
V_{a}\otimes V_{b}$, which is solution of the Yang-Baxter
equation written in $\End(V_{a}\otimes V_{b}\otimes V_{c})$:%
\begin{equation}
R_{a,b}(\lambda _{a}-\lambda _{b})R_{a,c}(\lambda _{a}-\lambda
_{c})R_{b,c}(\lambda _{b}-\lambda _{c})=R_{b,c}(\lambda _{b}-\lambda
_{c})R_{a,c}(\lambda _{a}-\lambda _{c})R_{a,b}(\lambda _{a}-\lambda _{b}),
\end{equation}%
and any matrix $K\in \End(\mathbb{C}^{n})$ is a scalar solution of the
Yang-Baxter equation with respect to it:%
\begin{equation}
R_{a,b}(\lambda _{a},\lambda _{b}|\eta )K_{a}K_{b}=K_{b}K_{a}R_{a,b}(\lambda
_{a},\lambda _{b}|\eta )\in \End(V_{a}\in V_{b}\otimes V_{c}),
\end{equation}%
i.e. it is a symmetries of the considered $R$-matrix. Then we can define the
following monodromy matrix,%
\begin{equation}
M_{a}^{(K)}(\lambda ,\{\xi _{1},...,\xi _{N}\})\equiv K_{a}R_{a,N}(\lambda
_{a}-\xi _{N})\cdots R_{a,1}(\lambda _{a}-\xi _{1}),
\end{equation}%
which satisfies the Yang-Baxter equation,%
\begin{equation}
R_{a,b}(\lambda _{a}-\lambda _{b})M_{a}^{(K)}(\lambda _{a},\{\xi
\})M_{b}^{(K)}(\lambda _{b},\{\xi \})=M_{b}^{(K)}(\lambda _{b},\{\xi
\})M_{a}^{(K)}(\lambda _{b},\{\xi \})R_{a,b}(\lambda _{a}-\lambda _{b})
\end{equation}%
in $\End(V_{a}\otimes V_{b}\otimes \mathcal{H})$, with $\mathcal{H}\equiv
\otimes _{l=1}^{N}V_{l}$ and its dimension $d=n^{N}$. Hence it defines a
representation of the Yang-Baxter algebra associated to this $R$-matrix and
the following one parameter family of commuting transfer matrices:%
\begin{equation}
T^{(K)}(\lambda ,\{\xi \})\equiv tr_{V_{a}}M_{a}^{(K)}(\lambda ,\{\xi \}).
\end{equation}%
In the above formulae, the complex parameters $\{\xi_{1},...,\xi _{N}\}$ are called {\em inhomogeneity parameters}, and we will assume in the following that  they are in generic position such that the above Yang-Baxter algebra representation is irreducible. The following proposition holds:

\begin{proposition}
\label{gln-basis} The set of covectors%
\begin{equation}
\langle h_{1},...,h_{N}|\equiv \langle S|\prod_{a=1}^{N}(T^{(K)}(\xi
_{a},\{\xi \}))^{h_{a}}\text{\ }\forall \{h_{1},...,h_{N}\}\in
\{0,...,n-1\}^{\otimes N},  \label{SoV-basis-0-YB}
\end{equation}%
defines a covector basis of $\mathcal{H}^{\ast }$ for almost any choice of the covector $\langle S|$ and of the inhomogeneity parameters $\{\xi_{1},...,\xi _{N}\}$ under the only condition that the given $K\in \End(\mathbb{C}^{n})$ is $w$-simple on $%
\mathbb{C}^{n}$. In particular, for almost all values of the inhomogeneity parameters $\{\xi
_{1},...,\xi _{N}\}$, the covector in $%
\mathcal{H^{\ast }}$ of tensor product form:%
\begin{equation}
\langle S|\equiv \bigotimes_{a=1}^{N}\langle S,a|,  \label{Tensor-S}
\end{equation}%
can be chosen as soon as we take for $\langle S,a|$ a local covector in $%
V_{a}^{\ast }$ such that%
\begin{equation}
\langle S,a|K_{a}^{h}\text{ with }h\in \{0,...,n-1\},
\end{equation}%
form a covector basis for $V_{a}$ for any $a\in \{1,...,N\}$, the existence
of $\langle S,a|$ being implied by the fact that $K$ is $w$-simple.
\end{proposition}

\begin{proof}
Let us define the $n^{N}\times n^{N}$ matrix $\mathcal{M}\left( \langle
S|,K,\{\xi \}\right) $ with elements:%
\begin{equation}
\mathcal{M}_{i,j}\equiv \langle h_{1}(i),...,h_{N}(i)|e_{j}\rangle ,\text{ \ 
}\forall i,j\in \{1,...,n^{N}\}
\end{equation}%
where we have defined uniquely the $N$-tuple $(h_{1}(i),...,h_{N}(i))\in
\{1,...,n\}^{\otimes N}$ by:%
\begin{equation}
1+\sum_{a=1}^{N}h_{a}(i)n^{a-1}=i\in \{1,...,n^{N}\},
\end{equation}%
and $|e_{j}\rangle \in \mathcal{H}$ is the element $j\in \{1,...,n^{N}\}$ of
the elementary basis in $\mathcal{H}$. Then the condition that the set $(\ref%
{SoV-basis-0-YB})$ form a basis of covector in $\mathcal{H^{\ast }}$ is equivalent
to the condition:%
\begin{equation}
\text{det}_{n^{N}}\mathcal{M}\left( \langle S|,K,\{\xi \}\right) \neq 0.
\label{Lin-indep-SoV}
\end{equation}%
The transfer matrix $T^{(K)}(\lambda ,\{\xi \})$ is a polynomial in the
parameters of the $K$ matrix and in the inhomogeneity parameters $\{\xi
_{1},...,\xi _{N}\}$. Then the determinant $\text{det}_{n^{N}}\mathcal{M}%
\left( \langle S|,K,\{\xi \}\right) $ is itself a polynomial of these
parameters and it is moreover a polynomial in the coefficients $\langle
S|e_{j}\rangle $ of the covector $\langle S|$. Taking into account this polynomial dependence in all the variables, it is enough to prove that the
condition $(\ref{Lin-indep-SoV})$ holds for some special limit on the
parameters to prove that it is true for almost any value of the parameters
except on the zeros of the
corresponding polynomials. The transfer matrix $T^{(K)}(\lambda ,\{\xi \})$
satisfies the following identities:%
\begin{equation}
T^{(K)}(\xi _{l},\{\xi \})=R_{l,l-1}(\xi _{l}-\xi _{l-1})\cdots R_{l,1}(\xi
_{l}-\xi _{1})K_{l}R_{l,N}(\xi _{l}-\xi _{N})\cdots R_{l,l+1}(\xi _{l}-\xi
_{l+1}),
\end{equation}%
then if we chose to impose:%
\begin{equation}
\xi _{a}=a\xi \text{ \ ,}\ \forall a\in \{1,...,N\}
\end{equation}%
it follows that the $T^{(K)}(\xi _{l},\{\xi \})$ are polynomials of degree $%
N-1$ in $\xi $ for all $l\in \{1,...,N\}$:%
\begin{equation}
T^{(K)}(\xi _{l},\{\xi \})=c_{l,N-1}\xi
^{N-1}K_{l}+\sum_{a=0}^{N-2}c_{l,a}\xi ^{a}T_{l,a},\text{ with }%
c_{l,N-1}=(-1)^{N-l}(l-1)!(N-l)!,
\end{equation}%
and so the same is true for the covectors:%
\begin{equation}
\langle h_{1},...,h_{N}|\equiv \xi
^{(N-1)\sum_{a=1}^{N}h_{a}}\prod_{a=1}^{N}c_{a,N-1}^{h_{a}}\langle
S|\prod_{a=1}^{N}K_{a}^{h_{a}}+\sum_{a=0}^{-1+(N-1)\sum_{a=1}^{N}h_{a}}\xi
^{a}\langle h_{1},...,h_{N},a|,
\end{equation}%
and similarly $\text{det}_{n^{N}}\mathcal{M}\left( \langle S|,K,\{\xi
\}\right) $ is a polynomial of degree $(N-1)\sum_{j=1}^{n^{N}}%
\sum_{a=1}^{N}h_{a}(j)$ with maximal degree coefficient given by:%
\begin{equation}
\prod_{j=1}^{n^{N}}\prod_{a=1}^{N}c_{a,N-1}^{h_{a}(j)}\text{det}%
_{n^{N}}||\left( \langle S|\prod_{a=1}^{N}K_{a}^{h_{a}(i)}|e_{j}\rangle
\right) _{i,j\in \{1,...,n^{N}\}}||.
\end{equation}%
If we take $\langle S|$ of the tensor product form $(\ref{Tensor-S})$\ then
it holds, with $|e_{j}\rangle =\otimes _{a}|e_{j}(a)\rangle $:%
\begin{equation}
\text{det}_{n^{N}}||\left( \langle
S|\prod_{a=1}^{N}K_{a}^{h_{a}(i)}|e_{j}\rangle \right) _{i,j\in
\{1,...,n^{N}\}}||=\prod_{a=1}^{N}\text{det}_{n}||\left( \langle
S,a|K_{a}^{i-1}|e_{j}(a)\rangle \right) _{i,j\in \{1,...,n\}}||
\end{equation}%
Now by hypothesis $K$ is $w$-simple and this implies the existence of $%
\langle S,a|$ such that these determinants are all non-zero from the
Proposition \ref{K-basis} or Corollary \ref{pre-sov}. So that we have proven
that the leading coefficient of $\text{det}_{n^{N}}\mathcal{M}\left( \langle
S|,K,\{\xi \}\right) $ is non-zero so it is non-zero for almost any choice
of the parameters.
\end{proof}

We have already proven that the possibility to introduce such type of basis
implies that the transfer matrix spectrum is $w$-simple, we want to show
that in general the transfer matrix is diagonalizable with simple spectrum
provided we impose some further requirements on the twist matrix.

\begin{proposition}
Let us assume that $K\in \End(\mathbb{C}^{n})$ is diagonalizable with simple
spectrum on $\mathbb{C}^{n}$, then, almost for any values of the
inhomogeneities, it holds:%
\begin{equation}
\langle t|t\rangle \neq 0,
\end{equation}%
where $|t\rangle $ and $\langle t|$ are the unique eigenvector and
eigencovector associated to $t(\lambda )$ a generic eigenvalue of $%
T^{(K)}(\lambda ,\{\xi \})$ so that $T^{(K)}(\lambda ,\{\xi \})$ is
diagonalizable with simple spectrum.
\end{proposition}

\begin{proof}
As we have already proven, if we impose:%
\begin{equation}
\xi _{a}=a\xi \text{ \ }\forall a\in \{1,...,N\},
\end{equation}%
then it follows that the $T^{(K)}(\xi _{l},\{\xi \})$ are polynomials of
degree $N-1$ in $\xi $ for all $l\in \{1,...,N\}$: 
\begin{equation}
T^{(K)}(\xi _{l},\{\xi \})=\xi
^{N-1}T_{l,N-1}^{(K)}+\sum_{a=0}^{N-2}c_{l,a}\xi ^{a}T_{l,a}\ ,
\end{equation}%
with%
\begin{equation}
T_{l,N-1}^{(K)}\equiv c_{l,N-1}K_{l},\text{ with }c_{l,N-1}=(-1)^{N-l}(l-1)!(N-l)!\,.
\end{equation}

Under the condition that $K$ is diagonalizable with simple spectrum on $%
\mathbb{C}^{n}$, it follows that the $T_{l,N-1}^{(K)}$ for all $l\in
\{1,...,N\}$ form a system of $N$ operators simultaneously diagonalizable
and with simple spectrum. In particular, let us denote:%
\begin{equation}
\langle K_{a},j|K_{a}=k_{j}\langle K_{a},j|\text{ \ and }K_{a}|K_{a},j%
\rangle =|K_{a},j\rangle k_{j}\text{ }\forall (a,j)\in \{1,...,N\}\times
\{1,...,n\},
\end{equation}%
where we can fix their normalization by imposing:%
\begin{equation}
\langle K_{a},r|K_{a},s\rangle =\delta _{r,s}
\end{equation}%
we have that the common left and right eigenbasis of the $T_{l,N-1}^{(K)}$
read:%
\begin{equation}
\langle t_{h_{1},...,h_{N}}|=\bigotimes_{a=1}^{N}\langle K_{a},h_{a}|\text{,
\ }|t_{h_{1},...,h_{N}}\rangle =\bigotimes_{a=1}^{N}|K_{a},h_{a}\rangle 
\text{ \ }\forall (h_{1},...,h_{N})\in \{1,...,n\}^{N},
\end{equation}%
with:%
\begin{align}
\langle t_{h_{1},...,h_{N}}|T^{(K)}(\xi _{l},\{\xi \})& =\xi ^{N-1}c_{l,N-1}%
\mathsf{k}_{h_{l}}\langle t_{h_{1},...,h_{N}}|+O(\xi ^{N-2}), \\
T^{(K)}(\xi _{l},\{\xi \})|t_{h_{1},...,h_{N}}\rangle &
=|t_{h_{1},...,h_{N}}\rangle \xi ^{N-1}c_{l,N-1}\mathsf{k}_{h_{l}}+O(\xi
^{N-2}).
\end{align}%
By using now the following interpolation formula:%
\begin{equation}
T^{(K)}(\lambda ,\{\xi \})=\text{tr\thinspace }K\text{ }\prod_{a=1}^{N}(%
\lambda -\xi _{a})+\sum_{a=1}^{N}\prod_{b\neq a,b=1}^{N}\frac{\lambda -\xi
_{b}}{\xi _{a}-\xi _{b}}\ T^{(K)}(\xi _{a},\{\xi \}),
\end{equation}%
we get that for general values of the spectral parameter $\lambda $ the
transfer matrix $T^{(K)}(\lambda ,\{\xi \})$ is a polynomial of degree $N$
in $\xi $ with the following expansion:%
\begin{equation}
T^{(K)}(\lambda ,\{\xi \})=\xi ^{N}(-1)^{N}N!\text{tr\thinspace }%
K+(-1)^{N-1}\xi ^{N-1}\sum_{a=1}^{N}\prod_{b\neq a,b=1}^{N}\frac{b}{a-b}%
T_{a,N-1}^{(K)}+\hat{T}^{(K)}(\lambda ,\xi ),
\end{equation}%
where $\hat{T}^{(K)}(\lambda ,\xi )$ is a polynomial in $\xi $ of order $N-2$%
. So the left and right states $\langle t_{h_{1},...,h_{N}}|$\ and $%
|t_{h_{1},...,h_{N}}\rangle $ are left and right eigenstates of the leading
terms in $\xi $ of $T^{(K)}(\lambda ,\{\xi \})$. Note that this implies that
for any eigenvalue $t(\lambda )$ of $T^{(K)}(\lambda ,\{\xi \})$ denoted
with $\langle t|$\ and $|t\rangle $ the associated eigencovectors and
eigenvectors there exists a unique set $(h_{1},...,h_{N})\in
\{1,...,n\}^{N}$ such that:%
\begin{equation}
\lim_{\xi \rightarrow \infty }\text{ }\xi ^{1-N}\frac{\langle t|}{%
n_{t,L}(\xi )}\,T^{(K)}(\xi _{l},\{\xi \})=c_{l,N-1}k_{h_{l}}\langle
t_{h_{1},...,h_{N}}|,
\end{equation}%
and for almost any finite $\lambda $:%
\begin{align}
\lim_{\xi \rightarrow \infty }\text{ }\xi ^{1-N}\langle
t|\,(T^{(K)}(\lambda ,\{\xi \})-\xi ^{N}(-1)^{N}N!\text{tr\thinspace }%
K)& =  \notag \\
=(-1)^{N-1}(\sum_{a=1}^{N}c_{a,N-1}& k_{h_{a}}\prod_{b\neq a,b=1}^{N}\frac{b%
}{a-b})\langle t_{h_{1},...,h_{N}}|\,,
\end{align}%
once the nonzero normalization $n_{t,L}(\xi )$ of the eigencovector is
chosen properly and similar limits for the eigenvectors. So that it has to
hold:%
\begin{equation}
\lim_{\xi \rightarrow \infty }\frac{\langle t|}{n_{t,L}(\xi )}\,=\langle
t_{h_{1},...,h_{N}}|,\text{ \ \ }\lim_{\xi \rightarrow \infty }\frac{%
|t\rangle }{n_{t,R}(\xi )}\,=|t_{h_{1},...,h_{N}}\rangle ,
\end{equation}%
which in particular implies that:%
\begin{equation}
\lim_{\xi \rightarrow \infty }\frac{\langle t|t\rangle }{n_{t,L}(\xi
)n_{t,R}(\xi )}\,=\langle t_{h_{1},...,h_{N}}|t_{h_{1},...,h_{N}}\rangle =1,
\end{equation}%
and by the continuity argument it implies our statement
\begin{equation}
\langle t|t\rangle \neq 0
\end{equation}%
almost for any values of the inhomogeneities. This statement is true for the
left and right eigenstates associated to any eigenvalue of the transfer
matrix. Together with the already proven $w$-simplicity of the transfer
matrix spectrum this implies that there is no nontrivial Jordan block
associated to any transfer matrix eigenvalue. Indeed, in a non trivial Jordan block the right and left eigenvectors are orthogonal. So the transfer matrix is diagonalizable with simple spectrum.
\end{proof}

\subsection{Further examples as deformation of the $Y(gl_{n})$ case}

Let us consider an R-matrix,%
\begin{equation}
R_{a,b}(\lambda _{a},\lambda _{b}|\eta )\in \End(V_{a}\otimes V_{b}),\text{ \
with }V_{a}=\mathbb{C}^{n}\text{, }V_{b}=\mathbb{C}^{n}\text{, }n\in \mathbb{%
N}^{\ast }
\end{equation}%
that is regular and continuous in its parameters $\lambda_{a},\lambda _{b}\in \mathbb{C},\eta \in \mathbb{C}$ (to be more precise, we consider in particular cases where the $R$-matrix is indeed a trigonometric or elliptic polynomial of these parameters), solution of the Yang-Baxter equation written in $\End(V_{a}\otimes V_{b}\otimes V_{c})$,%
\begin{equation}
R_{a,b}(\lambda _{a},\lambda _{b}|\eta )R_{a,c}(\lambda _{a},\lambda
_{c}|\eta )R_{b,c}(\lambda _{b},\lambda _{c}|\eta )=R_{b,c}(\lambda
_{b},\lambda _{c}|\eta )R_{a,c}(\lambda _{a},\lambda _{c}|\eta
)R_{a,b}(\lambda _{a},\lambda _{b}|\eta )\ ,
\end{equation}%
and let us denote by SYB$_{\eta }\subset \End(\mathbb{C}^{n})$ the set of the
scalar solution of the Yang-Baxter equation:%
\begin{equation}
R_{a,b}(\lambda _{a},\lambda _{b}|\eta )K_{a}K_{b}=K_{b}K_{a}R_{a,b}(\lambda
_{a},\lambda _{b}|\eta )\in \End(V_{a}\in V_{b}\otimes V_{c}),
\end{equation}%
for any $K\in $SYB$_{\eta }$, i.e. the set of symmetries of the considered
R-matrix. Then we can define the following monodromy matrix,%
\begin{equation}
M_{a}^{(K)}(\lambda ,\{\xi _{1},...,\xi _{N}\}|\eta )\equiv
K_{a}R_{a,N}(\lambda _{a},\xi _{N}|\eta )\cdots R_{a,1}(\lambda _{a},\xi
_{1}|\eta ),
\end{equation}%
which satisfies the Yang-Baxter equation,%
\begin{equation}
R_{a,b}(\lambda _{a},\lambda _{b}|\eta )M_{a}^{(K)}(\lambda _{a},\{\xi
\}|\eta )M_{b}^{(K)}(\lambda _{b},\{\xi \}|\eta )=M_{b}^{(K)}(\lambda
_{b},\{\xi \}|\eta )M_{a}^{(K)}(\lambda _{b},\{\xi \}|\eta )R_{a,b}(\lambda
_{a},\lambda _{b}|\eta )
\end{equation}%
in $\End(V_{a}\otimes V_{b}\otimes \mathcal{H})$, with $\mathcal{H}\equiv
\otimes _{l=1}^{N}V_{l}$ and its dimension $d=n^{N}$. Hence it defines a
representation of the Yang-Baxter algebra associated to this $R$-matrix and
the following one parameter family of commuting transfer matrices:%
\begin{equation}
T^{(K)}(\lambda ,\{\xi \}|\eta )\equiv tr_{V_{a}}M_{a}^{(K)}(\lambda ,\{\xi
\}|\eta ).
\end{equation}%
In the following of this section, we use an upper index $Y$ in the R-matrix
and the monodromy matrix and down index in the transfer matrix to evidence
that these are those associated to the rational $Y(gl_{n})$ case studied in
the previous section. Then the following proposition holds:

\begin{proposition}
\label{gln-basis copy(1)} Let us assume that there exists $\eta _{0}\in 
\mathbb{C}$, and continuous functions $f(x,\eta )\in C^{0}(\mathbb{C}^{2}%
\mathbb{)}$ and $g(x,\eta )\in C^{0}(\mathbb{C}^{2}\mathbb{)}$ such that up
to a rescaling of the parameters $\lambda _{a}$ and $\lambda _{b}$ and
trivial overall normalization, the R-matrix satisfies the following Yangian
R-matrix limit,%
\begin{equation}
\lim_{\eta \rightarrow \eta _{0}}R_{a,b}(f(\lambda _{a},\eta ),g(\lambda
_{b},\eta )|\eta )=R_{a,b}^{Y}(\lambda _{a}-\lambda _{b}),
\end{equation}%
then,%
\begin{equation}
\langle h_{1},...,h_{N}|\equiv \langle S|\prod_{a=1}^{N}(T^{(K)}(\xi
_{a},\{\xi \}|\eta ))^{h_{a}}\text{\ }\forall \{h_{1},...,h_{N}\}\in
\{0,...,n-1\}^{\otimes N},
\end{equation}%
is a covector basis of $\mathcal{H}^{\ast }$ for almost any choice of the
covector $\langle S|$, of the value of $\eta \in \mathbb{C}$ \ and of the
inhomogeneities parameters under the only condition that the given $K\in $SYB%
$_{\eta }$ is $w$-simple on $\mathbb{C}^{n}$. In particular, for almost all
the value of $\eta \in \mathbb{C}$ \ and of the inhomogeneities parameters,
the covector in $\mathcal{H^{\ast }}$ of tensor product form:%
\begin{equation}
\langle S|\equiv \bigotimes_{a=1}^{N}\langle S,a|,
\end{equation}%
can be chosen as soon as we take for $\langle S,a|$ a local covector in $%
V_{a}^{\ast }$ such that%
\begin{equation}
\langle S,a|K_{a}^{h}\text{ with }h\in \{0,...,n-1\},
\end{equation}%
form a covector basis for $V_{a}$ for any $a\in \{1,...,N\}$, the existence
of $\langle S,a|$ being implied by the fact that $K$ is $w$-simple.
\end{proposition}

\begin{proof}
Let us define the $n^{N}\times n^{N}$ matrix $\mathcal{M}\left( \langle
S|,K,\{\xi \},\eta \right) $ with elements:%
\begin{equation}
\mathcal{M}_{i,j}\equiv \langle h_{1}(i),...,h_{N}(i)|e_{j}\rangle ,\text{ \ 
}\forall i,j\in \{1,...,n^{N}\}
\end{equation}%
where we have defined uniquely the $N$-tuple $(h_{1}(i),...,h_{N}(i))\in
\{1,...,n\}^{\otimes N}$ by the isomorphism introduced in the previous
section. Then the condition that the set $(\ref{SoV-basis-0-YB})$ form a
basis of covector in $\mathcal{H^{\ast }}$ is equivalent to the condition:%
\begin{equation}
\text{det}_{n^{N}}\mathcal{M}\left( \langle S|,K,\{\xi \},\eta \right) \neq
0.
\end{equation}%
Here, we assume that the $R$-matrix and so the transfer matrix are smooth
functions of their parameters: the transfer matrix is a polynomial in the
parameters of the $K$ matrix and in general a polynomial or a trigonometric
or an elliptic polynomial in the parameters $\{\xi _{1},...,\xi _{N}\}$ and $%
\eta $. Then the determinant $\text{det}_{n^{N}}\mathcal{M}\left( \langle
S|,K,\{\xi \},\eta \right) $ is itself a smooth functions of its parameters
(of the same type of the transfer matrix) and it is moreover a polynomial in
the coefficients $\langle S|e_{j}\rangle $ of the covector $\langle S|$.
Taking that into account, it is
enough to prove that the condition $(\ref{Lin-indep-SoV})$ holds for some
special limit on the parameters to prove that it is true for almost any
value of the parameters. Hence, as the following rational limit
\begin{equation}
\mathcal{M}_{Y}\left( \langle S|,K,\{\xi \}\right) =\lim_{\eta \rightarrow
\eta _{0}}\mathcal{M}\left( \langle S|,K,\{g(\xi ,\eta )\},\eta \right) ,
\end{equation}%
is satisfied as a consequence of the fact that the transfer matrix reduces to $T_{Y}^{(K)}(\lambda ,\{\xi \})$,
the $K$-twisted rational $gl_n$ transfer matrix, our current proposition is proven as a consequence of the one shown in the previous section for the rational case.
\end{proof}

We can similarly prove a statement about the diagonalizability and the
simple spectrum for these more general transfer matrices once
we impose some further requirements on the twist matrix. In fact, the
following proposition holds 

\begin{proposition}
Let us assume that there exists $\eta _{0}\in \mathbb{C}$, $f(x,\eta )\in
C^{0}(\mathbb{C}^{2}\mathbb{)}$ and $g(x,\eta )\in C^{0}(\mathbb{C}^{2}%
\mathbb{)}$ such that the R-matrix satisfies the following Yangian limit:%
\begin{equation}
\lim_{\eta \rightarrow \eta _{0}}R_{a,b}(f(\lambda _{a},\eta ),g(\lambda
_{b},\eta )|\eta )=R_{a,b}^{Y}(\lambda _{a}-\lambda _{b})\,,
\end{equation}%
to the rational $gl_n$ R-matrix, and let us assume that $K\in $SYB$_{\eta }$ is
diagonalizable with simple spectrum on $\mathbb{C}^{n}$, then, almost for
any values of the inhomogeneities and $\eta $, it holds:%
\begin{equation}
\langle t|t\rangle \neq 0,
\end{equation}%
where $|t\rangle $ and $\langle t|$ are the unique eigenvector and
eigencovector associated to $t(\lambda )$ a generic eigenvalue of $%
T^{(K)}(\lambda ,\{\xi \}|\eta )$ so that $T^{(K)}(\lambda ,\{\xi \}|\eta )$
is diagonalizable with simple spectrum.
\end{proposition}

\begin{proof}
We have already proven these statements for $T_{Y}^{(K)}(\lambda ,\{\xi \})$%
, the $K$-twisted rational $gl_n$ transfer matrix, the continuity argument
implies then that these statements are true also for almost any value of $%
\eta $.
\end{proof}

Some comments are in order to conclude these general considerations
linking the existence of a basis such as \eqref{sb1} and the properties of
the transfer matrix spectrum. At first let us stress that the existence of a
basis such as \eqref{sb1} is not enough to have the full SoV features. One
needs in addition to provide the necessary closure relations enabling for
the computation of the action of the transfer matrix on it for elements of
the basis associated to boundary values for the coordinates $h_{j}$, namely
if for some $j$, $h_{j}=d_{j}-1$. As we will see in the following sections
this information is provided by the fusion relations satisfied by the
transfer matrices. They will lead to the spectrum characterization in the
form of a quantum spectral curve. This characterization of the spectrum is
given in terms of secular equations of much lower degree compared to the
characteristic polynomial. In fact one gains an exponential factor going
from the exact diagonalization to the quantum spectral curve. Thus the
essence of the integrability properties of a given model is coming from the
unraveling of the non trivial structure constants of the commutative (and
associative algebra) of conserved charges.

\section{The quasi-periodic $Y(gl_2)$ fundamental model}
The integrable quantum models associated to the fundamental representations
of the Yang-Baxter algebra for the rational and trigonometric
6-vertex R-matrix are the first natural models for which it is interesting
to make explicit our SoV basis construction. Indeed, on the one
hand, in several cases the SoV approach as proposed by Sklyanin (or some
natural generalization of it) applies also for these cases, and therefore we can make a
comparison with the SoV construction that we propose. On the other hand,
we can already present cases for which the Sklyanin SoV scheme does not work directly 
while our new scheme applies.

\subsection{The $Y(gl_2)$ rational Yang-Baxter algebra}

Here, we consider the rational 6-vertex R-matrix solution of the Yang-Baxter
equation:%
\begin{equation}
R_{a,b}(\lambda )\equiv \left( 
\begin{array}{cccc}
\lambda +\eta & 0 & 0 & 0 \\ 
0 & \lambda & \eta & 0 \\ 
0 & \eta & \lambda & 0 \\ 
0 & 0 & 0 & \lambda +\eta%
\end{array}%
\right) \in \End(V_{a}\otimes V_{b})
\end{equation}%
this is just $R_{a,b}^{\left( Y\right) }(\lambda )$ of the previous section in the $n=2$ case, where we have
reintroduced a $\eta $ parameter for convenience, and $V_{a}\,$,  $V_{b}$
are bidimensional linear spaces. Then any $K\in \End (\mathbb{C}^{2})$ is a
scalar solution of the Yang-Baxter equation:%
\begin{equation}
R_{a,b}(\lambda )K_{a}K_{b}=K_{b}K_{a}R_{a,b}(\lambda )\ ,
\end{equation}%
i.e. gives a symmetry of the rational 6-vertex R-matrix. Then we can define the
following monodromy matrix:%
\begin{equation}
M_{a}^{(K)}(\lambda ,\{\xi _{1},...,\xi _{N}\})\equiv K_{a}R_{a,%
N}(\lambda _{a}-\xi _{N})\cdots R_{a,1}(\lambda _{a}-\xi
_{1})=\left( 
\begin{array}{cc}
A^{(K)}(\lambda ) & B^{(K)}(\lambda ) \\ 
C^{(K)}(\lambda ) & D^{(K)}(\lambda )%
\end{array}%
\right) ,
\end{equation}%
which satisfies the Yang-Baxter equation with the rational 6-vertex
R-matrix, so defining a fundamental spin 1/2 representation of the rational
6-vertex Yang-Baxter algebra and the following one parameter family of
commuting transfer matrices:%
\begin{equation}
T^{(K)}(\lambda ,\{\xi \})\equiv tr_{V_{a}}M_{a}^{(K)}(\lambda ,\{\xi \}).
\end{equation}%
In the following we assume that the inhomogeneity condition 
\begin{equation}
\xi _{a}\neq \xi _{b}+r\eta \text{ \ }\forall a\neq b\in \{1,...,N%
\}\,\,\text{and\thinspace \thinspace }r\in \{-1,0,1\},  \label{Inhomog-cond}
\end{equation}%
is satisfied.

\subsection{Sklyanin's construction of the SoV basis}

Let us now observe that the Sklyanin's approach to SoV applies with the
separate variables generated by the operators zeros of $B^{(K)}(\lambda )$
if and only if the twist matrix satisfies the condition:%
\begin{equation}
K=\left( 
\begin{array}{cc}
a & b\neq 0 \\ 
c & d%
\end{array}%
\right) .  \label{SoV-0-cond}
\end{equation}%
However, let us remark that given a $K\in\, \End (\mathbb{C}^{2})$ such that $%
K\neq \alpha I$, for any $\alpha \in \mathbb{C}$, either it satisfies this
condition directly or it exists a $W^{(K)}\in\, \End (\mathbb{C}^{2})$ such
that:%
\begin{equation}
\bar{K}=\left( W^{(K)}\right) ^{1}KW^{(K)}=\left( 
\begin{array}{cc}
\bar{a} & \bar{b}\neq 0 \\ 
\bar{c} & \bar{d}%
\end{array}%
\right) ,
\end{equation}%
then we can use $B^{(\bar{K})}(\lambda )$ to generate the SoV variables for%
\begin{equation}
T^{(\bar{K})}(\lambda ,\{\xi \})\equiv tr_{V_{a}}M_{a}^{(\bar{K})}(\lambda
,\{\xi \}).
\end{equation}%
Now from the identity:%
\begin{equation}
T^{(K)}(\lambda ,\{\xi \})=\mathcal{W}_{K}T^{(\bar{K})}(\lambda ,\{\xi \})%
\mathcal{W}_{K}^{-1},\text{ \ with }\mathcal{W}_{K}=\otimes _{a=1}^{\mathsf{N%
}}W_{a}^{(K)},
\end{equation}%
then it follows that the separate variables for $T^{(K)}(\lambda ,\{\xi \})$
are generated by:%
\begin{equation}
\mathcal{W}_{K}B^{(\bar{K})}(\lambda )\mathcal{W}%
_{K}^{-1}=tr_{V_{a}}[W_{a}^{(K)}\left( 
\begin{array}{cc}
0 & 0 \\ 
1 & 0%
\end{array}%
\right) _{a}\left( W^{(K)}\right) _{a}^{1}M_{a}^{(K)}(\lambda ,\{\xi \})].
\end{equation}%
The previous discussion shows that for the fundamental spin 1/2 rational
representation of the 6-vertex Yang-Baxter algebra associated to the
symmetry matrix $K\in\, \End (\mathbb{C}^{2})$ such that $K\neq \alpha I$, for
any $\alpha \in \mathbb{C}$, either one can use directly the Sklyanin's
definition of SoV or one can easily redefine the SoV generators by using the
symmetries of the rational 6-vertex matrix.

\subsection{Our approach to the SoV basis}

The general property as described in Proposition \ref{gln-basis} for the SoV basis applies to this special case $n=2$. Let us
see in this framework how it works concretely. We define,
\begin{equation}
\langle h_{1},...,h_{N}|\equiv \langle S|\prod_{a=1}^{N}(%
\frac{T^{(K)}(\xi _{a},\{\xi \})}{a(\xi _{a})})^{h_{a}}\text{ \ for any }%
\{h_{1},...,h_{N}\}\in \{0,1\}^{\otimes N}\ ,
\end{equation}%
where we have set,
\begin{equation}
a(\lambda -\eta )=d(\lambda )=\prod_{a=1}^{N}(\lambda -\xi _{a}),
\end{equation}%
to introduce the above normalization for reasons to become clear later on. If for
simplicity we take the state $\langle S|$ of the following tensor product
form:%
\begin{equation}
\langle S|=\bigotimes_{a=1}^{N}(x,y)_{a},
\end{equation}%
then for any matrix $K\in \End (\mathbb{C}^{2})$ not proportional to the identity matrix, it holds that,
\begin{equation}
(x,y)K^{i-1}\text{ for }i=1,2
\end{equation}%
form a covector basis for almost any $x,y\in \mathbb{C}$. Indeed, denoting as usual the canonical basis of $\mathbb{C}^2$ by $|e_j \rangle$ for $j=1,2$, we have,
\begin{equation}
\text{det}||\left( (x,y)K^{i-1}|e_{j} \rangle\right) _{i,j\in \{1,2\}}||=\text{det}%
\left( 
\begin{array}{cc}
x & y \\ 
ax+cy & bx+dy%
\end{array}%
\right) =bx^{2}+(d-a)xy+cy^{2}
\end{equation}%
which under the condition $K\neq \alpha I$, for any $\alpha \in \mathbb{C}$,
is non-zero for almost all the values of $x,y\in \mathbb{C}$. This also
implies that the above set of covectors is a basis for almost any choice of $x,y\in \mathbb{C}$.

\subsection{Comparison of the two SoV constructions}

Here we want to show that under some special choice of the covector $\langle
S|$, when the twist matrix $K\in \End(\mathbb{C}^{2})$ is not proportional
to the identity, our SoV left basis reduces to the SoV basis associated to
the Sklyanin's construction for the $K$ matrix satisfying $(\ref{SoV-0-cond})$ or
otherwise to its generalization described above. Let us start assuming that the 
$K$ matrix satisfies $(\ref{SoV-0-cond})$, we can write down explicitly the left
eigenbasis of $B^{(K)}(\lambda )$, i.e. the Sklyanin's SoV basis. Let us
remark that it holds:%
\begin{eqnarray}
A^{(K)}(\lambda ) &=&aA(\lambda )+bC(\lambda ),\text{ \ }B^{(K)}(\lambda
)=aB(\lambda )+bD(\lambda ), \\
C^{(K)}(\lambda ) &=&cA(\lambda )+dC(\lambda ),\text{ \ }D^{(K)}(\lambda
)=cB(\lambda )+dD(\lambda ),
\end{eqnarray}%
in terms of the elements of the original untwisted monodromy matrix. It is
well known that it holds:%
\begin{eqnarray}
\langle 0|A(\lambda ) &=&a(\lambda )\langle 0|,\text{ \ \ }\langle
0|B(\lambda )=0, \\
\langle 0|D(\lambda ) &=&d(\lambda )\langle 0|,\text{ \ \ }\langle
0|C(\lambda )\neq 0,
\end{eqnarray}%
where we have defined:%
\begin{equation}
\langle 0|=\bigotimes_{a=1}^{N}(1,0)_{a}.
\end{equation}%
So that it holds:%
\begin{equation}
\langle 0|B^{(K)}(\lambda )=bd(\lambda )\langle 0|
\end{equation}%
and by the Yang-Baxter commutation relations it follows:%
\begin{equation}
\underline{\langle h_{1},...,h_{N}|}B^{(K)}(\lambda )\equiv
b\prod_{a=1}^{N}(\lambda -\xi _{a}+h_{a}\eta )\underline{\langle
h_{1},...,h_{N}|},
\end{equation}%
where we have defined:%
\begin{equation}
\underline{\langle h_{1},...,h_{N}|}\equiv \langle 0|\prod_{a=1}^{%
N}(\frac{A^{(K)}(\xi _{a},\{\xi \})}{a(\xi _{a})})^{h_{a}}\text{ \
for any }\{h_{1},...,h_{N}\}\in \{0,1\}^{\otimes N},
\end{equation}%
so that $B^{(K)}(\lambda )$ is diagonalizable with simple spectrum. Now, let us 
prove that in fact our SoV basis coincides with the above basis, i.e. it holds:%
\begin{equation}
\langle h_{1},...,h_{N}|=\underline{\langle h_{1},...,h_{N%
}|}\text{ \ for any }\{h_{1},...,h_{N}\}\in \{0,1\}^{\otimes 
N}
\end{equation}%
as soon as we take:%
\begin{equation}
\langle S|=\langle 0| \ .
\end{equation}%
The proof is done by induction just using the identity, %
\begin{equation}
\langle 0|D^{(K)}(\xi _{a})=0\text{ \ }\forall a\in \{1,...,N\} \ ,
\label{D-zeros}
\end{equation}%
and the Yang-Baxter commutation relations:%
\begin{equation}
A^{(K)}\left( \mu \right) D^{(K)}\left( \lambda \right) =D^{(K)}\left(
\lambda \right) A^{(K)}\left( \mu \right) +\frac{\eta }{\lambda -\mu }%
(B^{(K)}\left( \lambda \right) C^{(K)}\left( \mu \right) -B^{(K)}\left( \mu
\right) C^{(K)}\left( \lambda \right) ).  \label{YB-AD-BC}
\end{equation}%
Let us assume that our statement holds for any state:%
\begin{equation}
\langle h_{1},...,h_{N}|=\underline{\langle h_{1},...,h_{N%
}|}\text{\ \ with \ }l=\sum_{a=1}^{N}h_{a}\leq N-1,
\end{equation}%
and let us show it for any state with $l+1$. To this aim we fix a state in
the above set and we denote with $\pi $ a permutation on the set $\{1,...,%
N\}$ such that:%
\begin{equation}
h_{\pi (a)}=1\text{ for }a\leq l\text{ \ and }h_{\pi (a)}=0\text{ for }l<a%
\text{ \ }
\end{equation}%
and let us take $c\in \{\pi (l+1),...,\pi (N)\}$ and let us compute:%
\begin{equation}
\underline{\langle h_{1},...,h_{N}|}T^{(K)}(\xi _{c},\{\xi
\})=\langle 0|\frac{A^{(K)}(\xi _{\pi (1)},\{\xi \})}{d(\xi _{\pi (1)}-\eta )%
}\cdots \frac{A^{(K)}(\xi _{\pi (l)},\{\xi \})}{d(\xi _{\pi (l)}-\eta )}%
(A^{(K)}+D^{(K)})(\xi _{c},\{\xi \}),
\end{equation}%
so that we have just to prove that,%
\begin{equation}
\langle 0|\frac{A^{(K)}(\xi _{\pi (1)},\{\xi \})}{d(\xi _{\pi (1)}-\eta )}%
\cdots \frac{A^{(K)}(\xi _{\pi (l)},\{\xi \})}{d(\xi _{\pi (l)}-\eta )}%
D^{(K)}(\xi _{c},\{\xi \})=0.  \label{UnWant-SoV}
\end{equation}%
From the commutation relation $(\ref{YB-AD-BC})$, the above covector can be
rewritten as it follows:%
\begin{eqnarray}
&&\langle 0|\frac{A^{(K)}(\xi _{\pi (1)},\{\xi \})}{d(\xi _{\pi (1)}-\eta )}%
\cdots \frac{A^{(K)}(\xi _{\pi (l-1)},\{\xi \})}{d(\xi _{\pi (l-1)}-\eta )}%
d^{-1}(\xi _{\pi (l)}-\eta )(D^{(K)}(\xi _{c},\{\xi \})A^{(K)}\left( \xi
_{\pi (l)},\{\xi \}\right)  \notag \\
&&+\eta (B^{(K)}\left( \xi _{c},\{\xi \}\right) C\left( \xi _{\pi (l)},\{\xi
\}\right) -B^{(K)}\left( \xi _{\pi (l)},\{\xi \}\right) C\left( \xi
_{c},\{\xi \}\right) )/(\xi _{c}-\xi _{\pi (l)}))\ ,
\end{eqnarray}%
which reduces to:%
\begin{equation}
\langle 0|\frac{A^{(K)}(\xi _{\pi (1)},\{\xi \})}{d(\xi _{\pi (1)}-\eta )}%
\cdots \frac{A^{(K)}(\xi _{\pi (l-1)},\{\xi \})}{d(\xi _{\pi (l-1)}-\eta )}%
D^{(K)}(\xi _{c},\{\xi \})\frac{A^{(K)}( \xi _{\pi (l)},\{\xi \}) }{d(\xi
_{\pi (l)}-\eta )}
\end{equation}%
once we observe that the state on the left of $B^{(K)}\left( \xi _{c},\{\xi
\}\right) $ and $B^{(K)}\left( \xi _{\pi (l)},\{\xi \}\right) $ are left
eigenstates of $B^{(K)}\left( \lambda ,\{\xi \}\right) $ with eigenvalue zeros
at $\lambda =\xi _{\pi (l)}$, $\xi _{c}$. That is we can perform the commutation of $A^{(K)}\left( \xi _{\pi (l)},\{\xi \}\right) $ with 
$D^{(K)}(\xi_{c},\{\xi \})$ in the covector \ref{UnWant-SoV} and by the same argument of $A^{(K)}\left( \xi _{\pi (r)},\{\xi\}\right) $ 
with $D^{(K)}(\xi _{c},\{\xi \})$ for any $r\leq l-1$ up to bring $D^{(K)}(\xi _{c},\{\xi \})$ completely to the left acting on $\langle 0|$ which proves 
$(\ref{UnWant-SoV})$ as a consequence of $(\ref{D-zeros})$. Let us also mention that one can prove directly, using the Yang-Baxter commutation relations, that our SoV basis is an eigenstate basis of the operator $B^{(K)}(\lambda )$.

Finally, let us assume that $K$ does not satisfy $(\ref{SoV-0-cond})$ but it
isn't proportional to the identity. Then we can apply the generalization of
the Sklyanin's construction w.r.t. the $\bar{K}$ satisfying $(\ref%
{SoV-0-cond})$. Now the generalized Sklyanin's left SoV basis is the left
eigenbasis of the operator family $\mathcal{W}_{K}B^{(\bar{K})}(\lambda )%
\mathcal{W}_{K}^{-1}$ given by:%
\begin{equation}
\underline{\langle h_{1},...,h_{N}|}\equiv \langle 0|\prod_{a=1}^{%
N}(\frac{A^{(\bar{K})}(\xi _{a},\{\xi \})}{a(\xi _{a})})^{h_{a}}%
\mathcal{W}_{K}^{-1}\text{ \ for any }\{h_{1},...,h_{N}\}\in
\{0,1\}^{\otimes N},
\end{equation}%
and it holds:%
\begin{equation}
\underline{\langle h_{1},...,h_{N}|}\mathcal{W}_{K}B^{(\bar{K}%
)}(\lambda ,\{\xi \})\mathcal{W}_{K}^{-1}\equiv \bar{b}\prod_{a=1}^{\mathsf{N%
}}(\lambda -\xi _{a}+h_{a}\eta )\underline{\langle h_{1},...,h_{N}|}%
.
\end{equation}%
Repeating the argument proven above then it holds:%
\begin{equation}
\langle 0|\prod_{a=1}^{N}(\frac{A^{(\bar{K})}(\xi _{a},\{\xi \})}{%
a(\xi _{a})})^{h_{a}}=\langle 0|\prod_{a=1}^{N}(\frac{T^{(\bar{K}%
)}(\xi _{a},\{\xi \})}{a(\xi _{a})})^{h_{a}},
\end{equation}%
and so from the identity $T^{(K)}(\lambda ,\{\xi \})=\mathcal{W}_{K}T^{(\bar{%
K})}(\lambda ,\{\xi \})\mathcal{W}_{K}^{-1}$ it follows:%
\begin{equation}
\langle 0|\prod_{a=1}^{N}(\frac{A^{(\bar{K})}(\xi _{a},\{\xi \})}{%
a(\xi _{a})})^{h_{a}}\mathcal{W}_{K}^{-1}=\langle S|\prod_{a=1}^{N}(%
\frac{T^{(K)}(\xi _{a},\{\xi \})}{a(\xi _{a})})^{h_{a}}
\end{equation}%
once we fix:%
\begin{equation}
\langle S|=\langle 0|\mathcal{W}_{K}^{-1},
\end{equation}%
i.e. our SoV basis coincides with the generalized Sklyanin's one with a
special choice of the $\langle S|$ covector for any $K\neq \alpha I$, for
any $\alpha \in \mathbb{C}$, for which both do exist.

\subsection{Transfer matrix spectrum in our SoV scheme}

Let us show here how to characterize the transfer
matrix spectrum in our new SoV scheme. In our introductory sections, we have anticipated that in
our SoV basis the separate relations are given directly by the
particularization of the fusion relations in the spectrum of the separate
variables. In the case at hand these fusion relations just reduces to the
following identities:%
\begin{equation}
T^{(K)}(\xi _{a},\{\xi \})T^{(K)}(\xi _{a}-\eta ,\{\xi \})=q\text{-det}%
M^{(K)}(\xi _{a},\{\xi \}),\text{ }\forall a\in \{1,...,N\},
\label{Fusion-1-eq}
\end{equation}%
where:%
\begin{equation}
q\text{-det}M^{(K)}(\lambda ,\{\xi \})=a(\lambda )d(\lambda -\eta )\text{det}%
K
\end{equation}%
is the quantum determinant given as the quadratic expression,%
\begin{equation}
q\text{-det}M^{(K)}(\lambda ,\{\xi \})=A^{(K)}(\lambda ,\{\xi
\})D^{(K)}(\lambda -\eta ,\{\xi \})-B^{(K)}(\lambda ,\{\xi
\})C^{(K)}(\lambda -\eta ,\{\xi \})\ ,
\end{equation}%
and three other equivalent ones. One has to add the knowledge of
the analytic properties of the transfer matrix that we can easily derive. In
fact, $T^{(K)}(\lambda ,\{\xi \})$ is a polynomial of degree $1$ in all the $%
\xi _{a}$ and of degree $N$ in $\lambda $ with the following
leading central coefficient:%
\begin{equation}
\lim_{\lambda \rightarrow +\infty }\lambda ^{-N}T^{(K)}(\lambda
,\{\xi \})=\text{tr\thinspace }K\, .
\end{equation}%
Introducing the notation
\begin{equation}
g_{a}(\lambda )=\prod_{b\neq a,b=1}^{N}\frac{\lambda -\xi _{b}}{\xi
_{a}-\xi _{b}}\ ,
\end{equation}%
the following theorem holds:
\begin{theorem}
Let us assume that $K\neq \alpha I$, for any $\alpha \in \mathbb{C}$, and
that the inhomogeneities $\{\xi _{1},...,\xi _{N}\}\in \mathbb{C}^{%
N}$ satisfy the condition $(\ref{Inhomog-cond})$, then the spectrum
of $T^{(K)}(\lambda ,\{\xi \})$ is characterized by:%
\begin{equation}
\Sigma _{T^{(K)}}=\left\{ t(\lambda ):t(\lambda )=\text{tr\thinspace }K\text{
}\prod_{a=1}^{N}(\lambda -\xi _{a})+\sum_{a=1}^{N%
}g_{a}(\lambda )x_{a},\text{ \ \ }\forall \{x_{1},...,x_{N}\}\in
\Sigma _{T}\right\} ,  \label{SET-T}
\end{equation}%
$\Sigma _{T}$ is the set of solutions to the following inhomogeneous system
of $N$ quadratic equations:%
\begin{equation}
x_{n}[\text{tr\thinspace }K\text{ }\prod_{a=1}^{N}(\xi _{n}-\xi
_{a}-\eta )+\sum_{a=1}^{N}g_{a}(\xi _{n}-\eta )x_{a}]=a(\xi
_{n})d(\xi _{n}-\eta )\text{det}K,\text{ }\forall n\in \{1,...,N\},
\label{Quadratic System}
\end{equation}%
in $N$ unknown $\{x_{1},...,x_{N}\}$. Moreover, $%
T^{(K)}(\lambda ,\{\xi \})$ has $w$-simple spectrum and for any $t(\lambda
)\in \Sigma _{T^{(K)}}$ the associated unique (up-to normalization that we take to be given by $\langle S | t \rangle=1$)
eigenvector $|t\rangle $ has the following factorized wave-function in the left SoV
basis:%
\begin{equation}
\langle h_{1},...,h_{N}|t\rangle =\prod_{n=1}^{N}\left( 
\frac{t(\xi _{n})}{a(\xi _{n})}\right) ^{h_{n}}.  \label{SoV-Ch-T-eigenV-gl2}
\end{equation}
\end{theorem}

\begin{proof}
Let us start observing that the inhomogeneous system of $N$
quadratic equations $(\ref{Quadratic System})$ in $N$ unknown $%
\{x_{1},...,x_{N}\}$ is nothing else but  the rewriting of the transfer
matrix fusion equations:%
\begin{equation}
t(\xi _{a})t(\xi _{a}-\eta )=q\text{-det}M^{(K)}(\xi _{a},\{\xi \}),\text{ }%
\forall a\in \{1,...,N\},  \label{scalar-fusion-1-gl2}
\end{equation}%
for the set of all the polynomials of degree $N$ of the form:%
\begin{equation}
t(\lambda )=\text{tr\thinspace }K\text{ }\lambda ^{N}+\sum_{a=1}^{%
N}t_{a}\lambda ^{a-1},  \label{T-eigenVa-form}
\end{equation}%
where we have used for these functions the interpolation formula in the
points $\{\xi _{1},...,\xi _{N}\}$. Then, it is clear that any
eigenvalue of the transfer matrix $T^{(K)}(\lambda ,\{\xi \})$ is solution
of this system, i.e. $(\ref{scalar-fusion-1-gl2})$, and that the associated
right eigenvector $|t\rangle $ admits the characterization $(\ref%
{SoV-Ch-T-eigenV-gl2})$ in the left SoV basis. So we are left with the proof of
the reverse statement, i.e. that any polynomial $t(\lambda )$ of the above
form satisfying this system is an eigenvalue of the transfer matrix. We will prove this by showing that the vector $|t\rangle $ characterized by $(\ref%
{SoV-Ch-T-eigenV-gl2})$ is a transfer matrix eigenstate, i.e. we have to show:%
\begin{equation}
\langle h_{1},...,h_{N}|T^{(K)}(\lambda ,\{\xi \})|t\rangle
=t(\lambda )\langle h_{1},...,h_{N}|t\rangle ,\text{ }\forall
\{h_{1},...,h_{N}\}\in \{0,1\}^{\otimes N}.
\end{equation}%
Let us define:%
\begin{equation}
\xi _{a}^{\left( h\right) }=\xi _{a}-h\eta ,\text{ }h\in \{0,1\},
\end{equation}%
and let us write the following interpolation formula for the transfer matrix:%
\begin{equation}
T^{(K)}(\lambda ,\{\xi \})=\text{tr\thinspace }K\text{ }\prod_{a=1}^{\mathsf{%
N}}(\lambda -\xi _{a}^{\left( h_{a}\right) })+\sum_{a=1}^{N%
}\prod_{b\neq a,b=1}^{N}\frac{\lambda -\xi _{b}^{\left(
h_{b}\right) }}{\xi _{a}^{\left( h_{a}\right) }-\xi _{b}^{\left(
h_{b}\right) }}T^{(K)}(\xi _{a}^{\left( h_{a}\right) },\{\xi \}),
\end{equation}%
and use it to act on the generic element of the left SoV basis. Then, we
have:%
\begin{equation}
\langle h_{1},...,h_{a},...,h_{N}|T^{(K)}(\xi _{a}^{\left(
h_{a}\right) },\{\xi \})|t\rangle =\left\{ 
\begin{array}{l}
a(\xi _{a})\langle h_{1},...,h_{a}^{\prime }=1,...,h_{N}|t\rangle 
\text{ \ \ if \ }h_{a}=0 \\ 
q\text{-det}M^{(K)}(\xi _{a},\{\xi \})\frac{\langle h_{1},...,h_{a}^{\prime
}=0,...,h_{N}|t\rangle }{a(\xi _{a})}\text{ \ \ if \ }h_{a}=1%
\end{array}%
\right.
\end{equation}%
which by the definition of the state $|t\rangle $ can be rewritten as:%
\begin{equation}
\langle h_{1},...,h_{a},...,h_{N}|T^{(K)}(\xi _{a}^{\left(
h_{a}\right) },\{\xi \})|t\rangle =\left\{ 
\begin{array}{l}
t(\xi _{a})\prod_{n\neq a,n=1}^{N}\left( \frac{t(\xi _{n})}{a(\xi
_{n})}\right) ^{h_{n}}\text{ \ \ if \ }h_{a}=0 \\ 
\frac{q\text{-det}M^{(K)}(\xi _{a},\{\xi \})}{a(\xi _{a})}\prod_{n\neq
a,n=1}^{N}\left( \frac{t(\xi _{n})}{a(\xi _{n})}\right) ^{h_{n}}%
\text{ \ \ if \ }h_{a}=1%
\end{array}%
\right.
\end{equation}%
and finally by the equation $(\ref{scalar-fusion-1-gl2})$ reads:%
\begin{equation}
\langle h_{1},...,h_{a},...,h_{N}|T^{(K)}(\xi _{a}^{\left(
h_{a}\right) },\{\xi \})|t\rangle =\left\{ 
\begin{array}{l}
t(\xi _{a})\prod_{n\neq a,n=1}^{N}\left( \frac{t(\xi _{n})}{a(\xi
_{n})}\right) ^{h_{n}}\text{ \ \ \ if \ }h_{a}=0 \\ 
t(\xi _{a}-\eta )\prod_{n=1}^{N}\left( \frac{t(\xi _{n})}{a(\xi
_{n})}\right) ^{h_{n}}\text{ \ \ if \ }h_{a}=1%
\end{array}%
\right. ,
\end{equation}%
and so:%
\begin{equation}
\langle h_{1},...,h_{a},...,h_{N}|T^{(K)}(\xi _{a}^{\left(
h_{a}\right) },\{\xi \})|t\rangle =t(\xi _{a}^{\left( h_{a}\right) })\langle
h_{1},...,h_{a},...,h_{N}|t\rangle ,
\end{equation}%
from which we have, using the polynomial interpolation,%
\begin{equation}
\langle h_{1},...,h_{N}|T^{(K)}(\lambda ,\{\xi \})|t\rangle =( 
\text{tr\thinspace }K\text{ }\prod_{a=1}^{N}(\lambda -\xi
_{a}^{\left( h_{a}\right) })+\sum_{a=1}^{N}\prod_{b\neq a,b=1}^{%
N}\frac{\lambda -\xi _{b}^{\left( h_{b}\right) }}{\xi _{a}^{\left(
h_{a}\right) }-\xi _{b}^{\left( h_{b}\right) }}t(\xi _{a}^{\left(
h_{a}\right) })) \langle h_{1},...,h_{N}|t\rangle ,
\end{equation}%
proving our statement.
\end{proof}

Let us comment that the same characterization of the transfer matrix eigenvalues and eigenvectors is obtained in the Sklyanin's like SoV representations. This is natural as we have shown that under special choice of the covector $\bra S$ the two SoV basis coincide. Nevertheless, it is worth remarking the slightly different point of view that we used here. In the Sklyanin's like SoV approach the fusion relations are derived as compatibility conditions for the existence of nonzero eigenvectors. Here, instead, they are used as the starting point to prove that the vectors $\ket t$ of the form \rf{SoV-Ch-T-eigenV-gl2} are indeed eigenvectors of the transfer matrix.

The previous characterization of the spectrum allows to introduce a
functional equation which provides an equivalent characterization of it, by the so-called quantum
spectral curve, which in the case at hand is a second order Baxter
difference equation.

\begin{theorem}
Let us assume that $K\neq \alpha I$, for any $\alpha\in \mathbb{C}$, and has at least one
non-zero eigenvalue (The case where $K$ is a pure Jordan block with eigenvalue zero is not very interesting as in that situation the transfer matrix is quite degenerated being proportional to the nilpotent operator $B(\la)$ or $C(\la)$; however, our method would anyway work in those cases too), that the inhomogeneities $\{\xi _{1},...,\xi _{\mathsf{N%
}}\}\in \mathbb{C}^{N}$ satisfy the condition $(\ref{Inhomog-cond})$%
. Moreover, let us introduce the coefficients:%
\begin{equation}
\alpha (\lambda )=\beta (\lambda )\beta (\lambda -\eta ),\text{ }\beta
(\lambda )=\mathsf{k}_{0}a(\lambda ),
\end{equation}%
where $\mathsf{k}_{0}\neq 0$ is solution of the equation:%
\begin{equation}
\mathsf{k}_{0}^{2}-\mathsf{k}_{0}\text{\thinspace tr}K+\text{det\thinspace }%
K=0\text{,}  \label{Ch-Eq-K-gl2}
\end{equation}%
i.e. $\mathsf{k}_{0}$ is a non-zero eigenvalue of the matrix\thinspace $K$,
and let $t(\lambda )$ be an entire function of $\lambda $, then $t(\lambda )$
is an element of the spectrum of $T^{(K)}(\lambda ,\{\xi \})$ if and only if
there exists a unique polynomial:%
\begin{equation}
Q_{t}(\lambda )=\prod_{a=1}^{\mathsf{M}}(\lambda -\lambda _{a})\text{, with }%
\mathsf{M}\leq N\text{ \ such that }\lambda _{a}\neq \xi _{b},
\label{Q-form}
\end{equation}%
for any $\left( a,b\right) \in \{1,...,\mathsf{M}\}\times \{1,...,N%
\},$ such that  $t(\lambda )$ and $Q_{t}(\lambda )$ are solutions of the following
quantum spectral curve functional equation:%
\begin{equation}
\label{t-q-gl2}
\alpha (\lambda )Q_{t}(\lambda -2\eta )-\beta (\lambda )t(\lambda -\eta
)Q_{t}(\lambda -\eta )+q\text{-det}M^{(K)}(\lambda ,\{\xi \})Q_{t}(\lambda
)=0.
\end{equation}%
Moreover, up to a normalization the associated transfer matrix eigenvector $%
|t\rangle $ admits the following rewriting in the left SoV basis:%
\begin{equation}
\langle h_{1},...,h_{N}|t\rangle =\mathsf{k}_{0}^{\sum_{n=1}^{%
N}h_{n}}\prod_{n=1}^{N}Q_{t}(\xi _{n}^{(h_{n})}).
\end{equation}
\end{theorem}

\begin{proof}
Let us start assuming the existence of $Q_{t}(\lambda )$ satisfying with $%
t(\lambda )$ the functional equation \eqref{t-q-gl2}, then it follows that $t(\lambda )$ is
a polynomial of degree $N$ with leading coefficient $t_{N%
+1}$ satisfying the equation:%
\begin{equation}
\mathsf{k}_{0}^{2}-\mathsf{k}_{0}\text{\thinspace }t_{N+1}+\text{%
det\thinspace }K=0\text{,}
\end{equation}%
which imposes $t_{N+1}=$tr$K=\mathsf{k}_{0}+\mathsf{k}_{1}$, where $%
\mathsf{k}_{1}$ is the second eigenvalue or $\mathsf{k}_{1}=\mathsf{k}_{0}$
if $K$ has a non-trivial Jordan block. Now from the identities:%
\begin{equation}
q\text{-det}M^{(K)}(\xi _{a}+\eta ,\{\xi \})=\alpha (\xi _{a})=0,
\end{equation}%
we have that the functional equation reduces to the system of equations:%
\begin{eqnarray}
-t(\xi _{a}-\eta )Q_{t}(\xi _{a}-\eta )+\mathsf{k}_{1}d(\xi _{a}-\eta
)Q_{t}(\xi _{a}) &=&0,  \label{Sys-SoV-1} \\
\mathsf{k}_{0}a(\xi _{a})Q_{t}(\xi _{a}-\eta )-t(\xi _{a})Q_{t}(\xi _{a})
&=&0,  \label{Sys-SoV-2}
\end{eqnarray}%
once computed in the points $\xi _{a}$ and $\xi _{a}+\eta ,$ from which it
follows:%
\begin{equation}
t(\xi _{a}-\eta )\frac{t(\xi _{a})Q_{t}(\xi _{a})}{\mathsf{k}_{0}a(\xi _{a})}%
=\mathsf{k}_{1}d(\xi _{a}-\eta )Q_{t}(\xi _{a})
\end{equation}%
which being $Q_{t}(\xi _{a})\neq 0$ implies that $t(\lambda )$ satisfies
also the system of equations $(\ref{scalar-fusion-1-gl2})$, for any $a\in
\{1,...,N\}$, so that $t(\lambda )$ is a transfer matrix eigenvalue
for the previous theorem.

Let us now prove the reverse statement, i.e. we assume that $t(\lambda )$ is
a transfer matrix eigenvalue and we want to prove the existence of the
polynomial $Q_{t}(\lambda )$ satisfying the functional equation. The l.h.s.
of the equation is a polynomial in $\lambda $ of maximal degree $2N+%
\mathsf{M}$, with $\mathsf{M}\leq N$, so that if we prove that it
is zero in $3N+1$ different points we have proven the functional
equation. The leading coefficient of this polynomial is zero thanks to $(\ref%
{Ch-Eq-K-gl2})$ once we ask that $t(\lambda )$ is a transfer matrix eigenvalue.
It is easy to remark that in the points $\xi _{a}-\eta $, for any $a\in \{1,...,%
N\}$, the functional equation is directly satisfied. Finally, it is
satisfied in the $2N$ points $\xi _{a}$ and $\xi _{a}+\eta $, for
any $a\in \{1,...,N\},$ if the system $(\ref{Sys-SoV-1})$-$(\ref%
{Sys-SoV-2})$ is satisfied. As a consequence of the fact that $t(\lambda )$ satisfies \rf{scalar-fusion-1-gl2}, this last system reduces e.g. to the system of
the second $N$ equations:%
\begin{equation}
\mathsf{k}_{0}a(\xi _{a})Q_{t}(\xi _{a}-\eta )=t(\xi _{a})Q_{t}(\xi _{a}),
\end{equation}%
that one can prove to be satisfied by a polynomial $Q_{t}(\lambda )$ of the
form $(\ref{Q-form})$, see for example \cite{KitMNT16}. Moreover, following the proof of
the Theorem 2.3 of \cite{KitMNT16} one can prove also here that the function $%
Q_{t}(\lambda )$ is unique. Finally, from the identities:%
\begin{equation}
\prod_{n=1}^{N}Q_{t}(\xi _{n})\prod_{n=1}^{N}\left( \frac{%
t(\xi _{n})}{a(\xi _{n})}\right) ^{h_{n}}=\mathsf{k}_{0}^{\sum_{n=1}^{%
N}h_{n}}\prod_{n=1}^{N}Q_{t}(\xi _{n}^{(h_{n})}),
\end{equation}%
our statement on the representation of the transfer matrix eigenstate in the
left SoV basis follows.
\end{proof}

\section{The quasi-periodic $XXZ$ spin-1/2 fundamental model}

Let us consider here the trigonometric 6-vertex R-matrix:%
\begin{equation}
R_{12}(\lambda )=\left( 
\begin{array}{cccc}
\sinh (\lambda +\eta ) & 0 & 0 & 0 \\ 
0 & \sinh \lambda & \sinh \eta & 0 \\ 
0 & \sinh \eta & \sinh \lambda & 0 \\ 
0 & 0 & 0 & \sinh (\lambda +\eta )%
\end{array}%
\right)\ ,
\end{equation}% 
which is a solution of the Yang-Baxter equation:%
\begin{equation}
R_{12}(\lambda -\mu )R_{13}(\lambda )R_{23}(\mu )=R_{23}(\mu )R_{13}(\lambda
)R_{12}(\lambda -\mu ).
\end{equation}%
The following family of $2\times 2$ matrices:%
\begin{equation}
K_{0}^{\left( a,\alpha \right) }= \left[ \delta
_{0,a}\left( 
\begin{array}{cc}
e^{\alpha } & 0 \\ 
0 & e^{-\alpha }%
\end{array}%
\right) _{0}+\delta _{1,a}\left( 
\begin{array}{cc}
0 & e^{\alpha } \\ 
e^{-\alpha } & 0%
\end{array}%
\right) _{0}\right] ,\text{ }\forall (a,\alpha )\in \{0,1\}\times  \mathbb{C}^{2},
\end{equation}%
characterizes the symmetries of the trigonometric 6-vertex R-matrix:%
\begin{equation}
R_{12}(\lambda -\mu )K_{1}^{\left( a,\alpha \right) }K_{2}^{\left(
a,\alpha \right) }=K_{2}^{\left( a,\alpha \right) }K_{1}^{\left(
a,\alpha \right) }R_{12}(\lambda -\mu ),
\end{equation}%
i.e. the scalar solutions of the trigonometric 6-vertex Yang-Baxter
equation. Let us comment that we can add a normalization factor to this matrix which allows to describe also the case in which one ($a=0$ case) or both ($a=1$ case) the eigenvalues of $K$ are zero while $K$ stays $w$-simple. While our SoV basis construction applies also for these degenerate cases, we omit them to simplify the notations as these special cases correspond to simple transfer matrices, coinciding with one of the elements of the monodromy matrix and for which direct diagonalization methods exist already. Then by using the R-matrix as Lax operator and the K-matrix as
symmetry twist we can construct the monodromy matrices (Unless they play an explicit role, and to simplify the notations, we omit the upper indices in the matrix $K$ when it is used as un upper index itself for the monodromy and transfer matrices):%
\begin{equation}
M_{0}^{\left( K\right) }(\lambda )=K_{0}^{\left( a,\alpha \right)
}R_{0N}(\lambda -\xi _{N})\dots R_{01}(\lambda -\xi
_{1})=\left( 
\begin{array}{cc}
A(\lambda ) & B(\lambda ) \\ 
C(\lambda ) & D(\lambda )%
\end{array}%
\right)  \label{T}
\end{equation}%
solutions to the trigonometric 6-vertex Yang-Baxter equation:%
\begin{equation}
R_{12}(\lambda -\mu )M_{1}^{\left( K\right) }(\lambda )M_{2}^{\left(
K\right) }(\mu )=M_{2}^{\left( K\right) }(\mu )M_{1}^{\left( K\right)
}(\lambda )R_{12}(\lambda -\mu ),  \label{YB}
\end{equation}%
in the 2$^{N}$-dimensional representation space:%
\begin{equation}
\mathcal{H}=\otimes _{n=1}^{N}\mathcal{H}_{n}.
\end{equation}%
Then the associated transfer matrices:%
\begin{equation}
T^{\left( K\right) }(\lambda )=\text{tr}_{0}M_{0}^{\left( K\right) }(\lambda
)\in \End (\mathcal{H}),  \label{trig-6v-transfer}
\end{equation}%
defines a one parameter family of commuting operators. As in the rational case, in the following we assume that the inhomogeneity condition 
\begin{equation}
\xi _{a}\neq \xi _{b}+r\eta+im\pi \text{ \ }\forall a\neq b\in \{1,...,N%
\}\, ,\,r\in \{-1,0,1\} \text{ and\thinspace \thinspace } m\in\mathbb{Z} \label{Inhomog-cond-trigo}
\end{equation}%
is satisfied. The main properties enjoyed by this transfer matrices are collected in the following Lemma.

\begin{lemma}
The transfer matrix satisfies the following fusion equations:%
\begin{equation}
T^{(K)}(\xi _{a})T^{(K)}(\xi _{a}-\eta )=q\text{-det}M^{(K)}(\xi _{a}),\text{
}\forall a\in \{1,...,N\},
\end{equation}%
where:%
\begin{equation}
q\text{-det}M^{(K)}(\lambda )=a(\lambda )d(\lambda -\eta )\text{det}%
K^{\left( a,\alpha \right) }
\end{equation}%
with%
\begin{equation}
a(\lambda -\eta )\left. =\right. d(\lambda )=\prod_{n=1}^{N}\sinh
(\lambda -\xi _{n}),\text{ \ det} K^{\left( a,\alpha \right)
}=\left( -1\right) ^{a}
\end{equation}%
is the quantum determinant a one-parameter family of central elements of the
Yang-Baxter algebra, which admits the following quadratic form in the
generator of the Yang-Baxter algebra:%
\begin{equation}
q\text{-det}M^{(K)}(\lambda )=A^{(K)}(\lambda )D^{(K)}(\lambda -\eta
)-B^{(K)}(\lambda )C^{(K)}(\lambda -\eta )
\end{equation}%
and others three equivalent ones. Moreover, $T^{(K)}(\lambda )$ is a
trigonometric polynomial of degree $1$ in all the $\xi _{n}$ and of degree $%
N$ in $\lambda $ with the following leading operator coefficients:%
\begin{equation}
T_{\pm N}(S_{z})\equiv \lim_{\lambda \rightarrow \pm \infty }e^{\mp
\lambda N}T^{(K^{\left( a,\alpha \right) })}(\lambda
)=\delta _{0,a} \frac{(-1)^{(1\mp 1)\frac{N}{2}}e^{\pm (\frac{\eta 
N}{2}-\sum_{n=1}^{N}\xi _{n})}}{2^{N-1}}\cosh (%
\frac{\eta }{2}S_{z}\pm \alpha ),  \label{Asymp-T}
\end{equation}%
where:%
\begin{equation}
S_{z}=\sum_{n=1}^{N}\sigma _{n}^{z}.
\end{equation}
\end{lemma}

Moreover, it is interesting to present explicitly the interpolation formula
for the transfer matrix which follows from the previous lemma:

\begin{lemma}
Let us take the generic $K^{^{\left( a,\alpha \right) }}\neq vI$,
for any $v\in \mathbb{C}$, then the transfer matrix:%
\begin{equation}
T^{(K^{^{\left( a,\alpha \right) }})}(\lambda )=[\delta
_{0,a}(e^{\alpha }A(\lambda )-e^{-\alpha }D(\lambda ))+\delta
_{1,a}(e^{\alpha }C(\lambda )-e^{-\alpha }B(\lambda ))],
\end{equation}%
for any choice of $\{h_{1},...,h_{N}\}\in \{0,1\}^{N}$\
admits the following interpolation formula:%
\begin{align}
T^{(K^{^{\left( a,\alpha \right) }})}(\lambda )& =2i%
\delta _{0,a}\frac{\cosh \alpha \cosh (\eta S_{z}/2)}{\sinh t_{h_{1},...,h_{N}}}\prod_{n=1}^{N}\sinh
(\lambda -\xi _{n}^{\left( h_{n}\right) })+  \notag \\
& \sum_{n=1}^{N}\left( \frac{\sinh (t_{h_{1},...,h_{N%
}}+\lambda -\xi _{n}^{(h_{n})})}{\sinh t_{h_{1},...,h_{N}}}\right)
^{\delta _{0,a}}\prod_{b\neq n,b=1}^{N}\frac{\sinh (\lambda -\xi
_{b}^{\left( h_{b}\right) })}{\sinh (\xi _{n}^{\left( h_{n}\right) }-\xi
_{b}^{\left( h_{b}\right) })}T^{(K^{^{\left( a,\alpha \right)
}})}(\xi _{n}^{\left( h_{n}\right) }),
\end{align}
where we have defined:%
\begin{equation}
t_{h_{1},...,h_{N}}=\eta (\frac{N}{2}-\sum_{n=1}^{N%
}h_{n})+\frac{i\pi }{2},  \label{t.choice}
\end{equation}%
and moreover it holds:%
\begin{equation}
\delta _{0,a}  \sum_{\epsilon =\pm 1}\epsilon e^{\epsilon (\frac{%
\eta N}{2}-\sum_{n=1}^{N}h_{n})}\cosh (\frac{\eta }{2}%
S_{z}+\epsilon \alpha )=\sum_{a=1}^{N}\frac{T^{(K^{^{\left( a,%
\alpha \right) }})}(\xi _{a}^{\left( h_{a}\right) })}{%
\prod_{b\neq a,b=1}^{N}\sinh (\xi _{a}^{\left( h_{a}\right) }-\xi
_{b}^{\left( h_{b}\right) })},
\end{equation}%
which in the case $\left( a=0,\alpha =i\pi /2\right) $ reads:%
\begin{equation}
\sinh S_{z}=\frac{1}{2ik\cosh \eta (\frac{N}{2}%
-\sum_{n=1}^{N}h_{n})}\sum_{a=1}^{N}\frac{T^{(K^{^{\left(
a=0,\alpha =i\pi /2\right) }})}(\xi _{a}^{\left( h_{a}\right) })}{%
\prod_{b\neq a,b=1}^{N}\sinh (\xi _{a}^{\left( h_{a}\right) }-\xi
_{b}^{\left( h_{b}\right) })}.
\end{equation}
\end{lemma}

\begin{proof}
The proof just follows matching the known leading asymptotic of the transfer
matrix $T^{(K^{^{\left( a,\alpha \right) }})}(\lambda )$ with
that corresponding to the general interpolation formula:%
\begin{align}
T_{h_{1},...,h_{N}}\prod_{n=1}^{N}\sinh (\lambda -\xi
_{n}^{\left( h_{n}\right) })+\sum_{n=1}^{N}&\left( \frac{\sinh
(t_{h_{1},...,h_{N}}+\lambda -\xi _{a}^{(h_{a})})}{\sinh
t_{h_{1},...,h_{N}}}\right) ^{\delta _{0,a}} \times \nonumber\\
& \times \prod_{b\neq n,b=1}^{%
N}\frac{\sinh (\lambda -\xi _{b}^{\left( h_{b}\right) })}{\sinh
(\xi _{n}^{\left( h_{n}\right) }-\xi _{b}^{\left( h_{b}\right) })}%
T^{(K^{^{\left( a,\alpha \right) }})}(\xi _{n}^{\left(
h_{n}\right) }),
\end{align}%
%i.e determining $T_{h_{1},...,h_{N}}$
fixing $t_{h_{1},...,h_{N}}$ by $(\ref{t.choice})$.
\end{proof}

\subsection{Our approach to the SoV basis}

The general construction for the SoV basis applies to this special case. Let us
see in this framework how it works.

\begin{theorem}
For almost any choice of $\langle S|$, of $K^{^{\left( a,\alpha
\right) }}\neq vI$, for any $v\in \mathbb{C}$, and under the condition $(\ref%
{Inhomog-cond-trigo})$, then the following set of covectors:%
\begin{equation}
\langle h_{1},...,h_{N}|\equiv \langle S|\prod_{n=1}^{N}(%
\frac{T^{(K)}(\xi _{n})}{(e^{\alpha })^{\delta _{0,a}}a(\xi _{n})}%
)^{h_{n}}\text{ \ for any }\{h_{1},...,h_{N}\}\in \{0,1\}^{\otimes 
N},
\end{equation}%
forms a covector basis of $\cal{H}$. In particular, we can take the
state $\langle S|$ of the following tensor product form:%
\begin{equation}
\langle S|=\bigotimes_{a=1}^{N}(x,y)_{a},
\end{equation}%
asking that $xy\neq 0$ in the case $K^{^{\left( a=0,\alpha
\right) }}\neq vI$, for any $v\in \mathbb{C},$ while asking that $y\neq \pm
e^{\alpha }x$ in the case $K^{^{\left( a=1,\alpha \right) }}$.
\end{theorem}

\begin{proof}
The proof proceed taking the limit of the trigonometric $R$-matrix toward the rational case as explained in the general subsection 2.4. In particular, we have a basis choosing the
state of the above factorized form once the condition:%
\begin{equation}
\text{det}||\left( (x,y)K^{i-1}e_{j}(a)\right) _{i,j\in \{1,2\}}||=\text{det}%
\left( 
\begin{array}{cc}
x & y \\ 
Ax+Cy & Bx+Dy%
\end{array}%
\right) =Bx^{2}+(D-A)xy+Cy^{2},
\end{equation}%
is satisfied, where we have defined:\begin{equation}
A=\delta _{0,a}e^{\alpha },\text{ }D=\delta _{0,a}e^{-\alpha },\text{ }%
B=\delta _{1,a}e^{\alpha },\text{ }C=\delta _{0,a}e^{-\alpha },
\end{equation}
which clearly leads to the given requirements on the components $x,y\in \mathbb{C}$
of the two dimensional covector.
\end{proof}

\subsection{Comparison with the Sklyanin's SoV construction}

Let us start recalling that the Sklyanin's approach for the transfer
matrices associated to the trigonometric 6-vertex Yang-Baxter algebra have
been developed only for the case of a twist matrix of the form $K^{^{\left(
a=1,\alpha \right) }}$ \cite{Nic12,Nic13,Nic13a,Nic13b,FalKN14,KitMN14,NicT15}. In the elliptic 8-vertex Yang-Baxter
algebra it was possible to analyze also transfer matrices associated to
twist matrix $K=I$ or $\sigma ^{z}$ \cite{Nic13a,FalN13,LevNT15,NicT16}, however, only after the introduction of
Baxter's like gauge transformations. It can be interesting to analyze if the
trigonometric version of these Baxter's gauge transformation allow also to
describe in the generalized Sklyanin's approach the case $K=I$ or $\sigma
^{z}$. We leave this question for further analysis, on the other hand, our
approach to SoV applies for any $K^{^{\left( a,\alpha \right)
}}\neq vI$.

So the comparison can be made only in the case $K^{^{\left( a=1,%
\alpha \right) }}$ for which we have:

\begin{lemma}
Let us fix $a=1$, so that it holds:%
\begin{eqnarray}
A^{(K)}(\lambda ) &=&e^{\alpha }C(\lambda ),\text{ \ }%
B^{(K)}(\lambda )=e^{\alpha }D(\lambda ), \\
C^{(K)}(\lambda ) &=&e^{-\alpha }A(\lambda ),\text{ \ }%
D^{(K)}(\lambda )=e^{-\alpha }B(\lambda ),
\end{eqnarray}%
and the Sklyanin's SoV basis is the covector basis of $B^{(K)}(\lambda )$:%
\begin{equation}
\underline{\langle h_{1},...,h_{N}|}B^{(K)}(\lambda )\equiv e^{\alpha }\prod_{a=1}^{N}\sinh (\lambda -\xi _{a}+h_{a}\eta )%
\underline{\langle h_{1},...,h_{N}|},
\end{equation}%
where we have defined:%
\begin{equation}
\underline{\langle h_{1},...,h_{N}|}\equiv \langle 0|\prod_{a=1}^{%
N}(\frac{A^{(K)}(\xi _{a},\{\xi \})}{ia(\xi _{a})}%
)^{h_{a}}\text{ \ for any }\{h_{1},...,h_{N}\}\in \{0,1\}^{\otimes 
N},
\end{equation}%
so that $B^{(K)}(\lambda )$ is diagonalizable with simple spectrum. Then our
SoV approach reproduces the Sklyanin's SoV approach, i.e. it holds:%
\begin{equation}
\langle h_{1},...,h_{N}|=\underline{\langle h_{1},...,h_{N%
}|}\text{ \ for any }\{h_{1},...,h_{N}\}\in \{0,1\}^{\otimes 
N}
\end{equation}%
as soon as we fix:%
\begin{equation}
\langle S|=\bigotimes_{a=1}^{N}(1,0)_{a}.
\end{equation}
\end{lemma}

\begin{proof}
The proof is done by induction just by the same steps of the rational
6-vertex case, by using the identity:%
\begin{equation}
\langle 0|D^{(K)}(\xi_a )=0
\end{equation}%
and the following Yang-Baxter commutation relations:%
\begin{equation}
A^{(K)}\left( \mu \right) D^{(K)}\left( \lambda \right) =D^{(K)}\left(
\lambda \right) A^{(K)}\left( \mu \right) +\frac{\sinh \eta }{\sinh (\lambda
-\mu )}(B^{(K)}\left( \lambda \right) C^{(K)}\left( \mu \right) -B^{(K)}\left( \mu
\right) C^{(K)}\left( \lambda \right) ).
\end{equation}
\end{proof}

\subsection{Transfer matrix spectrum in our SoV scheme}

Let us show here how in our SoV schema it is characterized the transfer
matrix spectrum. In our introductory chapter, we have anticipated that in
our SoV basis the separate relations are given directly by the
particularization of the fusion relations at the spectrum of the separate
variables. Let us define:%
\begin{equation}
g_{n}^{(a)}(\lambda )=\left( \frac{\cosh (\eta N/2+\lambda -\xi
_{n})}{\cosh \eta N/2}\right) ^{\delta _{0,a}}\prod_{b\neq n,b=1}^{%
N}\frac{\sinh (\lambda -\xi _{b})}{\sinh (\xi _{n}-\xi _{b})},
\end{equation}%
then the following theorem holds.

\begin{theorem}\label{Ch-SoV-discrete-trig}
Let us assume that $K\neq vI$, for any $v\in \mathbb{C}$, and that the
inhomogeneities $\{\xi _{1},...,\xi _{N}\}\in \mathbb{C}^{N%
}$ satisfy the condition $(\ref{Inhomog-cond-trigo})$, then the spectrum of $%
T^{(K)}(\lambda )$ is characterized by:%
\begin{equation}
\Sigma _{T^{(K)}}=\bigcup_{l=\{-N,2-N,....,+N-2,%
N\}}\Sigma _{T^{(K)}}^{\left( l\right) },
\end{equation}%
where we have defined:%
\begin{align}
& \Sigma _{T^{(K)}}^{\left( l\right) }=\notag \\
&= \left\{ t(\lambda ):t(\lambda )=2
\delta _{0,a}  \frac{\cosh \alpha \cosh (\eta l/2)}{\cosh \eta N/2}\prod_{n=1}^{%
N}\sinh (\lambda -\xi _{n})+
 \sum_{n=1}^{N%
}g_{n}^{(a)}(\lambda )x_{n},\text{\ }\forall \{x_{1},...,x_{N}\}\in
\Sigma _{T}^{\left( l\right) }\right\} 
\end{align}%
$\Sigma _{T}^{\left( l\right) }$ is the set of solutions to the following
inhomogeneous system of $N$ quadratic equations:%
\begin{align}
& x_{n}[2\delta _{0,a}\frac{\cosh \alpha \cosh (\eta l/2)}{\cosh \eta N/2}\prod_{b=1}^{N}\sinh (\xi _{n}-\xi _{b}-\eta )+\sum_{b=1}^{\mathsf{%
N}}g_{b}^{(a)}(\xi _{n}-\eta )x_{b}]=  \notag \\
&\hspace{10 cm}%
\left. =\right. a(\xi _{n})d(\xi _{n}-\eta )\text{det}K,
\label{Q-System- trig}
\end{align}
in $N$ unknown $\{x_{1},...,x_{N}\}$, where the integer $l\in\{-N,2-N,....,+N-2,%
N\}$ is fixed without ambiguity by the following sum rule:
\begin{equation}\label{roots-sum-role}
\delta _{0,a}\sum_{\epsilon =\pm 1}\epsilon e^{\epsilon \frac{\eta N}{2}}\cosh (\frac{\eta }{2}l+\epsilon \alpha
)=\sum_{a=1}^{N}\frac{x_{a}}{\prod_{b\neq a,b=1}^{N}\sinh (\xi _{a}-\xi _{b})}.
\end{equation}
Moreover, $T^{(K)}(\lambda ,\{\xi \})$ has $w$-simple spectrum and for any $t(\lambda
)\in \Sigma _{T^{(K)}}$ the associated unique (up-to normalization that we take to be given by $\langle S | t \rangle=1$)
eigenvector $|t\rangle $ has the following wave-function in the left SoV
basis:%
\begin{equation}
\langle h_{1},...,h_{N}|t\rangle =\prod_{n=1}^{N}\left( 
\frac{t(\xi _{n})}{(e^{\alpha })^{\delta _{0,a}}a(\xi _{n})}%
\right) ^{h_{n}}.  \label{SoV-Ch-T-eigenV-trig}
\end{equation}
\end{theorem}

\begin{proof}
Let us start observing that the inhomogeneous system of $N$ quadratic
equations $(\ref{Q-System- trig})$ in $N$ unknown $\{x_{1},...,x_{N}\}$ is
nothing else the rewriting of the transfer matrix fusion equations:%
\begin{equation}
t(\xi _{a})t(\xi _{a}-\eta )=q\text{-det}M^{(K)}(\xi _{a}),\text{ }\forall
a\in \{1,...,N\},  \label{scalar-fusion-trig}
\end{equation}%
for the set of all the trigonometric polynomials of degree $N$ with
asymptotic terms given by $(\ref{Asymp-T})$. Then, it is clear that any
eigenvalue of the transfer matrix $T^{(K)}(\lambda ,\{\xi \})$ is solution
of this system, satisfying in particular the associated sum rule, and that
the associated right eigenvector $|t\rangle $ admits the characterization $(%
\ref{SoV-Ch-T-eigenV-trig})$ in the left SoV basis.

So we are left with the proof of the reverse statement, i.e. that any
polynomial $t(\lambda )$ of the above form satisfying this system is an
eigenvalue of the transfer matrix and this is proven by proving that the
vector $|t\rangle $ characterized by $(\ref{SoV-Ch-T-eigenV-trig})$ is a
transfer matrix eigenstate, i.e. we have to show:%
\begin{equation}
\langle h_{1},...,h_{N}|T^{(K)}(\lambda ,\{\xi \})|t\rangle =t(\lambda
)\langle h_{1},...,h_{N}|t\rangle ,\text{ }\forall \{h_{1},...,h_{N}\}\in
\{0,1\}^{\otimes N}.
\end{equation}%
Let us observe that:%
\begin{equation}
\langle h_{1},...,h_{n},...,h_{N}|T^{(K)}(\xi _{n}^{\left( h_{n}\right)
})|t\rangle =\left\{ 
\begin{array}{l}
(e^{\alpha })^{\delta _{0,a}}a(\xi _{n})\langle h_{1},...,h_{n}^{\prime
}=1,...,h_{N}|t\rangle \text{ \ \ if \ }h_{n}=0 \\ 
q\text{-det}M^{(K)}(\xi _{n})\frac{\langle h_{1},...,h_{n}^{\prime
}=0,...,h_{N}|t\rangle }{(e^{\alpha })^{\delta _{0,a}}a(\xi _{n})}\text{ \ \
if \ }h_{n}=1%
\end{array}%
\right. 
\end{equation}%
which by the definition of the state $|t\rangle $ can be rewritten as:%
\begin{equation}
\langle h_{1},...,h_{n},...,h_{N}|T^{(K)}(\xi _{n}^{\left( h_{n}\right)
})|t\rangle =\left\{ 
\begin{array}{l}
t(\xi _{n})\prod_{m\neq n,m=1}^{N}\left( \frac{t(\xi _{m})}{(e^{\alpha
})^{\delta _{0,a}}a(\xi _{m})}\right) ^{h_{m}}\text{ \ \ if \ }h_{n}=0 \\ 
\frac{q\text{-det}M^{(K)}(\xi _{n})}{(e^{\alpha })^{\delta _{0,a}}a(\xi _{n})%
}\prod_{m\neq n,m=1}^{N}\left( \frac{t(\xi _{m})}{(e^{\alpha })^{\delta
_{0,a}}a(\xi _{m})}\right) ^{h_{m}}\text{ if \ }h_{n}=1%
\end{array}%
\right. 
\end{equation}%
and finally by the equation $(\ref{scalar-fusion-trig})$ reads:%
\begin{equation}
\langle h_{1},...,h_{n},...,h_{N}|T^{(K)}(\xi _{n}^{\left( h_{n}\right)
})|t\rangle =\left\{ 
\begin{array}{l}
t(\xi _{n})\prod_{n\neq a,n=1}^{N}\left( \frac{t(\xi _{m})}{(e^{\alpha
})^{\delta _{0,a}}a(\xi _{m})}\right) ^{h_{m}}\text{ \ \ \ if \ }h_{n}=0 \\ 
t(\xi _{n}-\eta )\prod_{m=1}^{N}\left( \frac{t(\xi _{m})}{(e^{\alpha
})^{\delta _{0,a}}a(\xi _{m})}\right) ^{h_{m}}\text{ \ \ if \ }h_{n}=1%
\end{array}%
\right. ,
\end{equation}%
and so:%
\begin{equation}
\langle h_{1},...,h_{n},...,h_{N}|T^{(K)}(\xi _{n}^{\left( h_{n}\right)
})|t\rangle =t(\xi _{n}^{\left( h_{n}\right) })\langle
h_{1},...,h_{n},...,h_{N}|t\rangle .  \label{Inter-eigen-trig}
\end{equation}%
Let us comment that the sum rule \rf{roots-sum-role} implies that the trigonometric polynomial $t(\lambda )$
satisfies the following asymptotics: 
\begin{equation}
\lim_{\lambda \rightarrow \pm \infty }e^{\mp \lambda N}t(\lambda )=\frac{%
(-1)^{(1\mp 1)\frac{N}{2}}e^{\pm (\frac{\eta N}{2}-\sum_{n=1}^{N}\xi _{n})}}{2^{N-1}}%
\cosh (\frac{\eta }{2}l\pm \alpha )\equiv T_{\pm N}(l),  \label{Verif-asym}
\end{equation}%
from which it also follows that the following general sum rule is satisfied:%
\begin{equation}
\sum_{a=1}^{N}\frac{t(\xi _{a}^{\left( h_{a}\right) })}{\prod_{b\neq
a,b=1}^{N}\sinh (\xi _{a}^{\left( h_{a}\right) }-\xi _{b}^{\left(
h_{b}\right) })}\left. =\right. \delta _{0,a}\sum_{\epsilon =\pm 1}\epsilon
e^{\epsilon (\frac{\eta N}{2}-\sum_{n=1}^{N}h_{n})}\cosh (\frac{\eta }{2}%
l+\epsilon \alpha ).
\end{equation}%
Now by using this identity and the one for the transfer matrix%
\begin{equation}
\sum_{a=1}^{N}\frac{T^{(K^{^{\left( a,,\alpha \right) }})}(\xi _{a}^{\left(
h_{a}\right) })}{\prod_{b\neq a,b=1}^{N}\sinh (\xi _{a}^{\left( h_{a}\right)
}-\xi _{b}^{\left( h_{b}\right) })}=\delta _{0,a}\sum_{\epsilon =\pm
1}\epsilon e^{\epsilon (\frac{\eta N}{2}-\sum_{n=1}^{N}h_{n})}\cosh (\frac{%
\eta }{2}S_{z}+\epsilon \alpha ),
\end{equation}%
we obtain:%
\begin{eqnarray}
&&\langle h_{1},...,h_{N}|\left( \delta _{0,a}\sum_{\epsilon =\pm 1}\epsilon
e^{\epsilon (\frac{\eta N}{2}-\sum_{n=1}^{N}h_{n})}\cosh (\frac{\eta }{2}%
S_{z}+\epsilon \alpha )\right) |t\rangle  \\
&&\left. =\right. \langle h_{1},...,h_{N}|\left( \sum_{a=1}^{N}\frac{%
T^{(K^{^{\left( a,,\alpha \right) }})}(\xi _{a}^{\left( h_{a}\right) })}{%
\prod_{b\neq a,b=1}^{N}\sinh (\xi _{a}^{\left( h_{a}\right) }-\xi
_{b}^{\left( h_{b}\right) })}\right) |t\rangle  \\
&&\left. =\right. \left( \sum_{a=1}^{N}\frac{t(\xi _{a}^{\left( h_{a}\right)
})}{\prod_{b\neq a,b=1}^{N}\sinh (\xi _{a}^{\left( h_{a}\right) }-\xi
_{b}^{\left( h_{b}\right) })}\right) \langle h_{1},...,h_{N}|t\rangle  \\
&&\left. =\right. \left( \delta _{0,a}\sum_{\epsilon =\pm 1}\epsilon
e^{\epsilon (\frac{\eta N}{2}-\sum_{n=1}^{N}h_{n})}\cosh (\frac{\eta }{2}%
l+\epsilon \alpha )\right) \langle h_{1},...,h_{N}|t\rangle 
\end{eqnarray}%
that is $|t\rangle $ is an eigenvector of $S_{z}$ with eigenvalue $l$%
\begin{equation}
S_{z}|t\rangle =|t\rangle l,
\end{equation}%
for any $t(\lambda )\in \Sigma _{T^{(K)}}^{\left( l\right) }$ with $a=0$ and
for any $\alpha \neq in\pi $ with $n$ integer. So that by using the
interpolation formula for both the transfer matrix and the function $%
t(\lambda )\in \Sigma _{T^{(K)}}^{\left( l\right) }$ we prove our theorem.\end{proof}

The previous characterization of the spectrum allows to introduce an equivalent characterization in terms of a 
functional equation, the so-called quantum
spectral curve equation, which in the case at hand is a second order Baxter's difference equation.  Here, we present the functional equation reformulation of the SoV
spectrum characterization only in the case of diagonal $K$-twist as this case cannot be directly derived in the standard
SoV Sklyanin's approach while it can be achieved with our new formulation of SoV. For the case of non-diagonal $K$-twist the
reformulation proven in the standard SoV \cite{NicT15} applies as well as, of course, to our present reformulation of the SoV basis.

\begin{theorem}
Let us assume that the twist matrix has the form $K^{^{\left( a=0%
,\alpha \right) }}\neq xI$, for any $x\in \mathbb{C}$, that the
inhomogeneities $\{\xi _{1},...,\xi _{N}\}\in \mathbb{C}^{N%
}$ satisfy the condition $(\ref{Inhomog-cond-trigo})$. Moreover, let us introduce
the coefficients:%
\begin{equation}
\alpha (\lambda )=\beta (\lambda )\beta (\lambda -\eta ),\text{ }\beta
(\lambda )=\mathsf{k}_{0}\mathsf{\,}a(\lambda ),
\end{equation}%
with $\mathsf{k}_{0}$ the eigenvalue $e^{\alpha }$ of $K$ and let $%
t(\lambda )$ be an entire function of $\lambda $, then $t(\lambda )\in
\Sigma _{T^{(K)}}^{\left( N-2\mathsf{M}\right) }$ if and only if
there exists a unique polynomial:%
\begin{equation}
Q_{t}(\lambda )=\prod_{n=1}^{\mathsf{M}}\sinh(\lambda -\lambda _{n})\text{, with }%
\mathsf{M}\leq N\text{ \ such that }\lambda _{n}\neq \xi _{m}+ir\pi \text{ \ }\forall r\in\mathbb{Z},
\label{Q-form-trig}
\end{equation}%
for any $\left( n,m\right) \in \{1,...,\mathsf{M}\}\times \{1,...,N%
\},$ such that $t(\lambda )$ and $Q_{t}(\lambda )$ are solutions of the following
quantum spectral curve functional equation:%
\begin{equation}
\alpha (\lambda )Q_{t}(\lambda -2\eta )-\beta (\lambda )t(\lambda -\eta
)Q_{t}(\lambda -\eta )+q\text{-det}M^{(K)}(\lambda ,\{\xi \})Q_{t}(\lambda
)=0.
\end{equation}%
Moreover, up to a normalization the associated transfer matrix eigenvector $%
|t\rangle $ admits the following rewriting in the left SoV basis:%
\begin{equation}
\langle h_{1},...,h_{N}|t\rangle =\prod_{n=1}^{N}Q_{t}(\xi
_{n}^{(h_{n})}).
\end{equation}
\end{theorem}

\begin{proof}

Let us start assuming the existence of $Q_{t}(\lambda )$ satisfying with $%
t(\lambda )$ the functional equation, which can be rewritten also in a
Baxter's like form:%
\begin{equation}
\mathsf{k}_{0}a(\lambda )Q_{t}(\lambda -\eta )-t(\lambda )Q_{t}(\lambda )+%
\mathsf{k}_{1}d(\lambda )Q_{t}(\lambda +\eta )=0,
\end{equation}%
where $\mathsf{k}_{1}=e^{-\alpha }$ is the second eigenvalue of $K$. Then it
follows that $t(\lambda )$ is a trigonometric polynomial of degree $N$ with
the following leading coefficient:%
\begin{equation}
\lim_{\lambda \rightarrow +\infty }e^{\mp \lambda N}t(\lambda )\equiv T_{\pm
N}(N-2\mathsf{M}).  \label{Verif-asym-1}
\end{equation}%
Now from the identities:%
\begin{equation}
q\text{-det}M^{(K)}(\xi _{a}+\eta ,\{\xi \})=\alpha (\xi _{a})=0,
\end{equation}%
we have that the functional equation reduces to the system of equations:%
\begin{eqnarray}
-t(\xi _{a}-\eta )Q_{t}(\xi _{a}-\eta )+\mathsf{k}_{1}d(\xi _{a}-\eta
)Q_{t}(\xi _{a}) &=&0,  \label{S-SoV-1} \\
\mathsf{k}_{0}a(\xi _{a})Q_{t}(\xi _{a}-\eta )-t(\xi _{a})Q_{t}(\xi _{a})
&=&0,  \label{S-SoV-2}
\end{eqnarray}%
once computed in the points $\xi _{a}$ and $\xi _{a}+\eta ,$ from which it
follows:%
\begin{equation}
t(\xi _{a}-\eta )\frac{t(\xi _{a})Q_{t}(\xi _{a})}{\mathsf{k}_{0}a(\xi _{a})}%
=\mathsf{k}_{1}d(\xi _{a}-\eta )Q_{t}(\xi _{a})
\end{equation}%
which being $Q_{t}(\xi _{a})\neq 0$ implies that $t(\lambda )$ satisfies
also the system of equations $(\ref{scalar-fusion-trig})$, for any $a\in
\{1,...,N\}$, so that for the previous theorem $t(\lambda )$ is a transfer
matrix eigenvalue belonging to $\Sigma _{T^{(K)}}^{\left( a=\mathsf{N}-2%
\mathsf{M}\right) }$.

Let us now prove the reverse statement, i.e. we assume that $t(\lambda )$ is
a transfer matrix eigenvalue and we want to prove the existence of the
polynomial $Q_{t}(\lambda )$ satisfying the functional equation. It is easy
to remark that in the points $\xi _{a}-\eta $, for any $a\in \{1,...,N\}$,
the functional equation is directly satisfied. Moreover, it is satisfied in
the $2N$ points $\xi _{a}$ and $\xi _{a}+\eta $, for any $a\in \{1,...,N\},$
if the system $(\ref{S-SoV-1})$-$(\ref{S-SoV-2})$ is satisfied. This last
system reduces to the system of the second $N$ equations:%
\begin{equation}
\mathsf{k}_{0}a(\xi _{a})Q_{t}(\xi _{a}-\eta )=t(\xi _{a})Q_{t}(\xi _{a}),
\end{equation}
by the fusion equations satisfied by the transfer matrix eigenvalue $t(\la)$.
Then, following similar steps to those used in the rational case, one can prove the existence and unicity of a polynomial $Q_{t}(\lambda )
$ of the form:%
\begin{equation}
Q_{t}(\lambda )=\prod_{n=1}^{\mathsf{R}}\sinh (\lambda -\lambda _{n})\text{,
with }\mathsf{R}\leq N\text{ \ such that }\lambda _{n}\neq \xi _{m}+if\pi \
\forall n,\,m\text{ and}\ f\in \mathbb{Z},  \label{Q-form-trig2}
\end{equation}%
see for example \cite{NicT15}, solution of the above system of equations. So we are
left with the proof of the identity $\mathsf{R}=\mathsf{M}$ once $t(\lambda
)\in $ $\Sigma _{T^{(K)}}^{\left( N-2\mathsf{M}\right) }$. In order to prove
that, let us define:%
\begin{equation}
F_{1}(\lambda )=t(\lambda )Q_{t}(\lambda ),\text{ }F_{2}(\lambda )=\mathsf{k}%
_{0}a(\lambda )Q_{t}(\lambda -\eta )+\mathsf{k}_{1}d(\lambda )Q_{t}(\lambda
+\eta )
\end{equation}%
these are two trigonometric polynomial of degree $N+\mathsf{R}$ with the
following asymptotic:%
\begin{equation}
\lim_{\lambda \rightarrow +\infty }e^{\mp \lambda (N+\mathsf{R}%
)}F_{h}(\lambda )=F_{h}^{\pm }
\end{equation}%
with $\mathsf{M}_{h}=(\delta _{1,h}\mathsf{M}+\delta _{2,h}\mathsf{R})$ and%
\begin{equation}
F_{h}^{\pm }=\frac{e^{\pm (\frac{\eta N}{2}-\sum_{n=1}^{N}\xi
_{n}-\sum_{n=1}^{\mathsf{R}}\lambda _{n})}}{(-1)^{(1\mp 1)(N+\mathsf{R}%
)}2^{(N+\mathsf{R})-1}}\cosh (\frac{\eta }{2}(N-2\mathsf{M}_{h})\pm \alpha ).
\end{equation}%
Let us now consider the interpolation formulae for the $F_{h}(\lambda )$ in $%
N+\mathsf{R}$ distinct points $\{x_{1},...,x_{N+\mathsf{R}}\}$, (as already
shown in the case of the transfer matrix) then one can prove the following
sum rules:%
\begin{equation}
\sum_{\epsilon =\pm 1}\frac{\cosh (\frac{\eta }{2}(N-2\mathsf{M}%
_{h})+\epsilon \alpha )}{e^{\epsilon (\sum_{n=1}^{N}\xi _{n}+\sum_{n=1}^{%
\mathsf{R}}\lambda _{n}-\frac{\eta N}{2}-\sum_{n=1}^{N+\mathsf{R}}x_{n})}}%
=\sum_{n=1}^{N+\mathsf{R}}\frac{F_{h}(x_{n})}{\prod_{b\neq n,b=1}^{N}\sinh
(x_{n}-x_{b})}.
\end{equation}%
So that we obtain the identity:%
\begin{equation}
\sum_{\epsilon =\pm 1}\frac{\cosh (\frac{\eta }{2}(N-2\mathsf{M})+\epsilon
\alpha )}{e^{\epsilon (\sum_{n=1}^{N}\xi _{n}+\sum_{n=1}^{\mathsf{R}}\lambda
_{n}-\frac{\eta N}{2}-\sum_{n=1}^{N+\mathsf{R}}x_{n})}}=\sum_{\epsilon =\pm
1}\frac{\cosh (\frac{\eta }{2}(N-2\mathsf{R})+\epsilon \alpha )}{e^{\epsilon
(\sum_{n=1}^{N}\xi _{n}+\sum_{n=1}^{\mathsf{R}}\lambda _{n}-\frac{\eta N}{2}%
-\sum_{n=1}^{N+\mathsf{R}}x_{n})}}.  \label{Asymp-ID}
\end{equation}%
for any choice of $\{x_{1},...,x_{N+\mathsf{R}}\}\subset \{\xi _{1},...,\xi
_{N},\xi _{1}-\eta ,...,\xi _{N}-\eta \}$ as the system $(\ref{S-SoV-1})$-$(%
\ref{S-SoV-2})$ implies in particular the identities:%
\begin{equation}
F_{1}(x_{n})=F_{2}(x_{n})\text{ }\forall n\in \{1,....,N+\mathsf{R}\}.
\end{equation}%
Then the identity $(\ref{Asymp-ID})$ implies $\mathsf{R}=\mathsf{M}$ being $%
\eta ,$ $\alpha $ and $\xi _{n}$ arbitrary.

So that we have shown that the l.h.s. of the quantum spectral curve equation
is zero in $3N$ different points plus the points at infinity, this last
statement being true as we have shown $\mathsf{R}=\mathsf{M}$ for $t(\lambda
)\in $ $\Sigma _{T^{(K)}}^{\left( a=\mathsf{N}-2\mathsf{M}\right) }$. Now
being this l.h.s. a trigonometric polynomial in $\lambda $ of maximal degree 
$2N+\mathsf{M}$, with $\mathsf{M}\leq N$, we have proven this functional
equation.

Finally, the
identities:%
\begin{equation}
\prod_{n=1}^{N}Q_{t}(\xi _{n})\prod_{n=1}^{N}\left( \frac{t(\xi _{n})}{%
\mathsf{k}_{0}a(\xi _{n})}\right) ^{h_{n}}=\prod_{n=1}^{\mathsf{N}}Q_{t}(\xi
_{n}^{(h_{n})}),
\end{equation}%
imply our statement on the representation of the transfer matrix eigenstate
in the left SoV basis.
\end{proof}

\subsection{Algebraic Bethe ansatz form of separate states: the diagonal
case}

Let us consider here the case $K^{^{\left( a=0,\alpha \right)
}}\neq xI$, for any $x\in \mathbb{C}$, which can be described also by
algebraic Bethe ansatz and so it allows for a direct comparison with our SoV.

The following identity follows from direct calculation.%
\begin{equation}
T^{(K^{^{\left( a=0,\alpha \right) }})}(\lambda )|0\rangle
=|0\rangle t_{0}(\lambda )\text{ \ with }t_{0}(\lambda )=
e^{\alpha }a(\lambda )+e^{-\alpha }d(\lambda ),
\end{equation}%
so that in our SoV covector basis the eigenstate $|0\rangle $ admits the
following representation:%
\begin{equation}
\langle h_{1},...,h_{N}|t\rangle =1,
\end{equation}%
i.e the associated $Q$-function is the identity.

Let us now denote with $\mathbb{B}\left( \lambda \right) $ the one parameter
family of commuting operators characterized by:%
\begin{equation}
\langle h_{1},...,h_{N}|\mathbb{B}\left( \lambda \right) =\text{ }%
b_{h_{1},...,h_{N}}(\lambda )\langle h_{1},...,h_{N}|,
\end{equation}%
where we have defined:%
\begin{equation}
b_{h_{1},...,h_{N}}(\lambda )=\prod_{n=1}^{N}(\lambda -\xi
_{n}^{(h_{n})}),
\end{equation}%
then the following corollary holds:

\begin{corollary}
Let us assume that the condition $(\ref{Inhomog-cond})$ is satisfied and
that the $K^{^{\left( a=0,\alpha \right) }}\neq xI$, for any $%
x\in \mathbb{C}$, then for any $\mathsf{M}\leq N$ taken the generic 
$t(\lambda )\in \Sigma _{T^{(K)}}^{\left( N-2\mathsf{M}\right) }$
the unique associated eigenvector $|t\rangle $ admits the following ABA
representation:%
\begin{equation}
|t\rangle =(-1)^{\mathsf{NM}}\prod_{n=1}^{\mathsf{M}}\mathbb{B}\left(
\lambda _{n}\right) |0\rangle ,
\end{equation}%
where $\{\lambda _{1},...,\lambda _{\mathsf{M}}\}$ are the Bethe roots, i.e.
the zero of the associated $Q$-function.
\end{corollary}

\begin{proof}
Taken the generic $t(\lambda )\in \Sigma _{T^{(K)}}^{\left( N-2%
\mathsf{M}\right) }$ the unique associated eigenvector $|t\rangle $ admits
the following SoV representation:%
\begin{equation}
\langle h_{1},...,h_{N}|t\rangle =\prod_{n=1}^{N}Q_{t}(\xi
_{n}^{(h_{n})}),
\end{equation}%
while by the definition of $\mathbb{B}\left( \lambda \right) $ it holds:%
\begin{equation}
\langle h_{1},...,h_{N}|(-1)^{\mathsf{NM}}\prod_{n=1}^{\mathsf{M}}%
\mathbb{B}\left( \lambda _{n}\right) |0\rangle =(-1)^{\mathsf{NM}%
}\prod_{n=1}^{\mathsf{M}}b_{h_{1},...,h_{N}}\left( \lambda
_{n}\right) ,
\end{equation}%
from which our statement follows being:%
\begin{equation}
(-1)^{\mathsf{NM}}\prod_{n=1}^{\mathsf{M}}b_{h_{1},...,h_{N}}\left(
\lambda _{n}\right) =\prod_{n=1}^{N}Q_{t}(\xi _{n}^{(h_{n})}).
\end{equation}
\end{proof}

\section{The quasi-periodic $Y(gl_3)$ fundamental model}
We consider now the Yang-Baxter algebra associated to the rational $gl_3$
R-matrix:%
\begin{equation}
R_{a,b}(\lambda )=\lambda I_{a,b}+\eta \mathbb{P}_{a,b}=\left( 
\begin{array}{ccc}
a_{1}(\lambda ) & b_{1} & b_{2} \\ 
c_{1} & a_{2}(\lambda ) & b_{3} \\ 
c_{2} & c_{3} & a_{3}(\lambda )%
\end{array}%
\right) \in \End(V_{a}\otimes V_{b}),
\end{equation}%
where $V_{a}\cong V_{b}\cong \mathbb{C}^{3}$ and we have defined:%
\begin{align}
& a_{j}(\lambda )\left. =\right. \left( 
\begin{array}{ccc}
\lambda +\eta \delta _{j,1} & 0 & 0 \\ 
0 & \lambda +\eta \delta _{j,2} & 0 \\ 
0 & 0 & \lambda +\eta \delta _{j,3}%
\end{array}%
\right) ,\text{ \ \ }\forall j\in \{1,2,3\},  \notag \\
& b_{1}\left. =\right. \left( 
\begin{array}{ccc}
0 & 0 & 0 \\ 
\eta & 0 & 0 \\ 
0 & 0 & 0%
\end{array}%
\right) ,\text{ \ }b_{2}\left. =\right. \left( 
\begin{array}{ccc}
0 & 0 & 0 \\ 
0 & 0 & 0 \\ 
\eta & 0 & 0%
\end{array}%
\right) ,\text{ \ }b_{3}\left. =\right. \left( 
\begin{array}{ccc}
0 & 0 & 0 \\ 
0 & 0 & 0 \\ 
0 & \eta & 0%
\end{array}%
\right) ,  \notag \\
& c_{1}\left. =\right. \left( 
\begin{array}{ccc}
0 & \eta & 0 \\ 
0 & 0 & 0 \\ 
0 & 0 & 0%
\end{array}%
\right) ,\text{ \ }c_{2}\left. =\right. \left( 
\begin{array}{ccc}
0 & 0 & \eta \\ 
0 & 0 & 0 \\ 
0 & 0 & 0%
\end{array}%
\right) ,\text{ \ }c_{3}\left. =\right. \left( 
\begin{array}{ccc}
0 & 0 & 0 \\ 
0 & 0 & \eta \\ 
0 & 0 & 0%
\end{array}%
\right) ,
\end{align}%
which satisfies the Yang-Baxter equation:%
\begin{equation}
R_{12}(\lambda -\mu )R_{13}(\lambda )R_{23}(\mu )=R_{23}(\mu )R_{13}(\lambda
)R_{12}(\lambda -\mu )\ .
\end{equation}%
This R-matrix satisfies the following symmetry properties (scalar Yang-Baxter
equation):%
\begin{equation}
R_{12}(\lambda )K_{1}K_{2}=K_{2}K_{1}R_{12}(\lambda )\in \End%
(V_{1}\otimes V_{2}),
\end{equation}%
where $K\in \End(V)$ is any $3\times 3$ matrix. We can define the following
monodromy matrix: 
\begin{eqnarray}
M_{a}^{(K)}(\lambda ) &=&\left( 
\begin{array}{ccc}
A_{1}^{(K)}(\lambda ) & B_{1}^{(K)}(\lambda ) & B_{2}^{(K)}(\lambda ) \\ 
C_{1}^{(K)}(\lambda ) & A_{2}^{(K)}(\lambda ) & B_{3}^{(K)}(\lambda ) \\ 
C_{2}^{(K)}(\lambda ) & C_{3}^{(K)}(\lambda ) & A_{3}^{(K)}(\lambda )%
\end{array}%
\right) _{a} \\
&\equiv &K_{a}R_{a,N}(\lambda -\xi _{N})\cdots
R_{a,1}(\lambda -\xi _{1})\in \End(V_{a}\otimes \mathcal{H}),
\end{eqnarray}%
where $\mathcal{H}=\bigotimes_{n=1}^{N}V_{n}$. Moreover, in the following we will assume that the inhomogeneity parameters $\xi_j$ satisfy the following conditions:
\begin{equation}
\xi _{a}\neq \xi _{b}+r\eta \text{ \ }\forall a\neq b\in \{1,...,N%
\}\,\,\text{and\thinspace \thinspace }r\in \{-2,-1,0,1,2\}.  \label{Inhomog-cond-gl3}
\end{equation}%

\subsection{First fundamental properties of the transfer matrices}

Let us first recall some basic properties of the transfer matrices associated to this higher rank case  summarized in the following two propositions (see \cite{Skl96}).
\begin{proposition}
The transfer matrices,%
\begin{equation}
T_{1}^{\left( K\right) }(\lambda )\equiv \text{tr}_{a}M_{a}^{(K)}(\lambda ),%
\text{ \ \ \ }T_{2}^{\left( K\right) }(\lambda )\equiv \text{tr}%
_{a}U_{a}^{(K)}(\lambda )\ ,
\end{equation}%
where,%
\begin{equation}
U_{c}^{(K)}(\lambda )^t\equiv 3\text{tr}_{ab}P_{abc}^{-}M_{a}^{(K)}(\lambda
)M_{b}^{(K)}(\lambda +\eta ),
\end{equation}%
defines two one parameter families of commuting operators:%
\begin{equation}
\left[ T_{1}^{\left( K\right) }(\lambda ),T_{1}^{\left( K\right) }(\mu )%
\right] =\left[ T_{1}^{\left( K\right) }(\lambda ),T_{2}^{\left( K\right)
}(\mu )\right] =\left[ T_{2}^{\left( K\right) }(\lambda ),T_{2}^{\left(
K\right) }(\mu )\right] =0.
\end{equation}%
Moreover, the quantum determinant:%
\begin{eqnarray}
q\text{-det}M^{(K)}(\lambda ) \cdot \mathbf{1}_a &\equiv &\text{tr}%
_{123} (P_{123}^{-}M_{1}^{(K)}(\lambda )M_{2}^{(K)}(\lambda +\eta
)M_{3}^{(K)}(\lambda +2\eta ))\cdot \mathbf{1}_a  \label{Bound-q-detU_1} \\
&=&\left( M_{a}^{(K)}(\lambda )\right) ^{t_{a}}U_{a}^{(K)}(\lambda +\eta
)=U_{a}^{(K)}(\lambda +\eta )\left( M_{a}^{(K)}(\lambda )\right) ^{t_{a}} \\
&=&\left( U_{a}^{(K)}(\lambda )\right) ^{t_{a}}M_{a}^{(K)}(\lambda +2\eta
)=M_{a}^{(K)}(\lambda +2\eta )\left( U_{a}^{(K)}(\lambda )\right) ^{t_{a}},
\end{eqnarray}%
is a central element of the algebra, i.e.%
\begin{equation}
\lbrack q\text{-det}M^{(K)}(\lambda ),M_{a}^{(K)}(\mu )]=0.
\end{equation}
\end{proposition}
Moreover, it holds

\begin{proposition}
The quantum spectral invariants have the following polynomial form:\\

i) $T_{1}^{\left( K\right) }(\lambda )$ is a degree $N$ polynomial
in $\lambda $ with the following central asymptotic:%
\begin{equation*}
\lim_{\lambda \rightarrow \infty }\lambda ^{-N}T_{1}^{\left(
K\right) }(\lambda )=\text{tr}K,
\end{equation*}%

ii) $T_{2}^{\left( K\right) }(\lambda )$ is a degree $2N$
polynomial in $\lambda $ with the following $N$ central zeros and
asymptotic:%
\begin{equation}
\text{ }T_{2}^{\left( K\right) }(\xi _{a})=0\text{ \ }\forall a\in \{1,...,%
N\}\text{,\ }\lim_{\lambda \rightarrow \infty }\lambda ^{-2\mathsf{N%
}}T_{2}^{\left( K\right) }(\lambda )=\frac{\left( \text{tr}K\right) ^{2}-%
\text{tr}K^{2}}{2},
\end{equation}%

iii) the quantum determinant reads:%
\begin{equation}
q\text{-det}M^{(K)}(\lambda )=\text{det}K\text{ }\prod_{b=1}^{N%
}(\lambda -\xi _{b})(\lambda +\eta -\xi _{b})(\lambda +3\eta -\xi _{b}).
\end{equation}%
Moreover, the following fusion identities holds:%
\begin{eqnarray}
T_{1}^{\left( K\right) }(\xi _{a})T_{2}^{\left( K\right) }(\xi _{a}-2\eta )
&=&q\text{-det}M^{(K)}(\xi _{a}-2\eta ), \\
T_{1}^{\left( K\right) }(\xi _{a}-\eta )T_{1}^{\left( K\right) }(\xi _{a})
&=&T_{2}^{\left( K\right) }(\xi _{a}-\eta ).
\end{eqnarray}
\end{proposition}

\begin{proof}
These trivial roots of $T_{2}^{(K)}$ and fusion identities can be obtained from the general fusion properties of the transfer matrices in the various representations \cite{KulRS81,KulR82}   when computed in the special points associated to the inhomogeneities and by
direct calculation using the reduction of the $R$-matrix to permutation and
projections operators.
\end{proof}

Let us introduce the functions%
\begin{eqnarray}
g_{a,\mathbf{h}}(\lambda ) &=&\prod_{b\neq a,b=1}^{N}\frac{\lambda
-\xi _{b}^{(h)}}{\xi _{a}^{(h)}-\xi _{b}^{(h)}},\text{ \ \ }\xi
_{b}^{(h)}=\xi _{b}-h\eta , \\
f_{a,\mathbf{h}}(\lambda ) &=&g_{a,\mathbf{h}}(\lambda )\prod_{b=1}^{\mathsf{%
N}}\frac{\lambda -\xi _{b}}{\xi _{a}^{(h_a)}-\xi _{b}},
\end{eqnarray}%
and%
\begin{equation}
T_{2,\mathbf{h}}^{(K,\infty )}(\lambda )=\frac{\left( \text{tr}K\right) ^{2}-%
\text{tr}K^{2}}{2}\prod_{n=1}^{N}(\lambda -\xi _{n})(\lambda -\xi
_{n}^{(h_{n})}),
\end{equation}%
then the following corollary holds.

\begin{corollary}
The transfer matrix $T_{2}^{\left( K\right) }(\lambda )$ is completely
characterized in terms of $T_{1}^{\left( K\right) }(\lambda )$ by the fusion
equations, and the following interpolation formula holds:%
\begin{equation}
T_{2}^{\left( K\right) }(\lambda )=T_{2,\mathbf{h}=\mathbf{1}}^{(K,\infty
)}(\lambda )+\sum_{a=1}^{N}f_{a,\mathbf{h}=\mathbf{1}}(\lambda
)T_{1}^{\left( K\right) }(\xi _{a}-\eta )T_{1}^{\left( K\right) }(\xi _{a}).
\end{equation}
\end{corollary}

\begin{proof}
The known central zeros and asymptotic behavior imply the above interpolation formula
once we use the fusion equations to write $T_{2}^{\left( K\right) }(\xi
_{a}-\eta )$.
\end{proof}

\subsection{The new SoV covector basis}

Our general construction of the SoV covector basis applies
in particular to the fundamental representation of the $gl_3$ rational
Yang-Baxter algebra.

Let $K$ be a $3\times 3$ $w$-simple matrix and let us denote by $K_{J}$ an upper-triangular
Jordan form of the matrix $K$ and $W$ the invertible matrix defining the
change of basis:%
\begin{equation}
K=W_{K}K_{J}W_{K}^{-1}\text{ \ with \ }K_{J}=\left( 
\begin{array}{ccc}
\mathsf{k}_{0} & \mathsf{y}_{1} & 0 \\ 
0 & \mathsf{k}_{1} & \mathsf{y}_{2} \\ 
0 & 0 & \mathsf{k}_{2}%
\end{array}%
\right) .
\end{equation}%
The requirement that $K$ is $w$-simple implies that we can have only the following
three possible cases (up to trivial permutations of the basis vectors):%
\begin{align}
i)\text{ }\mathsf{k}_{i}& \neq \mathsf{k}_{j}\ \forall i,j\in \{0,1,2\},%
\text{ }\mathsf{y}_{1}=\mathsf{y}_{2}=0, \\
ii)\text{ }\mathsf{k}_{0}& =\mathsf{k}_{1}\neq \mathsf{k}_{2},\text{ }%
\mathsf{y}_{1}=1,\text{ }\mathsf{y}_{2}=0, \\
ii)\text{ }\mathsf{k}_{0}& =\mathsf{k}_{1}=\mathsf{k}_{2},\text{ }\mathsf{y}%
_{1}=1,\text{ }\mathsf{y}_{2}=1,
\end{align}%
then the following theorem holds:

\begin{proposition}
Let $K$ be a $3\times 3$ $w$-simple matrix, then for almost any choice of $%
\langle S|$ and of the inhomogeneities under the condition $(\ref%
{Inhomog-cond-gl3})$, the following set of covectors:%
\begin{equation}
\langle h_{1},...,h_{N}|\equiv \langle S|\prod_{n=1}^{N%
}(T_{1}^{(K)}(\xi _{n}))^{h_{n}}\text{ \ for any }\{h_{1},...,h_{N%
}\}\in \{0,1,2\}^{\otimes N},
\end{equation}%
forms a covector basis of $\cal{H}^*$. In particular, we can take the
state $\langle S|$ of the following tensor product form:%
\begin{equation}
\langle S|=\bigotimes_{a=1}^{N}(x,y,z)_{a}\Gamma _{W}^{-1},\text{ \
\ }\Gamma _{W}=\bigotimes_{a=1}^{N}W_{K,a}
\end{equation}%
simply asking $x\,y\,z\neq 0$ in the case i), $x\,\,z\neq 0$ in the case
ii), $x\,\neq 0$ in the case iii).
\end{proposition}

\begin{proof}
As shown in the general Proposition \ref{gln-basis}, the fact that the set of covectors is a
covector basis of $\cal{H}^*$ reduces to the requirement that the
covectors:%
\begin{equation}
(x,y,z)_{a}W^{-1},(x,y,z)_{a}W^{-1}K,(x,y,z)_{a}W^{-1}K^{2},
\end{equation}%
or equivalently:%
\begin{equation}
(x,y,z)_{a},(x,y,z)_{a}K_{J},(x,y,z)_{a}K_{J}^{2},
\end{equation}%
form a basis in $V_{a}\cong \mathbb{C}^{3}$, that is that the following
determinant is non-zero:%
\begin{equation}
\text{det}||\left( (x,y,z)K_{J}^{i-1}\ket{e_{j}}\right) _{i,j\in
\{1,2,3\}}||=\left\{ 
\begin{array}{l}
-xyzV(\mathsf{k}_{0},\mathsf{k}_{1},\mathsf{k}_{2})\text{ \ in the case i)}
\\ 
x^{2}zV^{2}(\mathsf{k}_{0},\mathsf{k}_{2})\text{ in the case ii)} \\ 
x^{3}\text{ in the case iii)}%
\end{array}%
\right. \ ,
\end{equation}%
where $\ket{e_{j}}$ is the canonical basis in $V_{a}\cong \mathbb{C}^{3}$ after the change of basis induced by $W_K$. It leads to the given requirements on the components $x,y,z\in \mathbb{C}$
of the three dimensional covector.
\end{proof}

\subsection{Transfer matrix spectrum in our SoV scheme}

The following characterization of the transfer matrix spectrum holds:

\begin{theorem}
Under the same assumptions ensuring that the set of SoV covector form a
basis, then the spectrum of $T_{1}^{(K)}(\lambda )$ is characterized by:%
\begin{equation}
\Sigma _{T^{(K)}}=\left\{ t_{1}(\lambda ):t_{1}(\lambda )=\text{tr\thinspace 
}K\text{ }\prod_{a=1}^{N}(\lambda -\xi _{a})+\sum_{a=1}^{N%
}g_{a,\mathbf{h}=0}(\lambda )x_{a},\text{ \ }\forall \{x_{1},...,x_{\mathsf{N%
}}\}\in \Sigma _{T}\right\} ,  \label{SET-T5}
\end{equation}%
$\Sigma _{T}$ is the set of solutions to the following inhomogeneous system
of $N$ cubic equations:%
\begin{equation}
x_{a}[T_{2,\mathbf{h}=\mathbf{1}}^{(K,\infty )}(\xi _{a}-2\eta )+\sum_{n=1}^{%
N}f_{n,\mathbf{h}=\mathbf{1}}(\xi _{a}-2\eta )t_{1}(\xi _{n}-\eta
)x_{n}]\left. =\right. q\text{-det}M^{(K)}(\xi _{a}-2\eta ),
\end{equation}%
in $N$ unknown $\{x_{1},...,x_{N}\}$. Moreover, $%
T_{1}^{(K)}(\lambda ,\{\xi \})$ is $w$-simple  and for any $%
t_{1}(\lambda )\in \Sigma _{T^{(K)}}$ the associated unique (up-to
normalization) eigenvector $|t\rangle $ has the following wave-function in
the left SoV basis:%
\begin{equation}
\langle h_{1},...,h_{N}|t\rangle =\prod_{n=1}^{N%
}t_{1}^{h_{n}}(\xi _{n}).  \label{SoV-Ch-T-eigenV}
\end{equation}
\end{theorem}

\begin{proof}
Let us start observing that the inhomogeneous system of $N$ cubic
equations in $N$ unknown $\{x_{1},...,x_{N}\}$ is nothing
else but the rewriting of the transfer matrix fusion equations:%
\begin{equation}
t_{1}(\xi _{a})t_{2}(\xi _{a}-2\eta )=q\text{-det}M^{(K)}(\xi _{a}-2\eta ),%
\text{ }\forall a\in \{1,...,N\},  \label{scalar-fusion-1}
\end{equation}%
for the eigenvalues of the transfer matrices $T_{1,2}^{(K)}(\lambda )$. So
that this system has to be satisfied and the associated eigenvector $%
|t\rangle $ admits the given characterization in the covector SoV basis.

So we are left with the proof of the reverse statement, i.e. that any
polynomial $t_{1}(\lambda )$ of the above form satisfying this system is an
eigenvalue of the transfer matrices and this is proven by showing that the
vector $|t\rangle $ characterized by $(\ref{SoV-Ch-T-eigenV})$ is a transfer
matrix eigenstate, i.e. we have to show:%
\begin{equation}
\langle h_{1},...,h_{N}|T_{1}^{(K)}(\lambda )|t\rangle
=t_{1}(\lambda )\langle h_{1},...,h_{N}|t\rangle ,\text{ }\forall
\{h_{1},...,h_{N}\}\in \{0,1,2\}^{\otimes N}.
\label{Eigen-cond-sl3}
\end{equation}%
Let us start observing that by using the interpolation formula:%
\begin{equation}
T_{1}^{(K)}(\lambda )=T_{1,\mathbf{h}}^{(K,\infty )}(\lambda )+\sum_{a=1}^{%
N}g_{a,\mathbf{h}}(\lambda )T_{1}^{(K)}(\xi _{a}^{(h_{a})}),\text{ }%
T_{1,\mathbf{h}}^{(K,\infty )}(\lambda )=\text{tr\thinspace }K\text{ }%
\prod_{a=1}^{N}(\lambda -\xi _{a}^{(h_{a})}),
\end{equation}%
the above statement is proven once we prove the identity in the points $\xi
_{a}^{(h_{a})}$ for any $a\in \{1,...,N\}$. Let be $h_{a}=0,1$ and $%
h_{b}\in \{0,1,2\}$ for any $b\in \{1,...,N\}\backslash a$, then we
have the following identities:%
\begin{align}
\langle h_{1},...,h_{N}|T_{1}^{(K)}(\xi _{a})|t\rangle & =\langle
h_{1},...,h_{a}+1,...,h_{N}|t\rangle  \notag \\
& =t_{1}(\xi _{a})\langle h_{1},...,h_{a},...,h_{N}|t\rangle ,
\label{Id-step1}
\end{align}%
as a direct consequence of the definition of the covector SoV basis and of
the state $|t\rangle $. So that we are left with the proof of the statement
in the case $h_{a}=2$. In this case we want to prove that it holds:%
\begin{equation}
\langle h_{1},...,h_{N}|T_{1}^{(K)}(\xi _{a}-\eta )|t\rangle
=t_{1}(\xi _{a}-\eta )\langle h_{1},...,h_{a},...,h_{N}|t\rangle ,
\label{Id-step2}
\end{equation}%
the proof is done by induction on the number $R$ of zeros contained in $%
\{h_{1},...,h_{N}\}\in \{0,1,2\}^{\otimes N}$. Let us
observe that by the fusion identities it holds:%
\begin{equation}
\langle h_{1},...,h_{a}=2,...,h_{N}|T_{1}^{(K)}(\xi _{a}-\eta
)|t\rangle =\langle h_{1},...,h_{a}=1,...,h_{N}|T_{2}^{(K)}(\xi
_{a}-\eta )|t\rangle .
\end{equation}%
Let us start to prove our identity for $R$ $=0$. We can use the
following interpolation formula:%
\begin{equation}
T_{2}^{\left( K\right) }(\xi _{a}-\eta )=T_{2,\mathbf{h}}^{(K,\infty )}(\xi
_{a}-\eta )+\sum_{n=1}^{N}f_{n,\mathbf{h}=\mathbf{2}}(\xi _{a}-\eta
)T_{2}^{\left( K\right) }(\xi _{n}-2\eta )
\end{equation}%
so that:%
\begin{align}
\langle h_{1},...,h_{a}^{\prime}&=1,...,h_{N}|T_{2}^{(K)}(\xi
_{a}-\eta )|t\rangle =T_{2,\mathbf{h}}^{(K,\infty )}(\xi _{a}-\eta )\langle
h_{1},...,h_{a}^{\prime },...,h_{N}|t\rangle \\
&+\sum_{n=1}^{N}f_{n,\mathbf{h}=\mathbf{2}}(\xi _{a}-\eta )\langle
h_{1},...,h_{a}^{\prime },...,h_{N}|T_{2}^{\left( K\right) }(\xi
_{n}-2\eta )|t\rangle .
\end{align}%
Then, by the assumption $R=0$ and by the fusion identity, it follows:%
\begin{align}
\langle h_{1},...,h_{a}^{\prime },...,h_{N}|T_{2}^{(K)}(\xi
_{a}-\eta )|t\rangle & =T_{2,\mathbf{h}=\mathbf{2}}^{(K,\infty )}(\xi
_{a}-\eta )\langle h_{1},...,h_{a}^{\prime },...,h_{N}|t\rangle 
\notag \\
& +\sum_{n=1}^{N}q\text{-det}M^{(K)}(\xi _{n}-2\eta )f_{n,\mathbf{h}%
=\mathbf{2}}(\xi _{a}-\eta )\langle h_{1},...,h_{n}^{\prime \prime },...,h_{%
N}|t\rangle ,
\end{align}%
where $h_{n}^{\prime \prime }=h_{n}-1$ for $n\neq a$ and $h_{a}^{\prime
\prime }=h_{a}^{\prime }-1=0$. Let us now define the function:%
\begin{equation}
t_{2}(\lambda )=T_{2,\mathbf{h}=\mathbf{1}}^{(K,\infty )}(\lambda
)+\sum_{n=1}^{N}f_{n,\mathbf{h}=\mathbf{1}}(\lambda )t_{1}(\xi
_{n}-\eta )t_{1}(\xi _{n}),  \label{function t2-def}
\end{equation}%
then by its definition it satisfies the equations:%
\begin{eqnarray}
t_{2}(\xi _{n}-\eta ) &=&t_{1}(\xi _{n}-\eta )t_{1}(\xi _{n}),\text{ }%
\forall n\in \{1,...,N\}, \\
t_{1}(\xi _{a})t_{2}(\xi _{a}-2\eta ) &=&q\text{-det}M^{(K)}(\xi _{a}-2\eta
),\text{ }\forall n\in \{1,...,N\},
\end{eqnarray}%
while the second (quantum determinant) equation is satisfied by the
definition of the function $t_{1}(\lambda )$. Using the function $%
t_{2}(\lambda )$ and these identities we get:%
\begin{align}
& \langle h_{1},...,h_{a}^{\prime },...,h_{N}|T_{2}^{(K)}(\xi
_{a}-\eta )|t\rangle \left. =\right.  \notag \\
& =\left( T_{2,\mathbf{h}=\mathbf{2}}^{(K,\infty )}(\xi _{a}-\eta
)+\sum_{n=1}^{N}t_{2}(\xi _{n}-2\eta )f_{n,\mathbf{h}=\mathbf{2}%
}(\xi _{a}-\eta )\right) \langle h_{1},...,h_{a}^{\prime },...,h_{N%
}|t\rangle , \\
& =t_{2}(\xi _{n}-\eta )\langle h_{1},...,h_{a}^{\prime },...,h_{N%
}|t\rangle \\
& =t_{1}(\xi _{n}-\eta )\langle h_{1},...,h_{a}=2,...,h_{N%
}|t\rangle ,
\end{align}%
where we have used the interpolation formula:%
\begin{equation}
t_{2}(\xi _{a}-\eta )=T_{2,\mathbf{h}=\mathbf{2}}^{(K,\infty )}(\xi
_{a}-\eta )+\sum_{n=1}^{N}t_{2}(\xi _{n}-2\eta )f_{n,\mathbf{h}=%
\mathbf{2}}(\xi _{a}-\eta ),
\end{equation}%
i.e. we have shown our identity $(\ref{Id-step2})$ for $R=0$. Let us now
make the proof by induction assuming that it holds for generic $%
\{h_{1},...,h_{N}\}\in \{0,1,2\}^{\otimes N}$ containing $%
R-1$ zeros. Then we have to show the same property for generic $\{h_{1},...,h_{N}\}\in
\{0,1,2\}^{\otimes N}$ containing $R$ zeros. Let us fix the generic 
$\{h_{1},...,h_{N}\}\in \{0,1,2\}^{\otimes N}$ with $%
h_{a}=2$ and let us denote with $\pi $ a permutation of $\{1,...,N%
\} $ such that:%
\begin{equation}
\left. 
\begin{array}{l}
h_{\pi (i)}=0,\text{ }\forall i\in \{1,...,R\}, \\ 
h_{\pi (i)}=1,\text{ }\forall i\in \{R+1,...,R+S\}, \\ 
h_{\pi (i)}=2,\text{ }\forall i\in \{R+S+1,...,N\},%
\end{array}%
\right.
\end{equation}%
with $a=\pi (R+S+1)$. Let us use now the following interpolation formula:%
\begin{equation}
T_{2}^{\left( K\right) }(\xi _{a}-\eta )=T_{2,\mathbf{k}}^{(K,\infty )}(\xi
_{a}-\eta )+\sum_{n=1}^{N}f_{n,\mathbf{k}}(\xi _{a}-\eta
)T_{2}^{\left( K\right) }(\xi _{n}^{\left( k_{n}\right) }),
\end{equation}%
where we have defined $\mathbf{k}$ by:%
\begin{equation}
\left. 
\begin{array}{l}
k_{\pi (i)}=1,\text{ }\forall i\in \{1,...,R\}, \\ 
k_{\pi (i)}=2,\text{ }\forall i\in \{R+1,...,N\},%
\end{array}%
\right.
\end{equation}%
then it holds: 
\begin{align}
\langle h_{1},...,h_{a}^{\prime }& =1,...,h_{N}|T_{2}^{(K)}(\xi
_{a}-\eta )|t\rangle =T_{2,\mathbf{k}}^{(K,\infty )}(\xi _{a}-\eta )\langle
h_{1},...,h_{a}^{\prime },...,h_{N}|t\rangle  \notag \\
& +\sum_{n=1}^{R}f_{\pi (n),\mathbf{k}}(\xi _{a}-\eta )\langle
h_{1},...,h_{a}^{\prime },...,h_{N}|T_{2}^{\left( K\right) }(\xi
_{\pi (n)}-\eta )|t\rangle  \notag \\
& +\sum_{n=R+1}^{N}f_{\pi (n),\mathbf{k}}(\xi _{a}-\eta )\langle
h_{1},...,h_{a}^{\prime },...,h_{N}|T_{2}^{\left( K\right) }(\xi
_{\pi (n)}-2\eta )|t\rangle .
\end{align}%
Using the fusion identity, we get:%
\begin{align}
\langle h_{1},...,h_{a}^{\prime },...,h_{N}|T_{2}^{(K)}(\xi
_{a}-\eta )|t\rangle & =T_{2,\mathbf{k}}^{(K,\infty )}(\xi _{a}-\eta
)\langle h_{1},...,h_{a}^{\prime },...,h_{N}|t\rangle  \notag \\
& +\sum_{n=1}^{R}f_{\pi (n),\mathbf{k}}(\xi _{a}-\eta )\langle
h_{1}^{(n)},...,h_{N}^{(n)}|T_{1}^{\left( K\right) }(\xi _{\pi
(n)}-\eta )|t\rangle  \notag \\
& +\sum_{n=R+1}^{N}q\text{-det}M^{(K)}(\xi _{\pi (n)}-2\eta )f_{\pi
(n),\mathbf{k}}(\xi _{a}-\eta )\langle h_{1}^{(n)},...,h_{N%
}^{(n)}|t\rangle ,  \label{T2-id-2}
\end{align}%
where we have defined:%
\begin{equation}
h_{\pi (m)}^{(n)}=\left\{ 
\begin{array}{l}
h_{\pi (m)}+\theta (R-m)\delta _{m,n}\text{ \ for }n\leq R \\ 
h_{\pi (m)}-\theta (m-(R+1))\delta _{m,n}-\delta _{m,R+S+1}\text{ \ for }%
R+1\leq n%
\end{array}%
\right. .
\end{equation}%
To compute $\langle h_{1}^{(n)},...,h_{N}^{(n)}|T_{1}^{\left(
K\right) }(\xi _{\pi (n)}-\eta )|t\rangle $ for $n\leq R$, we use the
following interpolation formula:%
\begin{equation}
T_{1}^{\left( K\right) }(\xi _{\pi (n)}-\eta )=T_{1,\mathbf{k}^{\prime
}}^{(K,\infty )}(\xi _{\pi (n)}-\eta )+\sum_{a=1}^{N}g_{a,\mathbf{k}%
^{\prime }}(\xi _{\pi (n)}-\eta )T_{1}^{\left( K\right) }(\xi _{a}^{(k_{a}^{\prime
})}),
\end{equation}%
where we have defined:%
\begin{equation}
k_{\pi (m)}^{\prime }=\left\{ 
\begin{array}{l}
0\text{ \ for }m\leq R+S+1 \\ 
1\text{ \ for }R+S+2\leq m%
\end{array}%
\right. ,
\end{equation}%
which gives:%
\begin{align}
\langle h_{1}^{(n)},...,h_{N}^{(n)}|T_{1}^{\left( K\right) }(\xi
_{\pi (n)}-\eta )|t\rangle & =T_{1,\mathbf{k}^{\prime }}^{(K,\infty )}(\xi
_{\pi (n)}-\eta )\langle h_{1}^{(n)},...,h_{N}^{(n)}|t\rangle 
\notag \\
& +\sum_{a=1}^{R+S+1}g_{\pi (a),\mathbf{k}^{\prime }}(\xi _{\pi (n)}-\eta
)\langle h_{1}^{(n)},...,h_{N}^{(n)}|T_{1}^{\left( K\right) }(\xi
_{\pi (a)})|t\rangle  \notag \\
& +\sum_{a=R+S+2}^{N}g_{\pi (a),\mathbf{k}^{\prime }}(\xi _{\pi
(n)}-\eta )\langle h_{1}^{(n)},...,h_{N}^{(n)}|T_{1}^{\left(
K\right) }(\xi _{\pi (a)}-\eta )|t\rangle , \label{Iter-step-2}
\end{align}%
which becomes:%
\begin{align}
\langle h_{1}^{(n)},...,h_{N}^{(n)}|T_{1}^{\left( K\right) }(\xi
_{\pi (n)}-\eta )|t\rangle & =T_{1,\mathbf{k}^{\prime }}^{(K,\infty )}(\xi
_{\pi (n)}-\eta )\langle h_{1}^{(n)},...,h_{N}^{(n)}|t\rangle 
\notag \\
& +\sum_{a=1}^{R+S+1}g_{\pi (a),\mathbf{k}^{\prime }}(\xi _{\pi (n)}-\eta
)t_{1}(\xi _{\pi (a)})\langle h_{1}^{(n)},...,h_{N}^{(n)}|t\rangle 
\notag \\
& +\sum_{a=R+S+2}^{N}g_{\pi (a),\mathbf{k}^{\prime }}(\xi _{\pi
(n)}-\eta )t_{1}(\xi _{\pi (a)}-\eta )\langle h_{1}^{(n)},...,h_{N%
}^{(n)}|t\rangle ,
\end{align}%
where in the second line we have used the identity $(\ref{Id-step1})$ while
in the third line the identity $(\ref{Id-step2})$, which holds by assumption
being $R-1$ the number of zeros in $\{h_{1}^{(n)},...,h_{N}^{(n)}\}$%
. So that we have shown for any $n\leq R$:%
\begin{equation}
\langle h_{1}^{(n)},...,h_{N}^{(n)}|T_{1}^{\left( K\right) }(\xi
_{\pi (n)}-\eta )|t\rangle =t_{1}(\xi _{\pi (n)}-\eta )\langle
h_{1}^{(n)},...,h_{N}^{(n)}|t\rangle ,
\end{equation}%
and substituting it in \rf{Iter-step-2} we get:%
\begin{align}
\langle h_{1},...,h_{a}^{\prime },...,h_{N}|T_{2}^{(K)}(\xi
_{a}-\eta )|t\rangle & =T_{2,\mathbf{k}}^{(K,\infty )}(\xi _{a}-\eta
)\langle h_{1},...,h_{a}^{\prime },...,h_{N}|t\rangle  \notag \\
& +\sum_{n=1}^{R}t_{1}(\xi _{\pi (n)}-\eta )f_{\pi (n),\mathbf{k}}(\xi
_{a}-\eta )\langle h_{1}^{(n)},...,h_{N}^{(n)}|t\rangle  \notag \\
& +\sum_{n=R+1}^{N}q\text{-det}M^{(K)}(\xi _{\pi (n)}-2\eta )f_{\pi
(n),\mathbf{k}}(\xi _{a}-\eta )\langle h_{1}^{(n)},...,h_{N%
}^{(n)}|t\rangle ,
\end{align}%
and so $\langle h_{1},...,h_{a}^{\prime },...,h_{N}|T_{2}^{(K)}(\xi
_{a}-\eta )|t\rangle $ reads:%
\begin{align}
& \left( T_{2,\mathbf{k}}^{(K,\infty )}(\xi _{a}-\eta
)+\sum_{n=1}^{R}t_{1}(\xi _{\pi (n)})t_{1}(\xi _{\pi (n)}-\eta )f_{\pi (n),%
\mathbf{k}}(\xi _{a}-\eta )+\sum_{n=R+1}^{N}t_{2}(\xi _{\pi
(n)}-2\eta )f_{\pi (n),\mathbf{k}}(\xi _{a}-\eta )\right)  \notag \\
& \times \langle h_{1},...,h_{a}^{\prime },...,h_{N}|t\rangle 
\notag \\
& =t_{2}(\xi _{a}-\eta )\langle h_{1},...,h_{a}^{\prime }=1,...,h_{N%
}|t\rangle =t_{1}(\xi _{a}-\eta )\langle h_{1},...,h_{a}=2,...,h_{N%
}|t\rangle ,
\end{align}%
i.e. we have proven our formula $(\ref{Id-step2})$. Finally, taking 
generic $\{h_{1},...,h_{N}\}\in \{0,1,2\}^{\otimes N}$
with:%
\begin{equation}
\left. 
\begin{array}{l}
h_{\pi (i)}=0,\text{ }\forall i\in \{1,...,R\}, \\ 
h_{\pi (i)}=1,\text{ }\forall i\in \{R+1,...,R+S\}, \\ 
h_{\pi (i)}=2,\text{ }\forall i\in \{R+S+1,...,N\},%
\end{array}%
\right.
\end{equation}%
and by using the interpolation formula:%
\begin{equation}
T_{1}^{\left( K\right) }(\lambda )=T_{1,\mathbf{p}}^{(K,\infty )}(\lambda
)+\sum_{n=1}^{N}g_{n,\mathbf{p}}(\lambda )T_{1}^{\left( K\right)
}(\xi _{n}^{\left( p_{n}\right) }),
\end{equation}%
where we have defined $\mathbf{p}$ by:%
\begin{equation}
\left. 
\begin{array}{l}
p_{\pi (i)}=0,\text{ }\forall i\in \{1,...,R+S\}, \\ 
p_{\pi (i)}=1,\text{ }\forall i\in \{R+S+1,...,N\},%
\end{array}%
\right.
\end{equation}%
we get: 
\begin{align}
\langle h_{1},...,h_{N}|T_{1}^{\left( K\right) }(\lambda )|t\rangle
& =T_{1,\mathbf{p}}^{(K,\infty )}(\lambda )\langle h_{1},...,h_{N%
}|t\rangle  \notag \\
& +\sum_{n=1}^{R}g_{\pi (n),\mathbf{p}}(\lambda )\langle h_{1},...,h_{%
N}|T_{1}^{\left( K\right) }(\xi _{\pi (n)})|t\rangle  \notag \\
& +\sum_{n=R+1}^{N}g_{\pi (n),\mathbf{p}}(\lambda )\langle
h_{1},...,h_{N}|T_{1}^{\left( K\right) }(\xi _{\pi (n)}-\eta
)|t\rangle .
\end{align}%
Then, using in the second line the identity $(\ref{Id-step1})$ and $(\ref%
{Id-step2})$ in the third line we get:%
\begin{align}
\langle h_{1},...,h_{N}|T_{1}^{\left( K\right) }(\lambda )|t\rangle
&=\left( T_{1,\mathbf{p}}^{(K,\infty )}(\lambda )+\sum_{n=1}^{N%
}g_{\pi (n),\mathbf{p}}(\lambda )t_{1}(\xi _{\pi (n)}^{(p_{\pi (n)})}-\eta
)\right) \langle h_{1},...,h_{N}|t\rangle \\
&=t_{1}(\lambda )\langle h_{1},...,h_{N}|t\rangle
\end{align}%
which complete the proof of our theorem. Note that clearly, from the
following interpolation formula:%
\begin{equation}
T_{2}^{\left( K\right) }(\lambda )=T_{2,\mathbf{h}=\mathbf{1}}^{(K,\infty
)}(\lambda )+\sum_{a=1}^{N}f_{a,\mathbf{h}=\mathbf{1}}(\lambda
)T_{1}^{\left( K\right) }(\xi _{a}-\eta )T_{1}^{\left( K\right) }(\xi _{a}),
\end{equation}
it also holds:%
\begin{eqnarray}
\langle h_{1},...,h_{N}|T_{2}^{\left( K\right) }(\lambda )|t\rangle
&=&\left( T_{2,\mathbf{h}=\mathbf{1}}^{(K,\infty )}(\lambda )+\sum_{a=1}^{%
N}f_{a,\mathbf{h}=\mathbf{1}}(\lambda )t_{1}(\xi _{a}-\eta
)t_{1}(\xi _{a})\right) \langle h_{1},...,h_{N}|t\rangle \nonumber \\
&=&t_{2}(\lambda )\langle h_{1},...,h_{N}|t\rangle .
\end{eqnarray}
\end{proof}
Let us make some elementary remark about this quite lengthy proof. In fact the main idea behind it is to use the fusion relations until we get the quantum determinant which acts trivially on any covector. While doing so we use interpolation formulae for the transfer matrix. Then one can note that the same fusion relations and interpolation formulae are true for the eigenvalues of the transfer matrices, hence giving the possibility to reverse the process and to reconstruct it in the necessary points. 
From the above discrete characterization of the transfer matrix spectrum in
our SoV basis we can prove the following quantum spectral curve functional
reformulation.

\begin{theorem}
Let us assume that the twist matrix $K$ is $w$-simple and it has at least one nonzero eigenvalues then the entire functions $t_{1}(\lambda )$  is a $T_{1}^{\left( K\right) }(\lambda )$ transfer matrix eigenvalue if and only
if there exists a unique polynomial:%
\begin{equation}
\varphi _{t}(\lambda )=\prod_{a=1}^{\mathsf{M}}(\lambda -\lambda _{a})\text{%
\ \ with }\mathsf{M}\leq N\text{ and }\lambda _{a}\neq \xi _{n}%
\text{ }\forall (a,n)\in \{1,...,\mathsf{M}\}\times \{1,...,N\},
\label{Phi-form}
\end{equation}%
such that $t_{1}(\lambda )$,%
\begin{equation}
t_{2}(\lambda )=T_{2,\mathbf{h}=\mathbf{1}}^{(K,\infty )}(\lambda
)+\sum_{n=1}^{N}f_{n,\mathbf{h}=\mathbf{1}}(\lambda )t_{1}(\xi
_{n}-\eta )t_{1}(\xi _{n}),
\end{equation}%
and $\varphi _{t}(\lambda )$ are solutions of the following quantum spectral
curve:%
\begin{eqnarray}
&&\alpha (\lambda )\varphi _{t}(\lambda -3\eta )-\beta (\lambda
)t_{1}(\lambda -2\eta )\varphi _{t}(\lambda -2\eta )  \notag \\
&&+\gamma (\lambda )t_{2}(\lambda -2\eta )\varphi _{t}(\lambda -\eta )-q%
\text{-det}M_{a}^{(K)}(\lambda -2\eta )\varphi _{t}(\lambda )\left. =\right.
0
\end{eqnarray}%
where:%
\begin{eqnarray}
\alpha (\lambda ) &=&\gamma (\lambda )\gamma (\lambda -\eta )\gamma (\lambda
-2\eta ), \\
\beta (\lambda ) &=&\gamma (\lambda )\gamma (\lambda -\eta ), \\
\gamma (\lambda ) &=&\gamma _{0}\prod_{a=1}^{N}(\lambda +\eta -\xi
_{a}),
\end{eqnarray}%
and $\gamma _{0}$ is a nonzero solution of the characteristic equation:%
\begin{equation}
\gamma _{0}^{3}-\gamma _{0}^{2}\text{ tr\thinspace }K+\gamma _{0}\frac{%
\left( \text{tr}K\right) ^{2}-\text{tr}K^{2}}{2}=\text{det\thinspace }K\text{%
,}  \label{Ch-Eq-K}
\end{equation}%
i.e. $\gamma _{0}$ is a nonzero eigenvalue of the matrix $K$. Moreover,
up to a normalization the common transfer matrix eigenstate $|t\rangle $
admits the following separate representation:%
\begin{equation}
\langle h_{1},...,h_{N}|t\rangle =\prod_{a=1}^{N}\gamma
^{h_{a}}(\xi _{a})\varphi _{t}^{h_{a}}(\xi _{a}-\eta )\varphi
_{t}^{2-h_{a}}(\xi _{a}).
\end{equation}
\end{theorem}

\begin{proof}
Let us assume that the entire function $t_{1}(\lambda )$ satisfies with the
polynomial $t_{2}(\lambda )$ and $\varphi _{t}(\lambda )$ the functional
equation then it is a degree $N$ polynomial in $\lambda $ with leading coefficient $t_{1,N}$
satisfying the equation:%
\begin{equation}
\gamma _{0}^{3}-\gamma _{0}^{2}\text{ }t_{1,N}+\gamma _{0}\frac{%
\left( \text{tr}K\right) ^{2}-\text{tr}K^{2}}{2}=\text{det\thinspace }K\text{%
,}
\end{equation}%
which being $\gamma _{0}$ an eigenvalue of $K$ implies:%
\begin{equation}
t_{1,N}=\text{tr\thinspace }K.
\end{equation}%
Let us observe that, for $\lambda =\xi _{a}$ it holds:%
\begin{equation}
\alpha (\xi _{a})=\beta (\xi _{a})=0,\text{ }\gamma (\xi _{a})\neq 0,\text{
det}K(\xi _{a}-2\eta )\neq 0,
\end{equation}%
so that the functional equation is reduced in these points to:%
\begin{equation}
\frac{\gamma (\xi _{a})\varphi _{t}(\xi _{a}-\eta )}{\varphi _{t}(\xi _{a})}=%
\frac{\text{det}_{q}M_{a}^{(K)}(\xi _{a}-2\eta )}{t_{2}(\xi _{a}-2\eta )},
\label{Ch-1-Q}
\end{equation}%
while for $\lambda =\xi _{a}+\eta $ it holds:%
\begin{equation}
\alpha (\xi _{a}+\eta )=\text{det}M_{a}^{(K)}(\xi _{a}-\eta )=0,\text{ \ }%
\beta (\xi _{a}+\eta )\neq 0,\text{ }\gamma (\xi _{a}+\eta )\neq 0,
\end{equation}%
so that the functional equation is reduced to:%
\begin{equation}
\frac{\beta (\xi _{a}+\eta )\varphi _{t}(\xi _{a}-\eta )}{\gamma (\xi
_{a}+\eta )\varphi _{t}(\xi _{a})}=\frac{t_{2}(\xi _{a}-\eta )}{t_{1}(\xi
_{a}-\eta )} \ . \label{Ch-2-Q}
\end{equation}%
Finally for $\lambda =\xi _{a}+2\eta $ it holds:%
\begin{equation}
t_{2}(\xi _{a})=\text{det}M_{a}^{(K)}(\xi _{a})=0,\text{ \ }\beta (\xi
_{a}+2\eta )\neq 0,\text{ }\alpha (\xi _{a}+2\eta )\neq 0,
\end{equation}%
so that the functional equation is reduced to:%
\begin{equation}
\frac{\alpha (\xi _{a}+2\eta )\varphi _{t}(\xi _{a}-\eta )}{\beta (\xi
_{a}+2\eta )\varphi _{t}(\xi _{a})}=t_{1}(\xi _{a}).  \label{Ch-3-Q}
\end{equation}%
These identities implies that the following equations are satisfied:%
\begin{eqnarray}
t_{2}(\xi _{n}-\eta ) &=&t_{1}(\xi _{n}-\eta )t_{1}(\xi _{n}),\text{ }%
\forall n\in \{1,...,N\}, \\
t_{1}(\xi _{a})t_{2}(\xi _{a}-2\eta ) &=&q\text{-det}M^{(K)}(\xi _{a}-2\eta
),\text{ }\forall n\in \{1,...,N\},
\end{eqnarray}%
so that by our previous theorem we have that $t_{1}(\lambda )$ and $%
t_{2}(\lambda )$ are eigenvalues of the transfer matrices $T_{1}^{\left(
K\right) }(\lambda )$ and $T_{2}^{\left( K\right) }(\lambda )$,
respectively, associated to the same eigenstate $|t\rangle $.

Let us now prove the reverse statement, i.e. we assume that $t_{1}(\lambda )$
is eigenvalue of the transfer matrix $T_{1}^{\left( K\right) }(\lambda )$
and we want to show that there exists a polynomial $\varphi _{t}(\lambda )$
which satisfies with $t_{1}(\lambda )$ and $t_{2}(\lambda )$ the functional
equation. Here, we characterize $\varphi _{t}(\lambda )$ by imposing that it
satisfies the following set of conditions:%
\begin{equation}
\gamma (\xi _{a})\frac{\varphi _{t}(\xi _{a}-\eta )}{\varphi _{t}(\xi _{a})}%
=t_{1}(\xi _{a}).  \label{Discrete-Ch-Q}
\end{equation}%
The fact that these relations characterize uniquely a polynomial of the form $\left( %
\ref{Phi-form}\right) $ can be shown just following the same steps given in
the Y($gl_2$) case. Let us show that this characterization of $%
\varphi _{t}(\lambda )$ implies that the functional equation is indeed
satisfied. The functional equation is a polynomial in $\lambda $ with maximal degree 4$N$, so to show that it holds we have just to prove that it is satisfied in 4$%
N$ distinct points as the leading coefficient is zero by the
choice of $\gamma _{0}$ to be a nonzero eigenvalue of $K$. We use the following 4$%
N$ points $\xi _{a}+k_{a}\eta $, for any $a\in \{1,...,N\}$
and $k_{a}\in \left\{ -1,0,1,2\right\} $. Indeed, for $\lambda =\xi
_{a}-\eta $ it holds:%
\begin{equation}
\alpha (\xi _{a}-\eta )=\beta (\xi _{a}-\eta )=\gamma (\xi _{a}-\eta )=\text{%
det}M_{a}^{(K)}(\xi _{a}-3\eta )=0,
\end{equation}%
from which the functional equation is satisfied for any $a\in \{1,...,%
N\}$ and in the remaining 3$N$ points the functional
equation reduces to the 3$N$ equations \rf{Ch-1-Q}, \rf{Ch-2-Q} and \rf{Ch-3-Q} which, thanks to the fusion equations,
satisfied by the transfer matrix eigenvalues, are all equivalent to the
discrete characterization $\left( \ref{Discrete-Ch-Q}\right) $ so that our
statement holds.

Finally, let us show that the SoV characterization of the transfer matrix
eigenvector associated to the eigenvalue $t_{1}(\lambda )$ is equivalent to
the one presented in this theorem. We have just to remark that renormalizing
the eigenvector $|t\rangle $ multiplying it by the non-zero product of the $%
\varphi _{t}^{2}(\xi _{a})$ over all the $a\in \{1,...,N\}$ we get: 
\begin{equation}
\prod_{a=1}^{N}\varphi _{t}^{2}(\xi _{a})\prod_{a=1}^{N%
}t_{1}^{h_{a}}(\xi _{a})\overset{\left( \ref{Discrete-Ch-Q}\right) }{=}%
\prod_{a=1}^{N}\gamma ^{h_{a}}(\xi _{a})\varphi _{t}^{h_{a}}(\xi
_{a}-\eta )\varphi _{t}^{2-h_{a}}(\xi _{a}).
\end{equation}
\end{proof}

\subsection{Algebraic Bethe ansatz rewriting of transfer matrix eigenvectors}

It is easy to see that the previous SoV representation of the transfer
matrix eigenvectors admit an equivalent rewriting of Algebraic Bethe Ansatz
type. For this let us first remark that there exists one eigenvector of the transfer matrix $%
T_{1}^{\left( K\right) }(\lambda )$ and $T_{2}^{\left( K\right) }(\lambda )$
which corresponds to the constant solution of the quantum spectral curve
equation.

\begin{lemma}
Let $K$ a generic $3\times 3$ matrix and let us denote with $K_{J}$ its
Jordan form:%
\begin{equation}
K=W_{K}K_{J}W_{K}^{-1}\text{ \ with \ }K_{J}=\left( 
\begin{array}{ccc}
\mathsf{k}_{0} & \mathsf{y}_{1} & 0 \\ 
0 & \mathsf{k}_{1} & \mathsf{y}_{2} \\ 
0 & 0 & \mathsf{k}_{2}%
\end{array}%
\right) ,
\end{equation}%
where we assume that $\mathsf{k}_{0}\neq 0$, then:%
\begin{equation}
|t_{0}\rangle =\Gamma _{W}\bigotimes_{a=1}^{N}\left( 
\begin{array}{c}
1 \\ 
0 \\ 
0%
\end{array}%
\right) _{a}\text{ \ with \ }\Gamma _{W}=\bigotimes_{a=1}^{N}W_{K,a}
\end{equation}%
is a common eigenstate of the transfer matrices $T_{1}^{\left( K\right)
}(\lambda )$ and $T_{2}^{\left( K\right) }(\lambda )$:%
\begin{align}
T_{1}^{\left( K\right) }(\lambda )|t_{0}\rangle =|t_{0}\rangle
t_{1,0}(\lambda )\text{ \ with }&t_{1,0}(\lambda )=\mathsf{k}_{0}\prod_{a=1}^{%
N}(\lambda -\xi _{a}+\eta )+(\mathsf{k}_{1}+\mathsf{k}%
_{2})\prod_{a=1}^{N}(\lambda -\xi _{a}), \\
T_{2}^{\left( K\right) }(\lambda )|t_{0}\rangle =|t_{0}\rangle
t_{2,0}(\lambda )\text{ \ with }&\nonumber \\
t_{2,0}(\lambda )=\prod_{a=1}^{N%
}(\lambda -\xi _{a})& ( \mathsf{k}_{1}\mathsf{k}_{2}\prod_{a=1}^{\mathsf{N%
}}(\lambda -\xi _{a}+\eta )+(\mathsf{k}_{1}\mathsf{k}_{0}+\mathsf{k}_{2}%
\mathsf{k}_{0})\prod_{a=1}^{N}(\lambda -\xi _{a})) ,
\end{align}%
and $t_{1,0}(\lambda )$ and $t_{2,0}(\lambda )$ satisfy the quantum spectral
curve with constant $\varphi _{t}(\lambda )$:%
\begin{equation}
\alpha (\lambda )-\beta (\lambda )t_{1,0}(\lambda -2\eta )+\gamma (\lambda
)t_{2,0}(\lambda -2\eta )-q\text{-det}M_{a}^{(K)}(\lambda -2\eta )=0,
\end{equation}%
and with $\gamma _{0}=\mathsf{k}_{0}.$
\end{lemma}

\begin{proof}
The proof of the statement is done proving that for the transfer matrices $%
T_{1}^{\left( K_{J}\right) }(\lambda )$ and $T_{2}^{\left( K_{J}\right)
}(\lambda )$ the vector:%
\begin{equation}
|0\rangle =\bigotimes_{a=1}^{N}\left( 
\begin{array}{c}
1 \\ 
0 \\ 
0%
\end{array}%
\right)
\end{equation}%
is eigenvector with eigenvalues $t_{1,0}(\lambda )$ and $t_{2,0}(\lambda )$,
respectively. The proof of this statement is standard and done proving that
it holds:%
\begin{equation}
A_{i}^{(I)}(\lambda )|0\rangle =|0\rangle \prod_{a=1}^{N}(\lambda
-\xi _{a}+\delta _{i,1}\eta ),\text{ \ }C_{i}^{(I)}(\lambda )|0\rangle =0,%
\text{ \ }i\in \{1,2,3\}.
\end{equation}%
It is interesting to remark that it is simple to verify by direct
computation that the $t_{1,0}(\lambda )$ and $t_{2,0}(\lambda )$ satisfies
the fusion equations $(\ref{scalar-fusion-1})$ and that it holds:%
\begin{eqnarray}
t_{1,0} &\equiv &\lim_{\lambda \rightarrow \infty }\lambda ^{-N%
}t_{1,0}(\lambda )=\text{tr\thinspace }K, \\
t_{2,0} &\equiv &\lim_{\lambda \rightarrow \infty }\lambda ^{-2N%
}t_{2,0}(\lambda )=\frac{\left( \text{tr}K\right) ^{2}-\text{tr}K^{2}}{2},
\end{eqnarray}%
so that $t_{1,0}(\lambda )$ satisfies the SoV characterization of the
eigenvalues of $T_{1}^{\left( K_{J}\right) }(\lambda )$. Observing now that
it holds:%
\begin{equation}
t_{1,0}(\xi _{a})=\gamma (\xi _{a})\text{ \ \ for any }a\in \{1,...,\mathsf{N%
}\}
\end{equation}%
it follows that the associated $\varphi _{t}(\lambda )$ satisfies the
equations:%
\begin{equation}
\varphi _{t}(\xi _{a})=\varphi _{t}(\xi _{a}-\eta )\text{ \ \ for any }a\in
\{1,...,N\}
\end{equation}%
and so $\varphi _{t}(\lambda )$ is constant. Indeed, let us define:%
\begin{equation}
\bar{\varphi}_{t}(\lambda )=\varphi _{t}(\lambda )-\varphi _{t}(\lambda
-\eta )
\end{equation}%
this is a degree $N-1$ polynomial in $\lambda $ which is zero in $%
N$ different points so that it is identically zero.
\end{proof}

Let us now denote by $\mathbb{B}^{(K)}\left( \lambda \right) $ the one
parameter family of commuting operators characterized by:%
\begin{equation}
\langle h_{1},...,h_{N}|\mathbb{B}^{(K)}\left( \lambda \right) =%
\text{ }b_{h_{1},...,h_{N}}(\lambda )\langle h_{1},...,h_{N%
}|,
\end{equation}%
where we have defined:%
\begin{equation}
b_{h_{1},...,h_{N}}(\lambda )=\prod_{a=1}^{N}(\lambda -\xi
_{a})^{2-h_{a}}(\lambda -\xi _{a}+\eta )^{h_{a}},
\end{equation}%
then the following corollary holds.

\begin{corollary}
Let $t_{1}(\lambda )\in \Sigma _{T_{1}}$ then the associated eigenvector $%
|t\rangle $ admits the following Algebraic Bethe Ansatz type formulation:%
\begin{equation}
|t\rangle =\prod_{a=1}^{\mathsf{M}}\mathbb{B}^{(K)}(\lambda
_{a})|t_{0}\rangle \text{\ \ with }\mathsf{M}\leq N\text{ and }%
\lambda _{a}\neq \xi _{n}\text{ }\forall (a,n)\in \{1,...,\mathsf{M}\}\times
\{1,...,N\},
\end{equation}%
where the $\lambda _{a}$ are the roots of the polynomial $\varphi
_{t}(\lambda )$ which satisfies with $t_{1}(\lambda )$ the third order
Baxter's like functional equation.
\end{corollary}

\begin{proof}
Let us observe that we have:%
\begin{eqnarray}
\langle h_{1},...,h_{N}|\prod_{a=1}^{\mathsf{M}}\mathbb{B}%
^{(K)}(\lambda _{a})|t_{0}\rangle &=&\prod_{j=1}^{\mathsf{M}}b_{h_{1},...,h_{%
N}}(\lambda _{j})\text{ }\langle h_{1},...,h_{N%
}|t_{0}\rangle \\
&=&\prod_{j=1}^{\mathsf{M}}\prod_{a=1}^{N}(\lambda _{j}-\xi
_{a})^{2-h_{a}}(\lambda _{j}-\xi _{a}+\eta )^{h_{a}}\prod_{a=1}^{N%
}\gamma ^{h_{a}}(\xi _{a}) \\
&=&\prod_{a=1}^{N}\gamma ^{h_{a}}(\xi _{a})\varphi _{t}(\xi
_{a})^{2-h_{a}}\varphi _{t}(\xi _{a}-\eta )^{h_{a}},
\end{eqnarray}%
where we have used that:%
\begin{equation}
\langle h_{1},...,h_{N}|t_{0}\rangle =\prod_{a=1}^{N%
}\gamma ^{h_{a}}(\xi _{a})
\end{equation}%
which coincides with the last SoV characterization of the same transfer
matrix eigenstate.
\end{proof}

\section{Conclusions and perspectives}

We have shown that the construction of separate basis for transfer matrices of quantum integrable lattice models can be achieved using the new paradigm given by \eqref{sb1}. It sheds a completely new light on the notion of quantum integrability itself. Indeed, as soon as the transfer matrix possess the $w$-simplicity property, which seems to be a quite widely shared property, the construction of a separate basis along these lines can be performed. What remains however non trivial is the possibility to get close relations that characterize the transfer matrix eigenvalues in the form of a quantum spectral curve equation of minimal degree. In all the examples we have been able to look at so far, the needed information is provided by the fusion relations among the tower of transfer matrices. These fusion relations are themselves  direct consequences of the structure of the $R$-matrix governing the Yang-Baxter algebra and of the representation theory for it.  Looking at this feature from a different point of view should be fruitful for the future investigations of our new separation of variable method. What seems in turn to be of fundamental importance for obtaining constraints that characterize the full spectrum of the transfer matrices is to unravel the structure of the commutative (and associative) algebra of conserved charges, namely, to get the explicit structure constants of such an algebra. The fusion relations are just one way to get this kind of information but we believe that there should be more algebraic ways to consider this problem; for example starting directly from general properties of the associated quantum groups. It should be interesting for example to test these ideas in the case of the Gaudin models and to understand the relation with the geometrical construction of eigenvectors presented in \cite{MukTV09}, see also \cite{MukTV09a,MukTV09c,MukTV11}. It should also be interesting to make contact with the approach of \cite{MauO12}. Another important direction is to use the concepts and ideas developed in the present paper for quantum integrable field theories. Let us also make comments in a quite different direction interesting for statistical physics.  The separate basis \eqref{sb1} is in our opinion of great interest for the considerations of Quench dynamics, so-called generalized Gibbs ensembles and generalized hydrodynamics equations in integrable models, see e.g. \cite{WouDNBFRC14,IliDNWCEP15,BerCDNF16,CasADY16,DoyY17}, and references therein. The first remark is that the existence of a basis \eqref{sb1} should give the right criterion to determine which set of conserved charges should be considered. They should be related to the one necessary  to construct a basis like  \eqref{sb1} having the minimal dimension for each $h_j$, namely associated to the prime decomposition of the dimension $d$ of the space of states of the model one consider.  In the forthcoming articles we will first explain how our new scheme applies to important classes of integrable quantum models. These are the models associated to fundamental representations of the Yang-Baxter and reflection algebra for the Yangian $Y(gl_{n})$, the quantum group $U_{q}(gl_{n})$ and the t-J model. We have also applied our method successfully to models associated to cyclic and higher spin representations. The next step for these models, beyond the complete resolution of their spectrum concerns their dynamical properties, namely the computation of their form factors and correlation functions. This amounts first to compute scalar products of states within this new method. We plan to address these important issues in the near future.

%%%%%%%%%%%%%%%%%%%%%%%%%%%%%%%%%%%%%%%%%%%%%%%%%%%%%%%%%%%%%%%%%%%%%%%%%%%%%%%%%%%%%%%%%%%%%%%%%%%%%%
\section*{Acknowledgements}
J. M. M. and G. N.  are supported by CNRS and ENS de Lyon.

%%%%%%%%%%%%%%%%%%%%%%%%%%%%%%%%%%%%%%%%%%%%%%%%%%%%%%%%%%%%%%%%%%%%%%%%%%%%%%%%%%%%%%%%%%%%%%%%%%%%%%%%%%%%%%%%%%

%--------------------------------------------------------

\appendix
\section{Comparison with Sklyanin SoV for the $Y(gl_3)$ case}

In this appendix we compare our construction of the SoV basis with the one
proposed by Sklyanin for the quantum model associated to the fundamental representation of  $Y(gl_3)$. As this analysis  does not play any
direct role in our SoV construction and associated results, we restrict here
to chains with a small finite number of sites. More
precisely, the claims in the following on the $\mathcal{B}^{(K)}$%
-Sklyanin and $\mathcal{A}^{(K)}$-Sklyanin operators are based on direct
proofs, i.e. the statements are proven valid for general values of the
parameters (inhomogeneities and $K$-matrix entries) by direct verifications by
using Mathematica for chains up to 3 sites. It is natural to believe that
these claims should be true for chains of any size and this can be seen as
conjectures whose mathematical proof can be interesting and to which we will come back in the future.

Let us recall that Sklyanin has introduced the following two operators:%
\begin{eqnarray}
\mathcal{B}^{(K)}(\lambda ) &=&B_{3}^{(K)}(\lambda )\mathsf{C}%
_{2}^{(K)}(\lambda -\eta )-B_{2}^{(K)}(\lambda )\mathsf{C}_{3}^{(K)}(\lambda
-\eta ), \\
\mathcal{A}^{(K)}(\lambda ) &=&-\left[ B_{3}^{(K)}(\lambda -\eta )\right]
^{-1}\mathsf{C}_{2}^{(K)}(\lambda -\eta ),
\end{eqnarray}%
where we use the notations:%
\begin{equation}
U_{a}^{(K)}(\lambda )=\left( 
\begin{array}{ccc}
\mathsf{A}_{1}^{(K)}(\lambda ) & \mathsf{B}_{1}^{(K)}(\lambda ) & \mathsf{B}%
_{2}^{(K)}(\lambda ) \\ 
\mathsf{C}_{1}^{(K)}(\lambda ) & \mathsf{A}_{2}^{(K)}(\lambda ) & \mathsf{B}%
_{3}^{(K)}(\lambda ) \\ 
\mathsf{C}_{2}^{(K)}(\lambda ) & \mathsf{C}_{3}^{(K)}(\lambda ) & \mathsf{A}%
_{3}^{(K)}(\lambda )%
\end{array}%
\right) _{a},
\end{equation}%
which respectively should generate the separate variables for the transfer
matrices and the shift operators on the separate variables spectrum.
Moreover, the separate relations for the spectral problem of the transfer
matrix should be the quantum spectral curve analog computed along the
separate variables spectrum.

Indeed, Sklyanin has proven the following identities: 
\begin{equation}
(\lambda -\mu )\mathcal{A}^{(K)}(\lambda )\mathcal{B}^{(K)}(\mu )=(\lambda
-\mu -\eta )\mathcal{B}^{(K)}(\mu )\mathcal{A}^{(K)}(\lambda )+\mathcal{B}%
^{(K)}(\lambda )\Xi _{1}^{(K)}(\lambda ,\mu )  \label{Skly-shift}
\end{equation}%
for the shift operator and%
\begin{align}
& \mathcal{A}^{(K)}(\lambda )\mathcal{A}^{(K)}(\lambda -\eta )\mathcal{A}%
^{(K)}(\lambda -2\eta )-\mathcal{A}^{(K)}(\lambda )\mathcal{A}^{(K)}(\lambda
-\eta )T_{1}^{(K)}(\lambda -2\eta )+  \notag \\
& \mathcal{A}^{(K)}(\lambda )T_{2}^{(K)}(\lambda -2\eta )-q\text{-det}%
M^{(K)}(\lambda -2\eta )\left. =\right. \mathcal{B}^{(K)}(\lambda )\Xi
_{2}^{(K)}(\lambda )  \label{Skly-Q-Sp-Curve}
\end{align}%
for the quantum spectral curve, the operators $\Xi _{1}^{(K)}(\lambda ,\mu )$
and $\Xi _{1}^{(K)}(\lambda ,\mu )$ have the following explicit formulae:%
\begin{align}
\Xi _{1}^{(K)}(\lambda ,\mu )& =\eta \mathcal{A}^{(K)}(\mu )\left[
B_{3}^{(K)}(\lambda )B_{3}^{(K)}(\lambda -\eta )\right] ^{-1}B_{3}^{(K)}(\mu
-\eta )B_{3}^{(K)}(\mu ), \\
\Xi _{2}^{(K)}(\lambda )& =\left[ B_{3}^{(K)}(\lambda )B_{3}^{(K)}(\lambda
-\eta )B_{3}^{(K)}(\lambda -2\eta )\right] ^{-1} \\
& \times \left[ C_{1}^{(K)}(\lambda -2\eta )\mathsf{C}_{3}^{(K)}(\lambda
-2\eta )-B_{3}^{(K)}(\lambda -2\eta )\mathsf{B}_{1}^{(K)}(\lambda -2\eta )%
\right] .
\end{align}%
Here, the main problem is that independently from the choice of the matrix $K $ our observation is that some zeros of $\mathcal{B}^{(K)}(\lambda )$ coincide with 
singularities of the operators $\Xi _{1}^{(K)}(\lambda ,\mu )$ and $\Xi
_{2}^{(K)}(\lambda )$, so that the r.h.s of the equations $(\ref{Skly-shift}%
) $ and $(\ref{Skly-Q-Sp-Curve})$ are nonzero in some points of the spectrum of the
zeros of $\mathcal{B}^{(K)}(\lambda )$. Hence, they do not imply the desired
shift and spectral curve equations. So that for the representation under consideration the operator $\mathcal{A}^{(K)}$ of Sklyanin does not seem to produce the right shift operator on the $\mathcal{B}^{(K)}$ spectrum. Recently however, there appeared an interesting article \cite{DerV18} using the same structure for the $sl_3$ non-compact case. We do not know if a problem similar to the one we discuss here could also be present there, the question being if  in the non-compact case the spectrum of the denominators in the above equations have points in common with the $\mathcal{B}^{(K)}$ spectrum. It would be interesting to clarify this point and also to investigate the method we propose here in this non-compact situation. In the fundamental representation we consider here the following statements holds for chain up to $N=3$ sites, they have been verified by symbolic computations using Mathematica:

\begin{property}\label{Property1}
Let us consider a chain with $N\leq 3$ sites then the following statements holds. Let us write explicitly the matrix:%
\begin{equation}
K=\left( 
\begin{array}{ccc}
k_{1} & k_{2} & k_{3} \\ 
k_{4} & k_{5} & k_{6} \\ 
k_{7} & k_{8} & k_{9}%
\end{array}%
\right) ,
\end{equation}%
then under the condition:%
\begin{equation}\label{cond-B-Skly-diag}
\kappa _{K}\equiv
k_{1}k_{3}k_{6}-k_{3}k_{5}k_{6}-k_{3}^{2}k_{4}+k_{2}k_{6}^{2}\neq 0
\end{equation}
and the inhomogeneity condition \rf{Inhomog-cond-gl3}, the one-parameter family $\mathcal{B}^{(K)}(\lambda )$ of commuting operators is
diagonalizable with simple spectrum. The eigenvalues of $\mathcal{B}%
^{(K)}(\lambda )$ admit the following representations:%
\begin{equation}
b_{h_{1},...,h_{N}}^{(K)}(\lambda )=\kappa _{K}b_{0}(\lambda
)\prod_{a=1}^{N}(\lambda -\xi _{a})^{2-h_{a}}(\lambda -\xi
_{a}+\eta )^{h_{a}},  \label{Skly-B-eigenV}
\end{equation}%
where for any $i\in \left\{ 1,...,N\right\} $ and $h_{i}\in \left\{
0,1,2\right\} $ and we have defined:%
\begin{equation}
b_{0}(\lambda )=\prod_{a=1}^{N}(\lambda -\xi _{a}-\eta ).
\end{equation}%
The associated covector eigenbasis of $\mathcal{B}^{(K)}(\lambda )$:%
\begin{equation*}
\langle h_{1},...,h_{N}|\mathcal{B}^{(K)}(\lambda )=b_{h_{1},...,h_{%
N}}^{(K)}(\lambda )\langle h_{1},...,h_{N}|,
\end{equation*}%
coincides with our SoV basis once we fix the form of our basis by imposing:%
\begin{equation}
\langle h_{1},...,h_{N}|= \bra{L_{1}}\prod_{n=1}^{N}T_{2}^{\left(
K\right) \delta _{h_{n},0}}(\xi _{n}-2\eta )T_{1}^{\left( K\right) \delta
_{h_{n},2}}(\xi _{n})\text{\ }\forall \text{ }h_{n}\in \{0,1,2\},
\end{equation}%
with%
\begin{equation}
\langle h_{1}=1,...,h_{N}=1|=\bra{L_{1}}\equiv \bigotimes_{a=1}^{\mathsf{N%
}}(-k_{6},k_{3},0)_{a}.  \label{Left-fix-basis}
\end{equation}%
It is interesting to note that it holds:%
\begin{align}
\langle h_{1}\left. =\right. 0,...,h_{N}\left. =\right. 0|& =2\eta
^{2N}\bigotimes_{a=1}^{N%
}(k_{6}(k_{3}k_{7}+k_{6}k_{8})-k_{9}(k_{3}k_{4}+k_{5}k_{6}),k_{9}(k_{1}k_{3}+k_{2}k_{6}) \notag\\
&-k_{3}(k_{3}k_{7}+k_{6}k_{8}),-\kappa _{K})_{a} \label{EigenCovector-0}\\
& =2\eta ^{2N}\langle h_{1}\left. =\right. 1,...,h_{N%
}\left. =\right. 1|\bigotimes_{a=1}^{N}\tilde{K}_{a},\text{ \ with }%
\tilde{K}\text{ the adjoint of }K\text{.} \\
\langle h_{1}\left. =\right. 2,...,h_{N}\left. =\right. 2|& =\eta ^{%
N}\bigotimes_{a=1}^{N%
}(k_{1}k_{6}-k_{3}k_{4},k_{2}k_{6}-k_{3}k_{5},0)_{a}=\eta \langle
h_{1}\left. =\right. 1,...,h_{N}\left. =\right. 1|\bigotimes_{a=1}^{%
N}K_{a}.  \label{EigenCovector-2}
\end{align}%
Moreover, $\mathcal{B}^{(K)}(\lambda )$ has also the following two
eigenvectors with tensor product form:%
\begin{align}
|h_{1}& =0,...,h_{N}=0\rangle =\bigotimes_{a=1}^{N}\left(
0,0,1\right) _{a}^{t_{a}},  \label{EigenVector-0} \\
|h_{1}& =2,...,h_{N}=2\rangle =\eta ^{N}\bigotimes_{a=1}^{%
N}\left( k_{3},k_{6},k_{9}\right) _{a}^{t_{a}}=\eta ^{N%
}\bigotimes_{a=1}^{N}K_{a}R_{0}.  \label{EigenVector-2}
\end{align}%
Finally, by using the operator family $\mathcal{A}^{(K)}(\lambda )$ and $%
\mathcal{D}^{(K)}(\lambda )\equiv \left[ \mathcal{A}^{(K)}(\lambda )\right]
^{-1}$we get:%
\begin{align}
\lim_{\lambda \rightarrow \xi _{n}}\langle h_{1},...,h_{n} &=0,...,h_{%
N}|\mathcal{A}^{(K)}(\lambda )=c_{h_{1},...,h_{N}}^{(0)}%
\text{det}K\,\langle h_{1},...,h_{n}=1,...,h_{N}|, \\
\lim_{\lambda \rightarrow \xi _{n}}\langle h_{1},...,h_{n} &=2,...,h_{%
N}|\mathcal{D}^{(K)}(\lambda )=c_{h_{1},...,h_{N%
}}^{(2)}\langle h_{1},...,h_{n}=1,...,h_{N}|,
\end{align}%
for any $n\in \{1,...,N\}$, $h_{j}\in \{0,1,2\}$ for $j\in \{1,...,%
N\}\backslash n$ and with $c_{h_{1},...,h_{N}}^{(i)}$ some
nonzero finite constants. However, it holds:%
\begin{align}
\lim_{\lambda \rightarrow \xi _{n},\xi _{n}\pm \eta }\langle h_{1},...,h_{n}
&=1,...,h_{N}|\mathcal{A}^{(K)}(\lambda )\neq \bar{c}%
_{h_{1},...,h_{N}}^{(0)}\langle h_{1},...,h_{n}=2,...,h_{N%
}|, \\
\lim_{\lambda \rightarrow \xi _{n},\xi _{n}\pm \eta }\langle h_{1},...,h_{n}
&=1,...,h_{N}|\mathcal{D}^{(K)}(\lambda )\neq \bar{c}%
_{h_{1},...,h_{N}}^{(0)}\langle h_{1},...,h_{n}=0,...,h_{N%
}|,
\end{align}%
for any $n\in \{1,...,N\}$, $h_{j}\in \{0,1,2\}$ for $j\in \{1,...,%
N\}\backslash n$ and any $\bar{c}_{h_{1},...,h_{N}}^{(i)}$
nonzero finite constants.
\end{property}

The above claim about the diagonalizability of the $\mathcal{B}%
^{(K)}(\lambda )$ operator are restricted to the case $\kappa _{K}\neq 0$.

\begin{property}\label{Property2}
Let us consider a chain with $N\leq 3$ sites and let us assume that the inhomogeneity condition  \rf{Inhomog-cond-gl3} holds
and that the matrix $K$ is $w$-simple then, we have the following identification:%
\begin{equation}\label{Coincidence of B}
\mathbb{B}^{(K)}\left( \lambda \right) =\mathcal{B}^{(K)}(\lambda )/(\kappa
_{K}b_{0}(\lambda ))\text{\ \ if \ }\kappa _{K}\neq 0
\end{equation}%
once we have fixed our SoV basis imposing $(\ref{Left-fix-basis})$. While if 
$\kappa _{K}=0$, then for any $\hat{K}$ similar to $K$ such that $\kappa _{%
\hat{K}}\neq 0$, we have the identification:%
\begin{equation}
\mathbb{B}^{(K)}\left( \lambda \right) =\Gamma _{W\hat{W}^{-1}}\mathcal{B}^{(%
\hat{K})}(\lambda )/(\kappa _{\hat{K}}b_{0}(\lambda ))\Gamma _{W\hat{W}%
^{-1}}^{-1}
\end{equation}%
where:%
\begin{equation}
\Gamma _{W\hat{W}^{-1}}=\bigotimes_{a=1}^{N}W_{K,a}\hat{W}_{\hat{K}%
,a}^{-1}
\end{equation}%
once we have fixed our SoV basis imposing:%
\begin{equation}
\langle h_{1}=1,...,h_{N}=1|=\bigotimes_{a=1}^{N}(-\hat{k}%
_{6},\hat{k}_{3},0)_{a}\Gamma _{W\hat{W}^{-1}}^{-1}.
\end{equation}
\end{property}

It is a natural conjecture that the Properties \ref{Property1} and \ref{Property2} holds for chains of any number of sites.

Let us comment that similar statements about the diagonalizability of the Sklyanin's $B$-operator and of the form of its eigenvalues were previously verified in \cite{GroLMS17}, always by symbolic computations in Mathematica for chains of small size, for some special class of twist $K$ matrix satisfying the condition \rf{cond-B-Skly-diag}. In \cite{GroLMS17}, it was moreover done the conjecture that transfer matrix eigenvectors have the usual algebraic Bethe ansatz form in terms of this Sklyanin's $B$-operator. Such conjecture has been recently verified in \cite{LiaS18} algebraically, mainly relying on the use of the Yang-Baxter commutation relations. 
It is then worth mentioning that if the identity \rf{Coincidence of B} is proven to hold for chains of any size then the algebraic Bethe ansatz form of the transfer matrix eigenvectors is derived in just one line proof starting from the SoV representation of these eigenvectors, as described in section 5.4.

%--------------------------------------------------------

%\appendix

%-------------------------------------------------------

%\section*{Appendices}

%\section{}

%\subsection{}

%\bibliographystyle{unsrt}

%\bibliography{biblio}

\end{document}